\documentclass[11pt]{article}
\usepackage{amsmath,amsthm,latexsym,amssymb,amsfonts,epsfig}

\newcommand{\Q}{\sqrt{\det{Q}}}
\newcommand{\Qb}{\sqrt{\det{\overline{Q}}}}

\newcommand{\N}{\vec{N}}
\newcommand{\Nb}{\vec{\overline{N}}}
\newcommand{\ov}[1]{\overline{#1}}
\newcommand{\QN}[2]{\frac{1}{2}\Big(\ov{Q}^{#1#2}+\frac{\ov{N}^{#1}\ov{N}^{#2}}{
\ov{N}^2}\Big)}
\newcommand{\Nv}[1]{\frac{\ov{N}^{#1}}{\ov{N}^2}}
\makeatletter
\@addtoreset{equation}{section}
\makeatother

\newtheorem{Lemma}{Lemma}[section]

\def\be{\begin{equation}}
\def\ee{\end{equation}}
\def\ba{\begin{eqnarray}}
\def\ea{\end{eqnarray}}
\def\Nl{{\mathchoice
{\setbox0=\hbox{$\displaystyle\rm N$}\hbox{\hbox to0pt
{\kern0.4\wd0\vrule height0.9\ht0\hss}\box0}}
{\setbox0=\hbox{$\textstyle\rm N$}\hbox{\hbox to0pt
{\kern0.4\wd0\vrule height0.9\ht0\hss}\box0}}
{\setbox0=\hbox{$\scriptstyle\rm N$}\hbox{\hbox to0pt
{\kern0.4\wd0\vrule height0.9\ht0\hss}\box0}}
{\setbox0=\hbox{$\scriptscriptstyle\rm N$}\hbox{\hbox to0pt
{\kern0.4\wd0\vrule height0.9\ht0\hss}\box0}}}}
\def\Zl{{\mathchoice
{\setbox0=\hbox{$\displaystyle\rm Z$}\hbox{\hbox to0pt
{\kern0.4\wd0\vrule height0.9\ht0\hss}\box0}}
{\setbox0=\hbox{$\textstyle\rm Z$}\hbox{\hbox to0pt
{\kern0.4\wd0\vrule height0.9\ht0\hss}\box0}}
{\setbox0=\hbox{$\scriptstyle\rm Z$}\hbox{\hbox to0pt
{\kern0.4\wd0\vrule height0.9\ht0\hss}\box0}}
{\setbox0=\hbox{$\scriptscriptstyle\rm Z$}\hbox{\hbox to0pt
{\kern0.4\wd0\vrule height0.9\ht0\hss}\box0}}}}
\def\Ql{{\mathchoice
{\setbox0=\hbox{$\displaystyle\rm Q$}\hbox{\hbox to0pt
{\kern0.4\wd0\vrule height0.9\ht0\hss}\box0}}
{\setbox0=\hbox{$\textstyle\rm Q$}\hbox{\hbox to0pt
{\kern0.4\wd0\vrule height0.9\ht0\hss}\box0}}
{\setbox0=\hbox{$\scriptstyle\rm Q$}\hbox{\hbox to0pt
{\kern0.4\wd0\vrule height0.9\ht0\hss}\box0}}
{\setbox0=\hbox{$\scriptscriptstyle\rm Q$}\hbox{\hbox to0pt
{\kern0.4\wd0\vrule height0.9\ht0\hss}\box0}}}}
\def\Rl{{\mathchoice
{\setbox0=\hbox{$\displaystyle\rm R$}\hbox{\hbox to0pt
{\kern0.4\wd0\vrule height0.9\ht0\hss}\box0}}
{\setbox0=\hbox{$\textstyle\rm R$}\hbox{\hbox to0pt
{\kern0.4\wd0\vrule height0.9\ht0\hss}\box0}}
{\setbox0=\hbox{$\scriptstyle\rm R$}\hbox{\hbox to0pt
{\kern0.4\wd0\vrule height0.9\ht0\hss}\box0}}
{\setbox0=\hbox{$\scriptscriptstyle\rm R$}\hbox{\hbox to0pt
{\kern0.4\wd0\vrule height0.9\ht0\hss}\box0}}}}
\def\Cl{{\mathchoice
{\setbox0=\hbox{$\displaystyle\rm C$}\hbox{\hbox to0pt
{\kern0.4\wd0\vrule height0.9\ht0\hss}\box0}}
{\setbox0=\hbox{$\textstyle\rm C$}\hbox{\hbox to0pt
{\kern0.4\wd0\vrule height0.9\ht0\hss}\box0}}
{\setbox0=\hbox{$\scriptstyle\rm C$}\hbox{\hbox to0pt
{\kern0.4\wd0\vrule height0.9\ht0\hss}\box0}}
{\setbox0=\hbox{$\scriptscriptstyle\rm C$}\hbox{\hbox to0pt
{\kern0.4\wd0\vrule height0.9\ht0\hss}\box0}}}}
\def\Hl{{\mathchoice
{\setbox0=\hbox{$\displaystyle\rm H$}\hbox{\hbox to0pt
{\kern0.4\wd0\vrule height0.9\ht0\hss}\box0}}
{\setbox0=\hbox{$\textstyle\rm H$}\hbox{\hbox to0pt
{\kern0.4\wd0\vrule height0.9\ht0\hss}\box0}}
{\setbox0=\hbox{$\scriptstyle\rm H$}\hbox{\hbox to0pt
{\kern0.4\wd0\vrule height0.9\ht0\hss}\box0}}
{\setbox0=\hbox{$\scriptscriptstyle\rm H$}\hbox{\hbox to0pt
{\kern0.4\wd0\vrule height0.9\ht0\hss}\box0}}}}
\def\Ol{{\mathchoice
{\setbox0=\hbox{$\displaystyle\rm O$}\hbox{\hbox to0pt
{\kern0.4\wd0\vrule height0.9\ht0\hss}\box0}}
{\setbox0=\hbox{$\textstyle\rm O$}\hbox{\hbox to0pt
{\kern0.4\wd0\vrule height0.9\ht0\hss}\box0}}
{\setbox0=\hbox{$\scriptstyle\rm O$}\hbox{\hbox to0pt
{\kern0.4\wd0\vrule height0.9\ht0\hss}\box0}}
{\setbox0=\hbox{$\scriptscriptstyle\rm O$}\hbox{\hbox to0pt
{\kern0.4\wd0\vrule height0.9\ht0\hss}\box0}}}}
\DeclareMathOperator{\HF}{\boldsymbol{\mathsf {H}}}
\DeclareMathOperator{\LF}{\boldsymbol{\mathsf {L}}}

\title{{\sf Manifestly Gauge-Invariant General Relativistic}\\{\sf  Perturbation Theory
: I. Foundations}}

\author{{\sf K. Giesel$^1$}\thanks{{\sf gieskri@aei.mpg.de}}, {\sf S.
Hofmann$^{2,3}$}\thanks{{\sf stefan@nordita.org}},
{\sf T.
Thiemann$^{1,2}$}\thanks{{\sf
thiemann@aei.mpg.de,tthiemann@perimeterinstitute.ca}},
{\sf O. Winkler$^2$}\thanks{{\sf owinkler@perimeterinstitute.ca}}\\
\\
{\sf $^1$ MPI f. Gravitationsphysik, Albert-Einstein-Institut,} \\
{\sf Am M\"uhlenberg 1, 14476 Potsdam, Germany}\\
\\
{\sf $^2$ Perimeter Institute for Theoretical Physics,}\\
{\sf 31 Caroline Street N, Waterloo, ON N2L 2Y5, Canada}\\
\\
{\sf $^3$ NORDITA,}\\
{\sf Roslagstullsbacken 23, SE--10691 Stockholm, Sweden}}
\date{{\small\sf Preprint AEI-2007-150}}

\begin{document}

\maketitle

\begin{abstract}
{\sf Linear cosmological perturbation theory is pivotal to a
theoretical understanding of current cosmological experimental data
provided e.g. by cosmic microwave anisotropy probes. A key issue in that
theory is to extract the gauge invariant degrees of freedom which
allow unambiguous comparison between theory and experiment.

When one goes beyond first (linear) order, the task of writing the
Einstein equations expanded to n'th order in terms of quantities
that are gauge invariant up to terms of higher orders becomes highly
non-trivial and cumbersome. This fact has prevented progress for
instance on the issue of the stability of linear perturbation theory
and is a subject of current debate in the literature.

In this series of papers we circumvent these difficulties by passing to
a manifestly gauge invariant framework. In other words, we only perturb
gauge invariant, i.e. measurable quantities, rather than gauge variant
ones. Thus, gauge invariance is preserved non perturbatively while we
construct the perturbation theory for the equations of motion for the
gauge invariant observables to all orders.

In this first paper we develop the general framework which is based
on a seminal paper due to Brown and Kucha{\v r} as well as the realtional formalism
due to Rovelli. In the second,
companion, paper we apply our general theory to FRW cosmologies and
derive the deviations from the standard treatment in linear order.
As it turns out, these deviations are negligible in the late
universe, thus our theory is in agreement with the standard treatment. However, the real
strength of our formalism is that it admits a straightforward and
unambiguous, gauge invariant generalisation to higher orders. This
will also allow us to settle the stability issue in a future
publication. }
\end{abstract}

\newpage

\tableofcontents

\newpage

\section{Introduction}
\label{s1}

General relativity is our best theory for gravitational
physics and, sofar, has stood the test of time and experiments. Its
complicated, highly nonlinear equations of motion, however, mean
that the calculation of many gravitational processes of interest has
to rely on the use of approximations. One important class of such
approximations is given by perturbation theory, where, generally
speaking, one perturbs quantities of interest, such as the metric
and matter degrees of freedom around an exact, known solution which,
typically, displays a high degree of symmetry.

It is well-known that perturbation techniques in general relativity
pose challenges above and beyond those typically associated with
them in other areas of physics, such as stability, convergence
issues etc. The reason is that general relativity is a gauge theory,
the gauge group being the diffeomorphism group Diff$(M)$ of the
spacetime manifold $M$. As a result, all metric and matter variables
transform non-trivially under gauge transformations. This creates
the problem of differentiating between (physical) perturbations of a
given variable and the effect of a gauge-transformation on the
latter. The obvious solution to this situation would be to calculate
only with observables and perturb those. It has proved extremely
difficult, however, to find observables in the full theory, with the
exception of a few special situations, such as for asymptotically
flat spacetimes. As a way out of this conundrum, one usually resorts
to calculating in a specific gauge, carefully ensuring that all
calculated quantities are gauge-independent. Alternatively, one
tries to construct quantities that are observables up to a certain
order. In the cosmological standard model this has been successfully
done in linear order and forms an integral part of the modern lore
of cosmology. In fact, there have been attempts to extend this even
to second order and beyond, see, e.g.,
\cite{oliver15,oliver7,oliver8,oliver9}. The sheer complexity of
those calculations, however, shows that there is a natural limit to
how far that approach can be pushed. Furthermore, it is not clear
whether it will succeed for other backgrounds, such as a black hole
spacetime etc.

This clearly makes the search for a more general framework for
perturbation theory of observable quantities highly desirable.
Another motivation comes from the prospects of developing
perturbation methods for non-perturbative quantum gravity
approaches, such as loop quantum gravity \cite{15}. It is clear that the
standard methods mentioned earlier will be extremely difficult, if
not impossible to implement.

This paper, the first in a series dedicated to this challenge, lays the
foundations at the level of the full theory.
Subsequent papers will deal with simplified cases of particular
interest, such as perturbations around an FRW background.

After this brief overview of the motivations behind our paper, let
us now discuss some of these issues in more detail. The crucial
ingredient in our undertaking is the construction of observables for
the full theory. To that end let us first recall the counting of the
true degrees of freedom of general relativity: The temporal-temporal
as well as the temporal-spatial components of the Einstein equations
do not contain temporal derivatives of four metric functions (known
as lapse and shift). Thus, in the Lagrangean picture, these four
sets of equations can be used, in principle, in order to eliminate
the temporal-temporal and temporal-spatial components of the metric
in terms of the spatial-spatial components\footnote{In the
Hamiltonian picture, these equations relate canonical momenta to
canonical configuration coordinates. There are four additional (so
called primary) constraints which impose that the momenta conjugate
to lapse and shift vanish which leaves only two independent momenta.
These eight constraints are of the first class type in Dirac's
terminology \cite{1}.}. In addition, diffeomorphism gauge invariance
displays four additional degrees of freedom as pure
gauge\footnote{In the Hamiltonian picture, the eight constraints
canonically generate gauge transformations which displays eight out
of ten configuration variables as pure gauge.}. This is why general
relativity in vacuum (without matter) has only two true
(configuration) degrees of freedom (gravitons).

In the canonical picture, the ten equations split into four plus six
equations. The four equations are the afore mentioned constraints
which canonically generate spacetime diffeomorphisms, that is, gauge
transformations. The other six equations are canonically generated
by a canonical ``Hamiltonian'' which is actually a linear
combination of these constraints, and thus also generates gauge
transformations and even is constrained to vanish. It is customary
not to call it a Hamiltonian but rather a Hamiltonian constraint.
The interpretation of Einstein's equations as evolution equations is
therefore unconvincing. Instead, the correct interpretation seems to
be that they actually describe the flow of unphysical degrees of
freedom under gauge transformations. Thus we contend that their
primary use is to extract the true degrees of freedom in the way
described below. These true degrees of freedom are gauge invariant
and thus have trivial evolution with respect to the canonical
Hamiltonian (constraint). This is the famous problem of time of
General Relativity \cite{7}: There is no true Hamiltonian, only a
Hamiltonian constraint and the observable quantities do not move
under its flow. Nothing seems to move, everything is frozen, in
obvious contradiction to our experience. This begs, of course, the
question of what determines the time evolution of the true physical
observables.

In \cite{8} a possible answer was proposed. Namely, it was shown
that the problem of time can be resolved without affecting the
interpretation of Einstein's equations as evolution equations by
adding certain matter to the system. The method for doing this is
based on Rovelli's relational formalism \cite{5}, which was recently
extended considerably by Dittrich \cite{6}, as well as on the Brown
-- Kucha{\v r} mechanism \cite{3}. This necessarily uses a canonical
approach. Furthermore, it was shown in \cite{8} that this in one
stroke provides the true degrees of freedom
 and provides us with a true (physical) Hamiltonian
which generates a non trivial evolution of the gauge invariant
degrees of freedom. Remarkably, these evolution equations look very
similar to Einstein's equations for the type of matter considered.
The type of matter originally used in \cite{8} was chosen somewhat
ad hoc and guided more by mathematical convenience rather than
physical arguments \footnote{Also, apart from cosmological settings,
the consequences of the deviations of these evolution equations from
Einstein's equations was not analysed.}. Furthermore, it seems
desirable to find the optimal matter which would minimally affect
the standard interpretation of Einstein's equations as evolution
equations while increasing the number of true degrees of freedom by
four. As it turns out, there is a natural candidate, which we will
use for our purposes: Pressure free dust as introduced in the
seminal paper by Brown and Kucha{\v r} \cite{3} cited before. The
dust particles fill time and space, they are present everywhere and
at every instant of time. They follow geodesics with respect to the
dynamical four metric under consideration. However, they only
interact gravitationally, not with the other matter and not with
itself. The dust serves as a dynamical reference frame solving
Einstein's hole problem \cite{4}. It can be used to build gauge
invariant versions of all the other degrees of freedom.

The dust supplies the physically meaningless
spacetime coordinates with a dynamical field interpretation and thus
solves the
``problem of time'' of General Relativity as outlined above. This is its
only purpose. For every non -- dust variable in the
usual formalism without dust there is unique gauge invariant
substitute in our theory. Once these observables, that is gauge
invariant matter and geometry modes, have been
constructed as complicated aggregates made out of the non gauge
invariant matter, geometry and dust modes, the dust itself completely
disappears from
the screen. The observable matter and geometry modes are now no longer
subject to constraints, rather, the constraints are replaced by
conservation laws of a gauge invariant energy momentum density.
This energy momentum density is the only trace that the dust leaves
on the system, it can be arbitrarily small but must not vanish in order
that the dust fulfills its role as a material reference frame of
``clocks and rods''. The evolution equations of the observables is
generated by a physical Hamiltonian which is simply the spatial integral
of the energy density. These evolution equations, under proper field
identifications, can be mapped exactly to the six of the Einstein
equations
for the unobservable matter and geometry modes without dust, up to
modifications proportional to the energy momentum density. Thus again
the influence of the dust can be tuned away arbitrarily and so it plays
a perfect role as a ``test observer''. It is interesting that in
contrast to \cite{3} the dust must be a ``phantom dust'', for the same
reason that the phantom scalars apperared in \cite{8}: If we would use
usual dust as in \cite{3} then the physical Hamiltonian would come out
negative definite rather than positive definite. Or equivalently,
physical time would run backwards rather than forward. Notice
that general
relativistic energy conditions for the gauge invariant energy momentum
tensor are {\it not} violated because it does not contain the dust
variables and it is the dust free and gauge invariant energy momentum
tensor that the positive physical Hamiltonian generates. Hence, while
the energy conditions for the phantom dust species are violated at the
gauge variant level, at the gauge invariant level there is no problem
because the dust has disappeared. Notice also that even at the gauge
variant level the energy conditions for the total energy momentum tensor
are still satisfied if there is sufficient additional, observable matter
present.

Based on these constructions we will develop a general relativistic
perturbation theory in this series of papers. In the current work we
treat the case of a general background, in the follow-up papers we
discuss special cases of particular interest.\\
\\
The plan of this paper is as follows:

In section \ref{s2} we review the seminal work of Brown and Kucha{\v r}
\cite{3}. We start from their Lagrangian (with opposite sign in
order to get phantom dust) and then perform the Legendre transform.
This leads to second class constraints which were not discussed in
\cite{3} and which we
solve in appendix \ref{sa}. 
After having solved the second class constraints the further
analysis agrees with \cite{3}. The Brown -- Kucha{\v r} mechanism
can now be applied to the dust plus geometry plus other matter
system and enables us to rewrite the four initial value constraints
of General Relativity in an equivalent way such that these
constraints are not only mutually Poisson commuting but also that
the system deparametrises. That is, they can be solved for the four
dust momentum densities, and the Hamiltonian densities to which they
are equal no longer depend on the dust variables.

In section \ref{s3} we pass to the gauge invariant observables and the
physical Hamiltonian.
In situations such as ours where the system deparametrises, the general
framework of \cite{6} drastically simplifies and one readily obtains the
Dirac observables and the physical Hamiltonian. Due to general
properties of the relational approach, the Poisson algebra among the
observables remains simple. More precisely, for every gauge variant non
-- dust variable
we obtain a gauge invariant analog and the gauge variant and gauge
invariant observables satisfy the same Poisson algebra. 
This is also proved for part of the gauge invariance by independent 
methods in appendix \ref{s3.2}. The physical
time evolution of these observables is generated by a unique, positive
Hamiltonian.

In section \ref{s4} we derive the equations of motion generated by the
physical Hamiltonian for the physical configuration and momentum
observables. We also derive the second order in time equations of
motion for the configuration observables. Interestingly, these
equations can be seen of almost precisely the usual form that they
have in the canonical approach \cite{19} if one identifies lapse and
shift with certain functions of the canonical variables. Hence we
obtain a {\it dynamical lapse and shift}. The system of evolution
equations is supplemented by four sets of conservation laws which
follow from the mutual commutativity of the constraints. They play a
role quite similar to the initial value constraints for the system
without dust written in gauge variant variables but now the
constraint functions do not vanish bur rather are constants of the
motion.

In section \ref{s5} we treat the case of asymptotically flat spacetimes
and derive the necessary boundary terms to make the Hamiltonian
functionally differentiable in that case. Not surprisingly, the 
boundary term is just the ADM Hamiltonian. However, while in the usual 
formalism the bulk term is a linear combination of constraints, in our 
formalism the bulk term does not vanish on the constraint surface, it
represents the total dust energy.

In appendix \ref{s6} we perform the inverse Legendre transform from the
physical Hamiltonian to an action. This cannot be done in closed
form, however, we can write the transform in the form of a fix point
equation which can be treated iteratively. The zeroth iteration
precisely becomes the Einstein -- Hilbert action for geometry and
non -- dust matter. Including higher orders generates a more
complicated ``effective'' action which contains arbitrarily high spatial
derivatives of the gauge invariant variables but only first time
derivatives.

In section \ref{s7}
 we perturb the equations of motion about a general
exact solution to first order, both in the first time derivative
order form and the second time derivative order form. Notice that
our perturbations are fully gauge invariant. In appendix \ref{sb} we show
that one can get the second time derivative equation of motion for
the perturbations in two equivalent ways: Perturbing the second time
derivative equations of motion to first order or deriving the second
time order equation from the perturbations to first order of the
first time order equations. This is an important check when one
derives the equations of motion for the perturbations on a general
background and the second avenue is easier at linear order. However, the 
first avenue is more economic at higher orders. In 
appendix
\ref{sc} we show that the equations of motion up to n'th order are
generated by the physical Hamiltonian expanded up to (n+1)th order.
Moreover we show that the invariants expanded to n'th order remain
constants of the motion under the (n+1)th order Hamiltonian up to
terms of at least order n+1. This is important in order to actually
derive the second time derivative equations of motion because we can
drop otherwise complicated expressions.

In section eight we compare our new approach to general-relativistic
perturbation theory with some other approaches that can be found in
the literature.

Finally, in section nine we conclude and discuss the implications
and open problems raised by the present paper.

In appendix \ref{sd} we ask the question whether the qualitative
conclusions of the present paper are generic or whether they are
special for the dust we chose. In order to test this question we
sketch the repetition of the analysis carried out for the phantom
dust for the phantom scalar field of \cite{8}. It seems that
qualitatively not much changes, although the dust comes closer than
the phantom scalar to reproducing Einstein's equations of motion.
This indicates that the Brown -- Kucha{\v r} mechanism generically
leads to equations of motion for gauge invariant observables which
completely resemble the equations of
motion of their gauge variant counter parts. 

Appendix \ref{LPA} contains more details concerning some calculations in
section seven.

Appendix \ref{sg} derives the connection between our manifestly
gauge invariant formalism and a corresponding gauge fixed version of 
it.\\
\\
Finally, our rather involved notation is listed, for the convenience
of the reader, on the next page.

\newpage

{\bf \large Notation}\\
\\
\\
As a rule of thumb, gauge non invariant quantities are denoted by lower
case letters, gauge invariant quantities by capital letters. The only
exceptions from this rule are the dust fields $T,S^j,\rho,W_j$, their
conjugate momenta $P, P_j, I, I^j$ and their associated primary
constraints $Z_j,Z,Z^j$ which however
disappear in the final picture. Partially gauge invariant quantities
(with respect to spatial diffeomorphisms) carry a tilde. Background
quantities carry a bar. Our
signature
convention is that of relativists, that is, mostly plus.\\
\\
\be \nonumber
\begin{array}{cl}
\mbox{symbol} & \mbox{meaning}  \\
 & \\
G_N & \mbox{Newton constant}\\
\kappa=16\pi G_N & \mbox{gravitational coupling constant}\\
\lambda & \mbox{scalar coupling constant}\\
\Lambda & \mbox{cosmological constant}\\
M & \mbox{spacetime manifold}\\
{\cal X} & \mbox{spatial manifold}\\
{\cal T} & \mbox{dust time manifold}\\
{\cal S} & \mbox{dust space manifold}\\
\mu,\nu,\rho,..=0,..,3 & \mbox{tensor indices on }M\\
a,b,c,..=1,2,3 & \mbox{tensor indices on }{\cal X}\\
i,j,k,..=1,2,3 & \mbox{tensor indices on }{\cal S}\\
X^\mu & \mbox{coordinates on }M\\
x^a & \mbox{coordinates on }{\cal X}\\
\sigma^j & \mbox{coordinates on }{\cal S}\\
t & \mbox{foliation parameter}\\
\tau & \mbox{dust time coordinate}\\
Y_t^\mu & \mbox{one parameter family of embeddings }{\cal X} \to M\\
{\cal X}_t=Y_t({\cal X}) & \mbox{leaves of the foliation}\\
g_{\mu\nu} & \mbox{metric on }M\\
q_{ab} & \mbox{(pullback) metric on }{\cal X}\\
\tilde{q}_{ij} & \mbox{(pullback) metric on }{\cal S}\\
Q_{ij} & \mbox{Dirac observable associated to } q_{ab}\\
p^{ab} & \mbox{momentum conjugate to }q_{ab}\\
\tilde{p}^{ij} & \mbox{momentum conjugate to }\tilde{q}_{ij}\\
P^{ij} & \mbox{momentum conjugate to }Q_{ij}\\
\zeta & \mbox{scalar field on }M\\
\xi & \mbox{scalar field on }{\cal X}\\
\tilde{\xi} & \mbox{pullback scalar field on }{\cal S}\\
\Xi & \mbox{Dirac observable associated to }\xi\\
\pi & \mbox{momentum conjugate to }\xi\\
\tilde{\pi} & \mbox{momentum conjugate to }\tilde{\xi}\\
\Pi & \mbox{momentum conjugate to }\Xi\\
v & \mbox{potential of }\zeta,\;\xi,\;\tilde{\xi},\;\Xi\\
T & \mbox{dust time field on }{\cal X}\\
\tilde{T} & \mbox{dust time field on }{\cal S}\\
S^j & \mbox{dust space fields on }{\cal X}\\
\rho & \mbox{dust energy density on }M,\;{\cal X}\\
W_j & \mbox{dust Lagrange multiplier field on }M,{\cal X}\\
U=-dT+W_j dS^j & \mbox{dust deformation covector field on }M\\
J=\det(\partial S/\partial x) & \mbox{dust field spatial density
on }{\cal X}\\
P & \mbox{momentum conjugate to }T\\
\tilde{P} & \mbox{momentum conjugate to }\tilde{T}\\
P_j & \mbox{momentum conjugate to }S^j\\
I & \mbox{momentum conjugate to }\rho\\
I^j & \mbox{momentum conjugate to }W_j\\
Z_j,\;Z,\;Z^j & \mbox{dust primary constraints on }{\cal X}\\
\mu^j,\;\mu,\;\mu_j & \mbox{dust primary constraint Lagrange
multipliers on }{\cal X}
\end{array}
\ee

\newpage

\be \nonumber\\
\begin{array}{cl}
\varphi & \mbox{diffeomorphism of } {\cal X}\\
n^\mu & \mbox{unit normal of spacelike hypersurface on } M\\
n & \mbox{coordinate lapse function on }{\cal X}\\
n^a & \mbox{coordinate shift function on }{\cal X}\\
p & \mbox{momentum conjugate to } n\\
p_a & \mbox{momentum conjugate to }n^a\\
z,\;z_a & \mbox{primary constraint for lapse, shift}\\
\nu,\;\nu^a & \mbox{lapse and shift primary constraint Lagrange
multipliers}\\
\phi,\psi, B, E & \mbox{MFB scalars on }{\cal X},\;{\cal S}\\
S_a, F_a & \mbox{MFB transversal vectors on }{\cal X}\\
S_j, F_j & \mbox{MFB transversal vectors on }{\cal S}\\
h_{ab} & \mbox{MFB transverse tracefree tensor on }{\cal X}\\
h_{jk} & \mbox{MFB transverse tracefree tensor on }{\cal S}\\
\Phi,\Psi & \mbox{linear gauge invariant completions of }\phi,\psi\\
V_a & \mbox{linear gauge invariant completions of }F_a\\
V_j & \mbox{linear gauge invariant completions of }F_j\\
c^{\rm tot}_a & \mbox{total spatial diffeomorphism constraint on }{\cal X}\\
c^{\rm tot}_j=S^a_j c^{\rm tot}_a & \mbox{total spatial diffeomorphism
constraint on }{\cal X}\\
c^{\rm tot} & \mbox{total Hamiltonian constraint on }{\cal X}\\
c_a & \mbox{non -- dust contribution to spatial diffeomorphism
constraint on }{\cal X}\\
c_j=S^a_j c_a & \mbox{non -- dust contribution to spatial diffeomorphism
constraint on }{\cal X}\\
\tilde{c}_j & \mbox{non -- dust contribution to spatial diffeomorphism
constraint on }{\cal S}\\
C_j\not=\tilde{c}_j & \mbox{momentum density: Dirac observable
associated to
}\tilde{c}_j\\
c & \mbox{non -- dust contribution to Hamiltonian
constraint on }{\cal X}\\
\tilde{c} & \mbox{non -- dust contribution to Hamiltonian
constraint on }{\cal S}\\
C\not=\tilde{c} & \mbox{Dirac observable associated to }\tilde{c}\\
h & \mbox{energy density on }{\cal X}\\
\tilde{h} & \mbox{energy density on }{\cal S}\\
H=\tilde{h} & \mbox{energy density: Dirac observable associated to
}\tilde{h}\\
h_j=c^{\rm tot}_j-P_j & \mbox{auxiliary density on }{\cal X}\\
\epsilon & \mbox{numerical energy density on }{\cal S}\\
\epsilon_j & \mbox{numerical momentum density on }{\cal S}\\
\HF=\int_{{\cal S}}\;d^3\sigma\; H & \mbox{physical Hamiltonian,
energy}\\
L & \mbox{Lagrange density associated to }H\\
\LF=\int_{{\cal S}}\;d^3\sigma\; L & \mbox{physical Lagrangian}\\
V_{jk} & \mbox{velocity associated to }Q_{jk}\\
\Upsilon & \mbox{velocity associated to }\Xi\\
N=C/H & \mbox{dynamical lapse function on }{\cal S}\\
N_j=-C_j/H & \mbox{dynamical shift function on }{\cal S}\\
N^j=Q^{jk} N_k & \mbox{dynamical shift function on }{\cal S}\\
\nabla_\mu & g_{\mu\nu} \mbox{ compatible covariant differential on } M\\
D_a & q_{ab} \mbox{ compatible covariant differential on } {\cal X}\\
\tilde{D}_j & \tilde{q}_{jk} \mbox{ compatible
covariant differential on } {\cal S}\\
D_j & Q_{jk} \mbox{ compatible covariant differential on } {\cal S}\\
\bar{Q}_{jk} &\mbox{ background spatial metric}\\
\bar{P}^{jk} &\mbox{ background momentum conjugate to }\bar{Q}_{jk}\\
\bar{\Xi}& \mbox{ background scalar field}\\
\bar{\Pi} & \mbox{ background momentum conjugate to }\bar{\Xi}\\
\bar{\rho}=\frac{1}{2\lambda}[\dot{\bar{\Xi}}^2+v(\bar{\Xi})] & \mbox{
background scalar energy density}\\
\bar{p}=\frac{1}{2\lambda}[\dot{\bar{\Xi}}^2-v(\bar{\Xi})] & \mbox{
background scalar pressure}
\end{array}
\ee

\newpage

\be \nonumber\\
\begin{array}{cl}
G_{jkmn}=Q_{j(m}Q_{n)k}-\frac{1}{2}Q_{jk} Q_{mn} & \mbox{ physical DeWitt
bimetric}\\
{[}G^{-1}]^{jkmn}=Q^{j(m}Q^{n)k}-Q^{jk} Q^{mn} & \mbox{
inverse physical DeWitt bimetric}\\
\bar{G}_{jkmn}=\delta_{j(m}\delta_{n)k}-\frac{1}{2}\delta_{jk} \delta_{mn} &
\mbox{
flat background DeWitt
bimetric}\\
{[}\bar{G}^{-1}]^{jkmn}=\delta^{j(m}\delta^{n)k}-\delta^{jk}
\delta^{mn}] &
\mbox{
inverse flat background DeWitt bimetric}
\end{array}
\ee

\newpage

\section{The Brown -- Kucha{\v r} formalism}
\label{s2}
In this section we review those elements of the Brown -- Kucha{\v r}
formalism~\cite{3} that are most relevant to us.
Furthermore, we present an explicit constraint analysis for the system where gravity is coupled to a generic scalar field and the Brown -- Kucha\v{r} dust, based on
a canonical analysis using the full arsenal of Dirac's algorithm for constrained
Hamiltonian systems.

For concreteness, we employ dust to deparametrise a system consisting of a
generic scalar field $\zeta$ on a four-dimensional hyperbolic spacetime (M,$g$).
The corresponding action,
$S_{\rm geo}+S_{\rm matter}$, is given by the Einstein -- Hilbert action
\be \label{2.16}
S_{\rm geo}=\frac{1}{\kappa}\int_{\rm M}\; {\rm d}^4X\; \sqrt{|\det(g)|}\;
[R^{(4)}+2\Lambda] \;
\ee
where
$\kappa\equiv 16\pi G_{\rm N}$, with $G_{\rm N}$ denoting Newton's constant,
$R^{(4)}$ is the Ricci scalar of $g$ and $\Lambda$ denotes the cosmological
constant,
and the scalar field action
\be \label{2.17}
S_{\rm matter}=\frac{1}{2\lambda} \int_{\rm M}\; d^4X\; \sqrt{\det(g)|}\;
[-g^{\mu\nu} \zeta_{,\mu} \zeta_{,\nu}-v(\zeta)] \;
\ee
with $\lambda$ denoting a coupling constant allowing for a dimensionless $\zeta$
and $v$ is a potential term.

\subsection{Lagrangian Analysis}
\label{s2.1}

In their seminal paper~\cite{3} Brown and Kucha{\v r} introduced the
following dust action\footnote{A classical particle interpretation of this
action will be
given in section \ref{s2.4}.}
\be \label{2.1}
S_{\rm dust}=-\frac{1}{2} \int_{\rm M}\;d^4X\;\sqrt{|\det(g)|}\;
\rho\;[g^{\mu\nu}\;
U_\mu U_\nu+1] \;.
\ee
Here, $g$ denotes the four-metric on the spacetime manifold M. The dust
velocity field is defined by $U=-{\rm d}T+W_j {\rm d}S^j $ $(j\in {1,2,3}$.
The action $S_{\rm dust}$
is a functional of the fields $\rho,\;g_{\mu\nu},\;T,\;S^j,\;W_j$\footnote{Here,
$T, S^j$ have dimension of length, $W_j$ is dimensionless and, thus,
$\rho$ has dimension length$^{-4}$. The notation used here is suggestive: $T$
stands for {\sl time}, $S^j$ for {\sl space} and $\rho$ for {\sl dust energy
density}.}.
The physical interpretation of the action will now be given in a series of
steps.

First of all, the energy momentum of the dust reads
\be \label{2.2}
T^{\rm dust}_{\mu\nu}=-\frac{2}{\sqrt{|\det(g)|}} \frac{\delta S_{\rm
dust}}{\delta
g^{\mu\nu}}=\rho \; U_\mu U_\nu-\frac{\rho}{2}\; g_{\mu\nu}\; [g^{\lambda\sigma}
U_\lambda U_\sigma+1] \;.
\ee
By the Euler--Lagrange equation for $\rho$
\be \label{2.3}
\frac{\delta S_{\rm dust}}{\delta \rho}=g^{\lambda\sigma}
U_\lambda U_\sigma+1=0 \;
\ee
the second term in (\ref{2.2}) vanishes on shell. Hence, $U$ is unit timelike on
shell.
Comparing with the
energy momentum tensor of a perfect fluid with energy density $\rho$, pressure
$p$
and unit (timelike) velocity field $U$
\be \label{2.4}
T^{\rm pf}_{\mu\nu}=\rho \; U_\mu U_\nu + p\; (g_{\mu\nu}+U_\mu U_\nu) \;
\ee
shows that the action (\ref{2.1}) gives an energy-momentum tensor
for a perfect fluid with vanishing pressure.

For $\rho\not=0$, variation with respect to $W_j$ yields an equation equivalent
to
\be \label{2.5}
{\cal L}_U S^j=0 \; 
\ee
where $\cal L$ denotes the Lie derivative.
Hence, the fields $S^j$ are constant along the
integral curves of $U$. Equation (\ref{2.5}) implies
\be \label{2.6}
{\cal L}_U T=U^\mu T_{,\mu}=U^\mu [T_{,\mu}-W_j S^j_{,\mu}]=-U^\mu
U_\mu=+1 \;
\ee
so that $T$ defines proper time along the dust flow lines.

Variation with respect to $T$ results in
\be \label{2.7}
\partial_\mu[\rho \sqrt{|\det(g)|} U^\mu]=\sqrt{|\det(g)|}
\nabla_\mu [\rho U^\mu]=0 \;
\ee
while variation with respect to $S^j$ gives
\be \label{2.8}
\partial_\mu[\rho \sqrt{|\det(g)|} U^\mu W_j]=\sqrt{|\det(g)|}
\nabla_\mu [\rho U^\mu W_j]=0 \;.
\ee
Using (\ref{2.7}), (\ref{2.8}) reduces to (assuming $\rho\not=0$)
\be \label{2.9}
\nabla_U W_j=0 \;.
\ee
Thus, $\nabla_U U_\mu = 0$, and, as a consequence,
the integral curves of $U$ are affinely parametrised geodesics.
The physical
interpretation of the fields $T, S^j$ is complete: the vector field $U$ is
geodesic
with proper time $T$, and each integral curve is completely determined by
a constant value of $S^j$.
This determines a dynamical foliation of M, with leaves characterized
by constant values of $T$. A given integral curve intersects each leave at
the same value of $S^j$.

\subsection{Hamiltonian Analysis}
\label{s2.2}
In this section we derive the constraints that restrict the phase space of the
system
of a generic scalar field on a spacetime $({\rm M},g)$, extended by
the Brown -- Kucha{\v r} dust. The reader not interested in the
details of the derivation, which uses the full arsenal of Dirac's algorithm for
constrained systems, may directly refer to the result (\ref{2.27a}--\ref{2.29}).

We assume $({\rm M},g)$ to be globally hyperbolic in order to guarantee a well
posed initial value problem.
As a consequence, M is
diffeomorphic to $\mathbb{R}\times {\cal X}$, where $\cal X$ is a
three-manifold of arbitrary topology. The spacelike leaves ${\cal X}_t$ of the
corresponding foliation are obtained as images of a one parameter family
of embeddings $t\mapsto Y_t$, see e.g. \cite{19}
for more details and
our notation table for ranges of indices etc.
The timelike unit normals to the leaves may be written\footnote{We have written
$Y(t,x)\equiv Y_t(x)$.} as
$n^\mu=[Y^\mu_{\; ,t}-n^a Y^\mu_{\; ,a}]/n$,
where $n,\;n^a$ are called lapse
and shift functions, respectively.
Throughout, $n^\mu$ is assumed to be future
oriented with respect to the parameter $t$, which requires $n>0$.

The three metric on $\cal X$ is the pull back of the spacetime metric
under the embeddings, that is, $q_{ab}(x,t)=Y^\mu_{\; ,a} Y^\nu_{\; ,b}
g_{\mu\nu}$. Denoting the inverse of $q_{ab}$ by $q^{ab}$ it is not
difficult to see that
\be \label{2.11}
g^{\mu\nu}=-n^\mu n^\nu+q^{ab} \; Y^\mu_{\;,a} Y^\nu_{\;,b} \;.
\ee
It follows that the dust action can be written as
\be \label{2.12}
S_{\rm dust}=-\frac{1}{2}\int_{\mathbb{R}}\; dt\;\int_{{\cal X}} \; d^3x\;
\sqrt{\det(q)}\; n\;\rho\;\left(-U_n^2+q^{ab} U_a U_b+1\right) 
\ee
with $U_n\equiv n^\mu U_\mu,\;U_a \equiv Y^\mu_{\; ,a} U_\mu$.

The form (\ref{2.12}) is useful to derive the momentum fields
canonically conjugate to $T, S^j$,
respectively, as
\ba \label{2.12a}
P &:=& \frac{\delta S_{\rm dust}}{\delta T_{,t}}=-\sqrt{\det(q)}\; \rho \; U_n
\nonumber \\
P_j &:=& \frac{\delta S_{\rm dust}}{\delta S^j_{,t}}=\sqrt{\det(q)}\;\rho\; U_n
W_j \;.
\ea
The second relation in (\ref{2.12a}) shows that the Legendre transform is
singular,
and we obtain the {\sl primary constraint} ({\sl Zwangsbedingung})
\be \label{2.13}
Z_j:=P_j+W_j P=0 \;.
\ee
Additional primary constraints arise when we compute the momenta
conjugate to $\rho$ and $W_j$
\ba \label{2.14}
I &:=& Z := \frac{\delta S_{\rm dust}}{\delta \rho_{,t}}=0 
\nonumber \\
I^j &:=& Z^j := \frac{\delta S_{\rm dust}}{\delta W_{j,t}}=0 \;.
\ea
Considering the total action $S\equiv S_{\rm geo}+S_{\rm matter}+S_{\rm dust}$,
further primary constraint follow from the calculation of the canonical momentum
fields
conjugate to lapse and shift $n,n^a$,
respectively,
\ba \label{2.15}
p &:=& z:= \frac{\delta S}{\delta n_{,t}}=0
\nonumber\\
p_a &:=& z_a := \frac{\delta S}{\delta n^a_{,t}}=0
\ea
The primary constraints signify the fact that we cannot solve for the
velocities $\{S^j_{\; ,t},\;\rho_{,t},\;W_{j\;,t},\;n_{,t},\;n^a_{\; ,t}\}$,
respectively, in terms of the momenta and configuration variables.
Therefore, all primary constraints must be included in the canonical action,
together with
appropriate Lagrange multipliers
$\left\{\mu^j,\;\mu,\;\mu_j,\nu,\;\nu^a\right\}$,
in order to reproduce the Euler -- Lagrange equations.

It is straightforward to solve for $T_{,t}$ and $\zeta_{,t},\;q_{ab\; ,t}$.
For instance,
\be \label{2.18}
T_{,t}=n T_n+n^a T_{,a}
=n\left[-U_n+W_j S^j_n\right]+n^a T_{,a}
=n \frac{1}{\rho}\frac{P}{\sqrt{\det(q)}}+S^j_{\; ,t} W_j+n^a \left[T_{,a}-W_j
S^j_{\; ,a}\right]\;.
\ee
How to eliminate the velocities of the scalar field and the three-metric is well
known,e.g.~\cite{19},
and will not be repeated here.

The resulting Hamiltonian constraint for the extended system, $c^{\rm tot}
\equiv c^{\rm geo}+c^{\rm matter}+c^{\rm dust}$, is explicitly given by
\ba
\kappa \; c^{\rm geo} &=& \frac{1}{\sqrt{\det(q)}}\left[q_{ac}q_{bd}
-\frac{1}{2}
q_{ab}q_{cd}\right]p^{ab}p^{cd}-\sqrt{\det(q)}\; R^{(3)}+2\Lambda \sqrt{\det(q)} 
\nonumber\\
\lambda \; c^{\rm matter}
&=&
\frac{1}{2}\left[\frac{\pi^2}{\sqrt{\det(q)}}+\sqrt{\det(q)}\left(q^{ab}\xi_{,a}
\xi_{,b}+v(\xi)\right)\right] 
\nonumber\\
c^{\rm dust} &=&
\frac{1}{2}\left[\frac{P^2/\rho}{\sqrt{\det(q)}}+\sqrt{\det(q)}\;
\rho\left(q^{ab} U_a
U_b+1\right)\right] 
\ea
with $U_a\equiv -T_{,a}+W_j \; S^j_{\; ,a}$.
The spatial diffeomorphism constraints for the extended system,
$c_a^{\rm tot} \equiv c^{\rm geo}_a+c^{\rm matter}_a+c^{\rm dust}_a$, are
explicitly given by
\ba
\kappa \; c_a^{\rm geo} &=& -2\; q_{ac}D_b \; p^{bc} \; 
\nonumber\\
\lambda \; c_a^{\rm matter} &=& \pi \; \xi_{,a}
\nonumber\\
c^{\rm dust}_a &=& P\left[T_{,a} - W_j \; S^j_{\; ,a}\right]
\; .
\ea

The total action in canonical form reads
\ba \label{2.19}
S
&& =
\int_{\mathbb{R}}\; dt\; \int_{{\cal X}}\; {\rm d}^3x\;\left(P T_{,t}+P_j\;
S^j_{\; ,t}+I \; \rho_{,t}+I^j \; W_{j,t}+p\; n_{,t}+p_a\;
n^a_{,t}+\frac{1}{\kappa} p^{ab} \; q_{ab,t}+\frac{1}{\lambda} \pi \;
\xi_{,t}\right)
\nonumber \\
&&- \int_{\mathbb{R}}\; dt\; H_{\rm primary}
\; 
\ea
with $p^{ab}$ denoting the momentum field conjugate to $q_{ab}$,
$\xi$ denoting the pullback of $\zeta$ to $\cal X$, $\pi$
denoting its canonical momentum and $D$ is the covariant
differential compatible with $q_{ab}$.
Furthermore, the Hamiltonian and spatial diffeomorphism constraints, together
with the
primary constraints, entered the definition of the primary Hamiltonian
\be
H_{\rm primary}\equiv\int_{{\cal X}}\; d^3x\;
h_{\rm primary}
\ee
via the density
\be
h_{\rm primary} \equiv \mu^j \; Z_j+\mu \; Z + \mu_j \; Z^j+\nu\; z+\nu^a \; z_a
+
n\; c^{\rm tot}+
n^a \; c^{\rm tot}_a
\; .
\ee

Consistency requires that the the constraint surface,
defined by the primary constraints (\ref{2.13}),
(\ref{2.14}) and (\ref{2.15}), is stable under the action
of $H_{\rm primary}$. In order to to check this, we summarise the only
non-vanishing
elementary Poisson brackets\footnote{
Notice that $n,n^a,W_j,\rho,S^j$ are {\sl not} Lagrange multipliers at
this point, they are canonical coordinates just like the other fields.}
\ba \label{2.20}
\{p^{ab}(x),q_{cd}(y)\} &=& \kappa \; \delta^a_{(c} \delta^b_{d)} \;
\delta(x,y) \; 
\nonumber\\
\{\pi(x),\xi(y)\} &=& \lambda \; \delta(x,y)\; 
\nonumber\\
\{P(x),T(y)\} &=& \delta(x,y)
\nonumber\\
\{P_j(x),S^k(y)\} &=& \delta_j^k\;\delta(x,y)
\nonumber\\
\{I(x),\rho(y)\} &=& \delta(x,y)
\nonumber\\
\{I^j(x),W_k(y)\} &=& \delta^j_k \;\delta(x,y)
\nonumber\\
\{p(x),n(y)\} &=& \delta(x,y)
\nonumber\\
\{p_a(x),n^b(y)\} &=& \delta_a^b\;\delta(x,y)\; .
\ea

The primary constraints transform under the action of the primary Hamiltonian
$H_{\rm
primary}$ as follows
\ba \label{2.21}
z_{,t} &=& \{H_{\rm primary},p\}=-c^{\rm tot}
\nonumber\\
z_{a,t} &=& \{H_{\rm primary},p_a\}=-c_a^{\rm tot}
\nonumber\\
Z_{,t} &=&
\{H_{\rm primary},I\}=\frac{n}{2}\left[-\frac{P^2/\rho^2}{\sqrt{\det(q)}}
+\sqrt{\det(q)}\left(q^{ab} U_a U_b+1\right)\right] \equiv \tilde{c}
\nonumber\\
Z^j_{,t} &=& \{H_{\rm primary},I^j\}=-\mu^j \; P -n\;\rho\; \sqrt{\det(q)}\;
q^{ab} \;U_a \;S^j_{,b} + P\; S^j_{,a}\; n^a
\nonumber\\
Z_{j,t} &=& \{H_{\rm primary},P_j+W_j P\}=\mu_j \; P - \left(n^a-\frac{n \rho
\sqrt{\det(q)}}{P}\;  q^{ab} U_b\right) P \; W_{j,a}\; .
\ea
Consistency demands that (\ref{2.21}) must vanish. Indeed, the last two
equations
in (\ref{2.21}) involve the Lagrange multipliers $\mu^j,\mu_j$, respectively,
and can be solved for them, since the system of equations has maximal
rank.
However, the first three equations in (\ref{2.21}) do not involve Lagrange
multipliers. Hence, they represent secondary
constraints. According to Dirac's algorithm, the
secondary constraints in equation (\ref{2.21}) force us to reiterate the
stability analysis, i.e.~to
calculate
the action of $H_{\rm primary}$ on the secondary constraints. A lengthy
calculation
presented in \ref{sa} shows that the secondary constraints are stable under the
Hamiltonian flow generated by $H_{\rm primary}$. In other words, no tertiary
constraints
arise in the stability analysis for the secondary constraints. However,
the action of $H_{\rm primary}$ on $\tilde{c}$ involves the Lagrange multipliers
$\mu^j,\mu_j$,$\mu$, and can be solved for $\mu$.

The final set of constraints is given by
$\left\{c^{\rm tot},\;c_a^{\rm
tot},\;\tilde{c},\;Z_j,\;Z^j,\;Z,\;z_a,\;z\right\}$ and
it
remains to classify them into first and second class, respectively. Obviously,
\ba \label{2.22}
\{Z^j(x),Z_k(y)\} &=& P\;\delta^j_k \; \delta(x,y) 
\nonumber\\
\{Z(x),\tilde{c}(y)\} &=& \frac{n P^2}{\rho^3\sqrt{\det(q)}} \;\delta(x,y) 
\ea
does not vanish on the constraint surface defined by the final set of
constraints, hence they are of second class. Next, since $n$ appears at
most linearly in the constraints, while $n^a$ does not appear at all,
it follows immediately that $z,\;z_a$ are of first class. Further, consider
the linear combination of constraints
\ba \label{2.23}
\tilde{c}^{\rm tot}_a &\equiv& I\; \rho_{,a}+I^j \; W_{j,a}+P\;T_{,a} +P_j
\;S^j_{,a}
+p\;n_{,a}+{\cal L}_{\vec{n}} \; p_a+c_a
\nonumber\\
&=& c^{\rm tot}_a + Z \; \rho_{,a}+Z^j \; W_{j,a}+Z_j\; S^j_{,a}+z\; n_{,a}
+{\cal L}_{\vec{n}} \; z_a
\ea
where
\be \label{2.24}
c_a \equiv c^{\rm geo}_a+c^{\rm matter}_a
\ee
is the non-dust contribution to the spatial diffeomorphism constraint
$c^{\rm tot}_a$. Since all constraints are scalar or covector densities of
weight one and $\tilde{c}^{\rm tot}_a$ is the generator of spatial
diffeomorphisms, it follows that $\tilde{c}^{\rm tot}_a$ is first class.
Finally, we consider the linear combination
\be \label{2.25}
\tilde{c}^{\rm tot} \equiv c^{\rm tot}+\alpha^j \;Z_j+\alpha_j\; Z^j+\alpha \;
Z\; 
\ee
and determine the phase space functions $\alpha^j,\alpha_j,\;\alpha$
such that $\tilde{c}^{\rm tot}$ has vanishing Poisson brackets with
$Z_j,Z^j,Z$ up to terms proportional to $Z_j, Z^j, Z$.
Then, $\tilde{c}^{\rm tot}$ is first class, as well. See appendix \ref{sa} for
details.

In the final step we should calculate the Dirac bracket~\cite{1,15b}
$\{f,g\}^\ast$ for phase space functions $f,g$. It differs
from the Poisson bracket $\{f,g\}$ by linear combinations of terms of
the form $\{f,Z_j(x)\}\;\{g,Z^k(y)\}$ and
$\{f,Z(x)\}\;\{g,\tilde{c}(y)\}$ (and terms with $f,g$ interchanged).
Fortunately, the Dirac bracket agrees with the Poisson bracket
on functions $f,g$ which only involve $\left\{T,\;S^j,\;
q_{ab},\;n,\;n^a\right\}$ and
their
conjugate momenta $\left\{P,\;P_j,\;P^{ab},\;p,\;p_a\right\}$ on which we focus
our attention in what follows. Using
the Dirac bracket, the second class constraints can be solved strongly:
\ba \label{2.26}
Z_j=0 &\Leftrightarrow& W_j=-P_j/P 
\nonumber\\
Z^j=0 &\Leftrightarrow& I^j=0 
\nonumber\\
Z=0 &\Leftrightarrow& I=0 
\nonumber\\
\tilde{c}=0 &\Leftrightarrow& \rho^2=\frac{P^2}{\det(q)}\left(q^{ab} U_a
U_b+1\right)^{-1} \;.
\ea
From the last equation in (\ref{2.26}) we find
\be \label{2.27}
\rho=\epsilon\; \frac{P}{\sqrt{\det(q)}}\;\big(\sqrt{q^{ab} U_a U_b +1}\big)^{-1} \;,
\ee
with $\epsilon=\pm 1$.
We may also partially reduce the phase space subject to
(\ref{2.26}) by setting $z=z_a=0$ and treating $n,\;n^a$ as Lagrange
multipliers, since they are pure gauge. Then, we are
left with two constraints
\ba \label{2.27a}
c^{\rm tot} &=& c+c^{\rm dust} \; 
\nonumber\\
c^{\rm tot}_a &=& c_a+c^{\rm dust}_a \; 
\ea
where
\be \label{2.28}
c\equiv c^{\rm geo}+c^{\rm matter}
\ee
and
\ba \label{2.29}
c^{\rm dust} &=& \epsilon \; P\; \sqrt{q^{ab} U_a U_b +1}\;
\nonumber\\
c^{\rm dust}_a &=& P\; T_{,a}+P_j \;S^j_{\;,a} \;.
\ea

Equations (\ref{2.27a}--\ref{2.29}) are the main result of this subsection.
They represent the final constraints that restrict the phase space of the system
consisting of a generic scalar field on (M, $g$), extended by dust.
The form of the dust Hamiltonian and spatial diffeomorphism
constraints $\left\{c^{\rm dust}\;, c^{\rm dust}_a\right\}$, respectively, is of
paramount importance
for utilising dust as a deparameterising system, as we will explain
in the next section.

\subsection{The Brown -- Kucha{\v r} Mechanism for Dust}
\label{s2.3}
In the previous section we have shown that
the canonical formulation of a classical system, originally described by General
Relativity and a
generic scalar field theory, then extended by a specific dust model, results in
a
phase space subject to the Hamiltonian and spatial diffeomorphism constraints
(\ref{2.27a}--\ref{2.29}).
The primary Hamiltonian, after having solved the second class constraints,
is a linear combination of those final first class constraints
(\ref{2.27a}--\ref{2.29})
and, thus, is constrained to vanish. This holds in general, independently of
the matter content, and is a direct consequence of the underlying
spacetime diffeomorphism invariance.

Now, observable quantities are special phase space functions, distinguished by
their
invariance under gauge transformations. In other words, their Poisson brackets
with the constraints must vanish when the constraints hold. In particular, they
have
vanishing Poisson brackets with the primary Hamiltonian $H_{\rm primary}$ on the
constraint surface.
This is one of the many facets of
the problem of time: observable quantities do not move with respect to
the primary Hamiltonian, because the latter generates gauge
transformations rather than physical evolution. It follows
that physical evolution must be generated by a true Hamiltonian (not
constrained to vanish, but still gauge invariant).

In this section we address the questions how to construct a true Hamiltonian
from a given Hamiltonian constraint, and, how to construct observable
quantities (gauge invariant phase space functions).

\subsubsection{Deparametrisation: General Theory}
\label{2.3.1}

The manifest gauge invariant construction of a true Hamiltonian,
generating physical evolution as opposed to mere gauge transformations,
becomes particular simple when the original system under consideration
can be extended to a system with constraints in  deparametrised form.

Consider first a general system subject to first class
constraints $c_I$. The set of canonical pairs on phase space split 
into two sets of canonical pairs $(q^a,p_a)$ and
$(T^I,\pi_I)$, respectively, such that the constraints can be solved,
at least locally in phase space, for the $\pi_I$. In other words,
\be \label{2.30}
c_I=0 \;\;\Leftrightarrow\;\; \tilde{c}_I=\pi_I+h_I(T^J;q^a,p_a)=0 \;.
\ee
Notice that, in general, the functions $h_I$ do depend on the $T^J$.
The first class property guarantees that the
$\tilde{c}_I$ are mutually Poisson commuting \cite{27}.

A system that deparametrises allows to split the set of canonical pairs into two
sets of canonical pairs such that
(1) equation (\ref{2.30}) holds globally on phase
space\footnote{This is not the case for Klein -- Gordon fields and
many other scalar field theories with a canonical action that is
at least quadratic in the $\pi_I$.}, and (2) the functions
$h_I$ are independent of the $T^J$.

Property (2) implies that the functions $h_I$ are gauge invariant.
Hence, any linear combination of the $h_I$ that is bounded from below
is a suitable candidate for a true Hamiltonian in the following sense:
let $\tilde{c}_\tau \equiv \tau^I \tilde{c}_I$ be such a linear combination,
with real coefficients $\tau^I$ in the range of $T^I$, and consider for
any phase space function $f$ the expression
\be \label{2.30a}
O_f(\tau)\equiv \left[\sum_{n=0}^\infty \;
\frac{1}{n!}\{\tilde{c}_\tau,f\}_{(n)}\right]_{\tau^I \to
(\tau-T)^I} \;.
\ee
Here\footnote{Notice that the
substitution of the phase space independent numbers $\tau^I$ by the phase
space dependent combination $(\tau-T)^I$ is performed only {\sl after} the
series
has been calculated.}, the iterated Poisson bracket is inductively defined by
$\{\tilde{c}_\tau,f\}_{(0)}=f,\;
\{\tilde{c}_\tau,f\}_{(n+1)}=
\{\tilde{c}_\tau,\{\tilde{c}_\tau,f\}_{(n)}\}$.
Then, $O_f(\tau)$ is an observable quantity. More precisely, it is a gauge
invariant
extension of the phase space function f. Furthermore, physical time translations
of
$O_f(\tau)$ are generated by the functions $h_I$:
\be \label{2.30b}
\frac{\partial O_f(\tau)}{\partial \tau^I}=\{h_I,O_f(\tau)\} 
\ee
provided that $f$ only\footnote{This is no restriction, since the $\pi_I$
can be expressed in terms of the $(q^a,\;p_a)$ (using (\ref{2.30})),
and the $T^I$ are pure gauge.} depends on $(q^a,\;p_a)$.

The observable quantities $O_f(\tau)$ can also be interpreted from the
point of view of choosing a {\sl physical} gauge. Indeed, $O_f(\tau)$ can be
interpreted as representing
the value of $f$ in the gauge $T^I=\tau^I$.

\subsubsection{Deparametrisation: Scalar Fields}
\label{s2.3.2}

The Brown -- Kucha{\v r} mechanism relies on the observation that free scalar
fields lead to deparametrisation of General Relativity, as we sketch below
(see \cite{8} for a detailed discussion).

A free scalar field contributes to the spatial diffeomorphism constraint a
term of the form
\be \label{2.31}
c_a^{\rm scalar}=\pi \phi_{,a} 
\ee
and to the Hamiltonian constraint a function
of $\pi^2$ and $q^{ab} \phi_{,a} \phi_{,b}$, in the absence of a
potential. On the constraint surface, defined by the spatial
diffeomorphism constraint, we have the identity
\be \label{2.32}
q^{ab} \phi_{,a} \phi_{,b}=
\frac{q^{ab} \; c^{\rm scalar}_a c^{\rm scalar}_b}{\pi^2}
=\frac{q^{ab} \; c_a c_b}{\pi^2}
\ee
with $c_a$ denoting the contribution to the total spatial diffeomorphism
constraint that is independent of the free scalar field.
Substitution of (\ref{2.32}) into the total Hamiltonian and
spatial diffeomorphism constraints yields the same constraint surface
 and gauge flow than before. In other words, the constraints with the
substitution
(\ref{2.32}) are equivalent to the original ones. However,
the new total Hamiltonian constraint does not depend on the free scalar
field $\phi$ any more. Therefore, at least locally in phase space,
we can solve the new total Hamiltonian constraint for
the momentum field $\pi$ and write locally
\be \label{2.33}
\tilde{c}^{\rm tot}(x)=\pi(x)+h(x)
\ee
where the scalar density $h$ of weight one is independent of $\pi,\phi$ and,
typically, positive definite, see~\cite{8} for details.

As mentioned above, the constraint (\ref{2.33}) and $h(x)$ are mutually Poisson
commuting, which
guarantees that the {\sl physical Hamiltonian}
\be \label{2.34}
H:=\int_{{\cal X}} \; {\rm d}^3x\; h(x)
\ee
is observable (it has vanishing Poisson brackets with the spatial diffeomorphism
constraint,
because $h$ has density weight one).

This is as much as the general theory goes. There are two remaining caveats:
first of all, the construction is only local in phase space. Secondly,
the construction based on a single free scalar field requires phase space
functions that are already invariant under spatial diffeomorphisms.
Only those can be completed to fully gauge invariant quantities\footnote{
This can be circumvented by employing e.g. three more free scalars but this
would be somewhat ad hoc.}.

\subsubsection{Deparametrisation: Dust}
\label{s2.3.3}
Dust described by the action (\ref{2.1}) does not
entirely fit into the classification scheme given in~\cite{8} and sketched in
the last section. It is not simply based on four free scalar fields
$T,S^j$, but in addition leads to second class constraints. However,
it has a clear interpretation as a system of test observers in
geodesic motion, and circumvents the remaining caveats mentioned
at the end of the last subsection as we will see.

Recall the final form of the Hamiltonian constraint (\ref{2.27a}--\ref{2.27a})
derived in the previous section:
\be \label{2.35}
c^{\rm tot}=c+\epsilon P\sqrt{1+q^{ab} U_a U_b}
\ee
with $U_a=-T_{,a}+W_j\; S^j_{\;,a}$. Solving the second class constraint
$Z_j=0$ for $W_j$, we find $U_a=-c^{\rm dust}_a/P$. Inserting the first class
spatial diffeomorphism constraint $c^{\rm tot}_a=c_a+c^{\rm dust}_a$, we arrive
at the equivalent Hamiltonian constraint
\be \label{2.36}
c^{{\rm tot}\;\prime}=c+\epsilon P\sqrt{1+\frac{q^{ab} c_a c_b}{P^2}}
\ee
which is already independent of $T,S^j$ and $P_j$, but still not of
the form $\tilde{c}^{\rm tot}=P+h$, as required for a system that deparametrises.

\subsubsection{Deparametrisation for Dust: Sign Issues}
\label{s2.3.4}

In order to bring (\ref{2.36}) into the form $\tilde{c}^{\rm tot}=P+h$,
we have to solve a quadratic equation. Each root
describes only one sheet of the constraint surface, unless the sign of $P$
is somehow fixed. As we argue below, this freedom will be fixed
by our interpretation of the
dust system as a {\sl physical reference system}.

Recall that
$P=-\rho \sqrt{\det(q)} U_n$ and $U^\mu \; T_{,\mu}=1, \;U^\mu\;
S^j_{\;,\mu}=0$. In accordance with our interpretation,
we identify $T$ with proper time along the dust
flow lines. Thus, $U$ is timelike and future pointing, hence $U_n<0$.
It follows that sgn$(P)=$sgn$(\rho)$, so $\epsilon=1$ in (\ref{2.27}).

In \cite{3} the authors assume $\rho>0$, as it is appropriate for {\sl
observable} dust\footnote{
This would be required by the usual energy
conditions if the dust would be the only observable matter. However,
notice that only the total energy momentum is subject to the energy
conditions, not the individual contributions from various matter
species.}.
In our case, however, the dust serves only as a tool to deparametrise the system and is, by construction, only pure gauge.
Therefore, we relax  the restriction $\rho>0$, when solving
(\ref{2.36}) for $P$:
\be \label{2.38}
P^2=c^2-q^{ab} c_a c_b\; .
\ee
The right hand side of (\ref{2.38}) is {\sl constrained to be} non -- negative,
albeit it is not manifestly non -- negative.
But this causes no problem, since it is sufficient to analyse the
system in an arbitrarily small neighborhood of the constraint surface,
where $c^2-q^{ab} c_a c_b \geq 0$.
Then,
\be \label{2.39}
\tilde{c}^{\rm tot}=P-\mbox{sgn}(P) \;h, \hspace{0.5cm} h=\sqrt{c^2-q^{ab} c_a
c_b}
\ee
is the general solution, globally defined on (the physically
interesting portion of) the full phase
space.
However, $\tilde{c}^{\rm tot}$ is not yet of the form required by a successful
deparametrisation, because of the sign function which also renders the
constraint non --
differentiable.

In order to utilise dust for deparametrisation, the choice $P<0$ is
required.
Before presenting reasons for this choice, we stress again that the
dust itself is not observable. There are three related arguments for the choice
of $P<0$:
\begin{itemize}
\item[1.] {\sl Dynamics}\\
The deparametrisation mechanism supplies us with a physical Hamiltonian
of the form \be \label{2.40} \HF=\int_{{\cal X}} \; d^3x\; h \; .
\ee In the case where dust is chosen as the clock of the system, the variation
of the physical Hamiltonian is given by \be \label{2.41} \delta
\HF=\int_{{\cal X}} \; {\rm d}^3x\; \left(\frac{c}{h}\; \delta c
-q^{ab}\; \frac{c_b}{h}\;\delta c_a+ \frac{1}{2 h} \; q^{ac} q^{bd}
c_c c_d \; \delta q_{ab}\right) \;. \ee For $P\not=0$, then
$h\not=0$ (in a sufficiently small neighbourhood of the constraint
surface). Hence, the coefficients of the variations on the right
hand side of (\ref{2.41}) are non singular. Moreover, for $P\not=0$,
also $c\not=0$, as we see from (\ref{2.36}). In fact, using
sgn$(c)=-$sgn$(P)$ (from (\ref{2.36})) in a neighborhood of the
constraint surface, \be \label{2.42} \frac{c}{h}=-{\rm
sgn}(P)\;\sqrt{1+q^{ab}\frac{c_a}{h}\frac{c_b}{h}} \ee has absolute
value no less than one.

Let us now compare (\ref{2.41}) with the differential of the primary
Hamiltonian constraint in the absence of dust:
\be \label{2.43}
H_{\rm primary}=\int_{{\cal X}}\;{\rm d}^3x\; \left(nc+n^a c_a\right) 
\ee
which is given by (lapse and shift functions are considered as
Lagrange
multipliers, i.e. are phase space independent)
\be \label{2.44}
{\rm d} H_{\rm primary}=\int_{{\cal X}}\;{\rm d}^3x\; \left(n dc+n^a d
c_a\right)
\; .
\ee
Comparison between (\ref{2.41}) and (\ref{2.44}) reveals that the
differentials coincide, up to the additional term proportional to
$\delta q_{ab}$, provided we identify $n:=c/h$ as dynamical lapse and $n^a:=-
q^{ ab} c_b/h$ as dynamical shift.
This is promising in our aim to
derive physical equations of motions for observable quantities which
nevertheless come close to the usual Einstein equations for gauge
dependent quantities. However, in the standard framework the lapse
function is always positive, guaranteeing that the normal vector field is
future oriented. This fact is correctly reflected in our framework only
if $P<0$.
\item[2.] {\sl Kinematics}\\
The identification
$n\equiv c/h$ and $n^a\equiv - q^{ ab} c_b/h$ can also be motivated as
follows:\\
Consider a spacetime
diffeomorphism defined by $X^\mu\mapsto
(\tau,\sigma^j):=(T(X),S^j(X))=:Y^\mu(X)$ and let $(\tau,\sigma^j)\to
Z^\mu(\tau,\sigma)$ be its inverse. We can define a {\sl dynamical}
foliation of $M$ by $T(X)=\tau=$const. hypersurfaces. The leaves ${\cal
S}_\tau$ of
that foliation are the images of $\cal S$ (which is the range of the
$S^j$) under
the map $Z$ at constant $\tau$. Using the identity
\be \label{2.45}
\delta^\mu_\nu=Z^\mu_{\;,\tau}\; T_{,\nu}+Z^\mu_{\;,j}\; S^j_{\;,\nu}
\ee
and $U^\mu \; T_{,\mu}=1\;, \; U^\mu \; S^j_{,\mu}=0$, we find
$U^\mu=Z^\mu_{\;,\tau}$. Thus, as expected, the foliation is generated by
the vector field $U=\partial/\partial\tau$, which is unit timelike.

It is useful to decompose the deformation vector field $U$ with respect
to the arbitrary coordinate foliation that we used before:
\be \label{2.46}
U^\mu=g^{\mu\nu} U_\nu=-n^\mu U_n+X^\mu_{\;,a} \; q^{ab}\; U_b \;.
\ee
From (\ref{2.12}) and (\ref{2.27}) with $\epsilon=1$ we find
$U_n=-\sqrt{1+q^{ab} U_a U_b}$. Next,
\be \label{2.47}
U_a=-\frac{c^{\rm dust}_a}{P}=\frac{c_a}{P} \; .
\ee
On the other hand $n\equiv c/h={\rm sgn}(P) U_n$ and $n_a\equiv -c_a/h
=-{\rm sgn}(P) \; U_a$. Therefore, (\ref{2.46}) can be written
\be \label{2.48}
U^\mu=-{\rm sgn}(P)\; \left(n\; n^\mu+X^\mu_{,a} \; n^a\right) \;.
\ee
Hence, the sign for which $n$ is positive yields the correct
decomposition of the deformation vector field $U$ in terms of lapse and shift.
This calculation reveals also the geometrical origin of the identification
$n\equiv c/h$ and $n^a \equiv - q^{ ab} c_b/h$.

As a side remark: the identity $-n^2+q^{ab} n_a n_b=-1$ is an
immediate consequence of the normalisation of the deformation vector field,
$g_{\mu\nu} U^\mu U^\nu=-1$.
That is, the deformation vector field is timelike, future
oriented and normalised, but not normal to the
leaves of the foliation that it defines.

\item[3.] {\sl Stability and flat spacetime limit}\\
Of course, we could choose $P>0$ and use $-h$ instead of $h$ in order to
obtain equations of motion. However, in that case the physical Hamiltonian
would be unbounded from below, leading to an unstable theory. Alternatively,
we could stick to $+h$ for the equations of motion, but
then the $\tau$ evolution would run backwards.

Moreover, since $c^{\rm tot}=c+P\sqrt{1+q^{ab} U_a U_b}=0$
on the constraint surface, we would have $c<0$ for $P>0$. Since
$c=c^{\rm geo}+c^{\rm matter}$ and $c^{\rm matter}>0$, this would enforce
$c^{\rm geo}<0$.
Hence, flat space would not be a solution.

As a side remark: for $c_a/h\ll 1$ and $P<0$, $h\approx c$,
while $h\approx -c$ for $P>0$. Thus, the physical Hamiltonian
density, with respect to dust as a physical reference system,
approximates the standard model Hamiltonian density $c^{\rm matter}$
only for $P<0$.
\end{itemize}
We emphasise again that the dust used for deperametrisation is
not observable, and should not be confused with observable matter.
It solely provides a dynamical reference frame.

\subsection{Dust Interpretation}
\label{s2.4}
In this section we derive a physical interpretation of the
Brown -- Kucha{\v r} action based on the geodesic motion
of otherwise free particles~\cite{3}.

Consider first the action for a single relativistic particle with
mass $m$ on a background $g$: \be \label{2.49}
S_m=-m\int_{\mathbb{R}}\; ds\; \sqrt{-g_{\mu\nu}\; \dot{X}^\mu \;
\dot{X}^\nu} \;. \ee The momentum conjugate to the configuration
variable $X^\mu$ is given by \be \label{2.50} P_\mu=\frac{\delta
S_m}{\delta
\dot{X}^\mu}=m\frac{g_{\mu\nu}\dot{X}^\nu}{\sqrt{-g_{\rho\sigma}
\dot{X}^\rho \dot{X}^\sigma}} \;, \ee rendering the Legendre
transformation singular. This is a consequence of the
reparametrisation invariance of the action (\ref{2.49}). Hence, the
system exhibits no physical Hamiltonian, but instead a primary
Hamiltonian constraint enforcing the mass shell condition: \be
\label{2.51} C=\frac{1}{2m}\left(m^2+g^{\mu\nu} P_\mu P_\nu\right)\;
. \ee

Let us proceed to the canonical formulation.
In terms of the embeddings $X\equiv Y_t(x)$, the particle trajectory
reads $X(s)=Y_{t(s)}(x(s))$, so that
\be \label{2.52}
\dot{X}(s)=\dot{t}(s)\; Y_{,t}+\dot{x}^a(s)\; Y_{,a} \;
\ee
where the overdot refers to differentiation with respect to the trajectory
parameter $s$.
The momenta are then given by
\ba \label{2.53}
p_a &\equiv& Y^\mu_{,a} P_\mu=\frac{m}{\sqrt{-g_{\rho\sigma}
\dot{X}^\rho \dot{X}^\sigma}}\; q_{ab}\left(\dot{t}\; n^b+\dot{x}^b\right) 
\nonumber\\
p_t &\equiv& Y^\mu_{,t} P_\mu=\frac{m}{\sqrt{-g_{\rho\sigma}
\dot{X}^\rho \dot{X}^\sigma}} \left(\dot{t}\; g_{tt}+q_{ab} \;n^b
\dot{x}^a\right) 
\ea
where $g_{tt}=-n^2+q_{ab} n^a n^b$.
We can only eliminate the spatial velocities $\dot{x}^a$.
To do this set $A\equiv g_{ta} \dot{x}^a=q_{ab} n^b
\dot{x}^a$ and $B\equiv q_{ab} \dot{x}^a \dot{x}^b$. Then,
\be \label{2.54}
w\equiv g_{\mu\nu} \dot{X}^\mu \dot{X}^\nu=g_{tt}\; \dot{t}^2+2A \;\dot{t}+B \;.
\ee
On the other hand,
\be \label{2.55}
-\frac{w}{m^2} q^{ab} p_a p_b=B+2 A\dot{t}+q_{ab} n^a n^b \dot{t}^2 \;.
\ee
Substituting $w$ from (\ref{2.54}) and collecting coefficients of
$A,B,\dot{t}^2$ yields
\ba \label{2.56}
0 &=& B+2A \; \dot{t}+\dot{t}^2
\frac{g_{tt}\; \frac{q^{ab} p_a p_b}{m^2}+q_{ab} n^a n^b}{1+\frac{q^{ab} p_a
p_b}{m^2}}
\nonumber\\
&=& w +\frac{\dot{t}^2\;n^2}{1+\frac{q^{ab} p_a p_b}{m^2}} \;.
\ea
Now we can solve the first equation in (\ref{2.53}) for $\dot{x}^a$:
\be \label{2.57}
\dot{x}^a=\dot{t}\;\left(-n^a\pm\sqrt{1+\frac{q^{ab} p_a p_b}{m^2}}\right) \;.
\ee
Inserting this into the second equation in (\ref{2.53}) leads to a
constraint of the form $C\equiv p_s+h$:
\be \label{2.58}
C=p_s-n^a p_a\pm n\sqrt{m^2+q^{ab} p_a p_b}
\ee
while the canonical Hamiltonian is obtained from the
Lagrangian in (\ref{2.49}) as
\be \label{2.59}
H_{\rm canon}=P_\mu \dot{X}^\mu-L=\dot{t} \; C \;.
\ee
Since the constraint (\ref{2.58}) is in deparametrised form,
the phase space can easily be reduced, leading to the reduced action
\be \label{2.60}
S_{\rm reduced}=\int\; {\rm d}s\; \left(p_a \dot{x}^a -h\right)\; .
\ee
We extend this phase space by adding a canonical pair $(\tau,m)$ and
consider the extended action
\be \label{2.61}
S_{\rm extended}=\int\; {\rm d}s\; \left(m \dot{\tau}+p_a \dot{x}^a -h\right) 
\ee
where the particle mass $m$ is now considered as a dynamical variable.
The equations of motion for $m,\tau$ give $\dot{m}=0$ and
$\dot{\tau}=\dot{t}\sqrt{-w}$. Thus, the mass is constant and $\tau$ is the
proper time (in the gauge $s=t$).

We generalize our results now to the case of many particles. More precisely,
let $\cal S$ be a label set and consider a relativistic particle for
each label $\sigma \in {\cal S}$. This amounts to provide
each variable appearing in the extended action with a corresponding label,
i.e.~$x^a_\sigma,\;p_a^\sigma,\;\tau_\sigma,\;m^\sigma$, and the
total action for those particles is then just the sum over the
corresponding actions $S^\sigma_{\rm extended}$
\be \label{2.62}
S_{\rm extended}=\sum_{\sigma \in {\cal S}} S^\sigma_{\rm extended} \; .
\ee

Next we consider the limit in which $\cal S$ becomes a three -- manifold,
with the labels $\sigma$ becoming coordinates on this manifold.
In this limit, we introduce the following fields:
\ba \label{2.63}
\tilde{T}(\sigma) &\equiv& \tau_\sigma \nonumber\\
\tilde{P}(\sigma)\;d^3\sigma&\equiv& m^\sigma \nonumber\\
\tilde{S}^a(\sigma)&\equiv& x^a_\sigma\nonumber\\
\tilde{p}_a(\sigma)d^3\sigma &\equiv &p_a(x_\sigma) \nonumber \\
\tilde{n}(\sigma)&\equiv& n(x_\sigma)\nonumber\\
\tilde{n}^a(\sigma)&\equiv &n^a(x_\sigma) \nonumber \\
\tilde{q}_{ab}(\sigma) &\equiv& q_{ab}(x_\sigma).
\ea
Then, in the specified limit, the extended action (\ref{2.62}) becomes
\be \label{2.64}
S_{\rm extended}=\int {\rm d}t \int {\rm d}^3\sigma \left(
\dot{\tilde{T}} \tilde{P}+\dot{\tilde{S}}^a \tilde{P}_a+ \tilde{n}^a
\tilde{P}_a
\mp \tilde{n} \; \sqrt{\tilde{P}^2+\tilde{q}^{ab} \tilde{P}_a
\tilde{P}_b}\right)\;.
\ee

Finally, we perform a canonical transformation: instead of the
fields $\tilde{S}^a(\sigma)$ with values in $\cal X$, we would like to
consider the inverse fields $S^j(x)$ with values in $\cal S$, that is
$S^j(\tilde{S}(\sigma))=\sigma^j,\; \tilde{S}^a(S(x))=x^a$.
This is at the same time a diffeomorphism and we can transform the other
fields as well. For instance ($T$ is a scalar and $P$ is a scalar density),
\ba \label{2.65}
T(x) &=& \tilde{T}(S(x))=\int_{{\cal S}}\; {\rm d}^3\sigma \;
\delta\left(x,\tilde{S}(\sigma)\right)\; \left|\det(\partial
\tilde{S}/\partial\sigma)\right|
\Tilde{T}(\sigma) 
\nonumber\\
P(x) &=& \frac{\tilde{P}}{|\det(\partial \tilde{S}/\partial\sigma)|}(S(x))
=\int_{{\cal S}}\; {\rm d}^3\sigma \;
\delta\left(x,\tilde{S}(\sigma)\right)\;
\tilde{P}(\sigma) 
\nonumber\\
S^j(x) &=& \int_{{\cal S}}\; {\rm d}^3\sigma \; \sigma^j\;
\delta(x,\tilde{S}(\sigma))\; \left|\det(\partial
\tilde{S}/\partial\sigma)\right| \; .
\ea
Calculating the time derivatives and performing integrations by
parts, we find
\be \label{2.66}
\int_{{\cal S}} \;{\rm d}^3\sigma\; \dot{\tilde{T}} \tilde{P}
=\int_{{\cal X}} \;{\rm d}^3x\; \left(\dot{T} P-\dot{S}^j S^a_j P T_{,a}\right)
\ee
with $S^a_j$ denoting the inverse of the matrix $S^j_{,a}$. Using
\be \label{2.67}
\dot{\tilde{S}}^a(\sigma)=-\left[\dot{S}^j S^a_j\right]_{S(x)=\sigma} 
\ee
and defining $P_j(x)$ implicitly through
\be \label{2.68}
\tilde{P}_a=-\left[\frac{P T_{,a}+P_j S^j_{,a}}{\left|\det(\partial S/\partial
x)\right|}
\right]_{S(x)=\sigma}
\ee
we find that $S_{\rm extended}$ precisely turns into the dust action on ${\cal
X}$ with
the second class constraints eliminated\footnote{The metric field has to be
pulled
back by the dynamical spatial diffeomorphism, as well. For details, see next
section.}.

\section{Relational Observables and Physical Hamiltonian}
\label{s3}

In this section we present an explicit prescription for constructing
gauge invariant completions of arbitrary phase space functions.
The construction is non -- perturbative and technically involved,
but the physical picture behind it will become crystal clear.
Furthermore, the formal expressions are only required to establish
 certain properties of the construction, but are not required for the
calculation
of physical properties. This is a great strength of the
relational formalism.

Let us summarise the situation.
After having solved the second class
constraints and having identified lapse and shift fields as Lagrange
multipliers, we are left with the following canonical pairs
\be \label{3.1}
(q_{ab},p^{ab}),\;\;(\xi,\pi),\;(T,P),\;(S^j,P_j) \; ,
\ee
subject to the following first class constraints
\ba \label{3.2}
c_a^{\rm tot} &=& c_a + c_a^{\rm dust}\;, \hspace{0.5cm} c_a^{\rm dust} =P \;
T_{,a}+P_j \; S^j_{\; ,a} 
\nonumber\\
c^{\rm tot} &=& c + c^{\rm dust}\;, \hspace{0.5cm} c^{\rm
dust}=-\sqrt{P^2+q^{ab}
c_a^{\rm dust} c_b^{\rm dust}}
\ea
where $c_a,\;c$ are independent of the dust variables
$\left\{T,\;P,\;S^j,P_j\right\}$.
We already used $P<0$.

As explained in \ref{s2.3}, we aim at deparametrisation of the
theory and therefore solve (\ref{3.2}) for the dust momenta, leading to
the equivalent form of the constraints
\ba \label{3.3}
\tilde{c}^{\rm tot} &=& P+h\;, \hspace{0.5cm} h=\sqrt{c^2-q^{ab} c_a c_b} 
\nonumber\\
\tilde{c}_j^{\rm tot} &=& P_j+h_j\;, \hspace{0.5cm} h_j=S^a_j\; \left(-h
T_{,a}+c_a\right)
\ea
with $S^a_j S_{\;,a}^k =\delta_j^k,\;\; S^a_j \; S_{\;,b}^j =\delta_b^a$,
hence $S^a_j$ is the inverse of $S^j_{\; ,a}$ (assuming, as before, that
$S:\;{\cal X}\to
{\cal S}$ is a diffeomorphism). These constraints are mutually
Poisson commuting\footnote{One can either prove this by direct calculation,
or one uses the following simple argument: the Poisson bracket between
the constraints must be proportional to a linear combination of constraints,
because the constraint algebra is first class. Since the constraints are linear
in the dust momenta,
the result of the Poisson bracket calculation no longer depends on them.
Therefore, the coefficients of proportionality must vanish.}.
However, only $\tilde{c}^{\rm tot}$ is
in deparametrised form (i.e. $h$ is independent of $T,S^j$), but
$\tilde{c}^{\rm tot}_j$
is not. In particular, we can only conclude that the $h(x)$ are mutually
Poisson commuting. Still, this will be enough for our purposes\footnote{
In what follows we will drop the tilde in noting the constraints
for notational simplicity.}.

Following the works~\cite{8,5,6}, we describe the construction of fully gauge
invariant completions
of phase space functions. Consider the smeared constraint
\be \label{3.4}
K_\beta\equiv \int_{{\cal X}} \;{\rm d}^3x\; \left[\beta(x) c^{\rm
tot}(x)+\beta^j(x)
c^{\rm tot}_j(x)\right] 
\ee
where $\beta(x),\;\beta^j(x)$ are phase space independent smearing functions in
the range
of $T(x),\;S^j(x)$.
Under a gauge transformation generated by this constraint, an
arbitrary phase space function $f$ is mapped to:
\be \label{3.5}
\alpha_{\beta}(f)\equiv \sum_{n=0}^\infty \; \frac{1}{n!}\;\{K_\beta,f\}_{(n)}\;
.
\ee
The fully gauge invariant completion of $f$ is given by
\be \label{3.6}
O_f[\tau,\sigma]\equiv\Big[\alpha_\beta(f)\Big]_{
\!\!\!\!\!\beta\to\tau-T\atop\beta^j\to\sigma^j-S^j} \; .
\ee
Here, the functions $\tau(x),\;\sigma^j(x)$ are also in the range
of $T(x),\;S^j(x)$, respectively\footnote{We denote the {\sl
functional
dependence} of (\ref{3.6}) on the functions $\tau(x),\;\sigma^j(x)$ by
square brackets. Below we show that it is sufficient to choose
those functions to be constant and replace the square
brackets by round ones for notational convenience.}.
It is important to {\sl first}
calculate the Poisson brackets appearing in (\ref{3.5}) with the phase space
independent functions
$\beta,\beta^j$, and {\sl afterwards} to replace them with the phase space
dependent
functions $\tau-T,\;\sigma^j-S^j$, respectively.
 This connection can be established based on the gauge transformation
properties of $T, S^j$:
$\alpha_\beta(T)=T+\beta,\; \alpha_\beta(S^j)=S^j+\beta^j$. Hence,
\be \label{3.7}
O_f[\tau,\sigma] \equiv
\Big[\alpha_\beta(f)\Big]_{\!\!\!\!\alpha_\beta(T)=\tau\atop\alpha_\beta(S^j)
=\sigma^j} \; .
\ee
Indeed, (\ref{3.7}) motivates the following interpretation:
$O_f[\tau,\sigma]$ is the gauge invariant completion of  $f$,
which in the gauge $T=\tau,\;S^j=\sigma^j$ takes the value $f$. This is not the
only interpretation we entertain a different one below.

For the purpose of this paper it suffices to consider the infinite series
appearing in the gauge invariant completions as expressions useful for
formal manipulations. There is no need to actually calculate these series for
any physical problem.

Further important properties~\cite{27,14} of the completion are:
\be \label{3.8}
\{O_f[\tau,\sigma],O_{f'}[\tau,\sigma]\}=\{O_f[\tau,\sigma],O_{f^{\prime}}[\tau,
\sigma]\}^*=O_{\{f,f'\}^*}[\tau,\sigma]
\ee
\be \label{3.10}
O_{f+f'}[\tau,\sigma]=O_f[\tau,\sigma]+O_{f'}[\tau,\sigma]\;, \hspace{0.5cm}
O_{f\cdot f'}[\tau,\sigma]=O_f[\tau,\sigma]\cdot O_{f'}[\tau,\sigma]\; .
\ee
Here, $\{.,.\}^\ast$ is the Dirac bracket\footnote{
For completeness, we note the definition of the Dirac bracket:
\be \label{3.9}
\{f,f'\}^\ast\equiv \{f,f'\}-\int_{{\cal X}} \; {\rm d}^3x\;\sum_{\mu=0}^3
\left[\{f,c^{\rm tot}_\mu(x)\}\{f',S^\mu(x)\}
-\{f',c^{\rm tot}_\mu(x)\}\{f,S^\mu(x)\}\right]
\ee
where $c_0^{\rm tot}\equiv c^{\rm tot},\;S^0:=T$.
The Dirac bracket is antisymmetric and
$\{f,c^{\rm tot}_\mu(x)\}^\ast=\{f,S^\mu(x)\}^\ast=0$ everywhere.
This follows from the fact that
$c^{\rm tot}_\mu(x)$ and $S^\mu(x)$ are mutually Poisson commuting, and that
$\{c^{\rm tot}_\mu(x),S^\nu(y)\}=\delta(x,y)\delta_\mu^\nu$.}
\cite{1} associated with the constraints
and the {\sl gauge fixing functions} $T,S^j$.
Relations (\ref{3.8}) and (\ref{3.10})
show that the map $f\mapsto O_f[\tau,\sigma]$ is a Poisson homomorphism of
the algebra of functions on phase space with pointwise multiplication,
equipped with the Dirac
bracket\footnote{The Dirac bracket coincides with the Poisson bracket
on gauge invariant functions. It is degenerate, since it
annihilates
coinstraints and {\sl gauge fixing functions}. Hence, it defines only a Poisson
structure, but not a symplectic structure on the full phase space.} as Poisson
structure.

In particular, for a general functional
$f=f[q_{ab}(x),P^{ab}(x),\xi(x),\pi(x),T(x),P(x),S^j(x),P_j(x)]$ the
following useful identity holds:
\be \label{3.11}
O_f=f\big[O_{q_{ab}(x)},O_{P^{ab}(x)},O_{\xi(x)},O_{\pi(x)},
O_{T(x)},O_{P(x)},O_{S^j(x)},O_{P_j(x)}\big](\tau,\sigma) \; .
\ee
This has important consequences: (\ref{3.11}) ensures that it suffices
to know the completions of the elementary phase space variables.
In fact, we are only interested in those functions that are independent
of the dust variables $\left\{T,S^j,P,P_j\right\}$.
The reason for this is that, first of all, $P,P_j$ are expressible in terms
of all other variables on the constraint surface.
Alternatively, since the constraints are
mutually Poisson commuting, we have
\be \label{3.12}
O_{P(x)}=O_{c^{\rm tot}(x)}+O_{h(x)}=c^{\rm tot}(x)+O_{h(x)}\;, \hspace{0.5cm}
O_{P_j(x)}=O_{c^{\rm tot}_j(x)}+O_{h_j(x)}=c^{\rm tot}_j(x)+O_{h_j(x)} \; .
\ee
Hence, these functions are known, once we know the completion of the
remaining variables. Secondly, 
\be \label{3.13}
O_{T(x)}[\tau,\sigma]=\tau(x)\;, \hspace{0.5cm}
O_{S^j(x)}[\tau,\sigma]=\sigma^j(x) \
\ee
are phase space independent. Thus, the only
interesting variables to consider are $\left\{q_{ab},P^{ab},\xi,\pi\right\}$.

In what follows we consider only dust -- independent functions $f$.
For those
(\ref{3.8}) simplifies to
\be \label{3.14}
\{O_f[\tau,\sigma],O_{f'}[\tau,\sigma]\}=O_{\{f,f'\}}[\tau,\sigma]\; .
\ee
Equations
(\ref{3.10}) and (\ref{3.8}) imply that $f\mapsto O_f[\tau,\sigma]$ is a
{\sl Poisson automorphism}
of the Poisson subalgebra of functions that do not depend on the dust
variables with the ordinary Poisson bracket as Poisson structure. This will be
absolutely crucial for all what follows.

Further useful properties of the completion are:
\be \label{3.15}
O_f[\tau,\sigma]=O^{(2)}_{O^{(1)}_f[\sigma]}[\tau]
\ee
where (recall (\ref{3.4}), (\ref{3.5}) and (\ref{3.6}))
\ba \label{3.16}
O^{(1)}_f[\sigma] &=& \Big[\alpha_\beta(f)\Big]_{\!\!\!\!\!\!\!\!\!\!\beta\to
0\atop\beta^j\to
\sigma^j-S^j}
\nonumber\\
O^{(2)}_f[\tau] &=& [\alpha_\beta(f)]_{\beta\to \tau-T\atop \!\!\!\!\!\beta^j\to
0} \; .
\ea
This follows from the fact that the constraints are mutually Poisson
commuting and $\{c^{\rm tot}(x),S^j(y)\}=0$.
The important consequence of (\ref{3.15}) is that we can accomplish
full gauge invariance in two stages: we establish first
invariance under the action of the spatial diffeomorphism constraint,
and afterwards achieve invariance with respect to gauge transformations
generated by the Hamiltonian
constraint.
This holds even under more general circumstances~\cite{6}, i.e.~when
the constraints can not be deparametrised.

\subsection{Implementing spatial diffeomorphism invariance}
\label{s3.1}

Keeping the physical interpretation of the completion in mind,
the map $f\mapsto O^{(1)}_f[\sigma]$ can be worked out explicitly.
In the first stage of the construction, the corresponding smeared constraint
reads
\be
K_\beta=\int_{{\cal X}} \; {\rm d}^3x\; \beta^j(x)\;c^{\rm tot}_j(x) \; .
\ee
Given a phase space function $f$, its completion $O^{(1)}_f$ with respect to
gauge transformations
generated by $K_\beta$ becomes
\ba \label{3.17}
O^{(1)}_f[\sigma] &=& \sum_{n=0}^\infty \frac{1}{n!}
\left[\{K_\beta,f\}_{(n)}\right]_{\beta^j\to \sigma^j-S^j}
\nonumber\\
&=& f+\sum_{n=1}^\infty\;\frac{1}{n!}
\int_{{\cal X}} \; {\rm d}^3x_1\; [\sigma^{j_1}(x_1)-S^{j_1}(x_1)]\;\dots\;
\int_{{\cal X}} \; {\rm d}^3x_n\; [\sigma^{j_n}(x_n)-S^{j_n}(x_n)]\;
\nonumber\\
&& \times\; \left\{c^{\rm tot}_{j_1}(x_1),\left\{c^{\rm tot}_{j_2}(x_2),\dots,
\left\{c^{\rm tot}_{j_n}(x_n),f\right\}\dots\right\}\right\}\; .
\ea
Let us begin with $f=\xi(x)$. We claim that
\be \label{3.18}
\left\{K_\beta,\xi(x)\right\}_{(n)} = \left[\beta^{j_1}\dots\beta^{j_n}\;
v_{j_1}\dots v_{j_n}\cdot \xi\right](x) 
\ee
where $v_j$ is the vector field defined by
\be \label{3.18a}
v_j\cdot \xi(x):=S^a_j(x)\; \xi_{,a}(x) \; .
\ee
In fact the vectors $v_j$ are mutually commuting.
\ba \label{3.18b}
[v_j,v_k]
& = & S^a_j S^b_{k,a} \partial_b-j\leftrightarrow k
\nonumber\\
& = & -S^a_j S^b_l S^l_{,ca} S^c_k  \partial_b-j\leftrightarrow k
\nonumber\\
& = &  -S^a_j S^b_l S^l_{,ac} S^c_k  \partial_b-j\leftrightarrow k
\nonumber\\
& = &  S^a_j S^b_{l,c} S^l_{,a} S^c_k  \partial_b-j\leftrightarrow k
\nonumber\\
& = &  S^b_{j,c}  S^c_k  \partial_b-j\leftrightarrow k
\nonumber\\
& = & S^a_k S^b_{j,a}   \partial_b-j\leftrightarrow k\nonumber\\
&=&0
\ea
To prove (\ref{3.18}) by induction over $n$ we need 
\be \label{3.19}
\{K_\beta,S^a_j(x)\} = -S^a_k(x) \{K_\beta, S^k_{,b}(x)\} S^b_j(x)=
-S^a_k(x) \beta^k_{,b}(x) S^b_j(x)=-[v_j\cdot \beta^k] S^a_k
\ee
For $n=1$ we have 
\be \label{3.20}
\{K_\beta,\xi(x)\}_{(1)} = [\beta^j S^a_j \xi_{,a}](x)=\beta^j v_j\cdot 
\xi
\ee
which coincides with (\ref{3.18}). 
Suppose that (\ref{3.18}) is correct up to $n$, then
\ba \label{3.21}
\{K_\beta,\xi\}_{(n+1)} &=&
\beta^{j_1} \dots \beta^{j_n}\;\left\{K_\beta,
v_{j_1} \dots v_{j_n}\cdot \xi\right\}
\nonumber\\
&=&
\beta^{j_1} \dots \beta^{j_n}\;\left[
v_{j_1} \dots v_{j_n}\cdot \left\{K_\beta,\xi\right\}
+\sum_{l=1}^n v_{j_1} \dots v_{j_{l-1}} \left\{K_\beta, S^a_{j_l}\right\}
\partial_a
v_{j_{l+1}} \dots v_{j_n} \cdot \xi\right]
\nonumber\\
&=&
\beta^{j_1} \dots \beta^{j_n}\;\left[
v_{j_1} \dots v_{j_n}\cdot \beta^{j_{n+1}} v_{j_{n+1}} \cdot\xi
-\sum_{l=1}^n v_{j_1} \dots v_{j_{l-1}} \left[v_{j_l} \beta^{j_{n+1}}\right]
v_{j_{n+1}}
v_{j_{l+1}} \dots v_{j_n} \cdot \xi\right]
\nonumber\\
&=&
\beta^{j_1} \dots \beta^{j_n}\;\left[
v_{j_1} \dots v_{j_n}\cdot \beta^{j_{n+1}} v_{j_{n+1}} \cdot\xi
-\sum_{l=1}^n v_{j_1} \dots v_{j_{l-1}} \left[v_{j_l} \beta^{j_{n+1}}\right]
v_{j_{l+1}} \dots v_{j_{n+1}} \cdot \xi\right]
\nonumber\\
&=&
\beta^{j_1} \dots \beta^{j_n}\;\left[
v_{j_1} \dots v_{j_n}\cdot \beta^{j_{n+1}} v_{j_{n+1}} \cdot\xi
-\left(v_{j_1} \dots v_{j_n} \beta^{j_{n+1}} - \beta_{j_{n+1}} v_{j_1} \dots
v_{j_n}\right)
v_{j_{n+1}} \cdot \xi\right]
\nonumber\\
&=& \beta^{j_1} \dots \beta^{j_{n+1}}\;
v_{j_1} \dots v_{j_{n+1} }\cdot\xi \; .
\ea
where we used commutativity of the $v_j$ and the Leibniz rule. 
\\
It follows that
\be \label{3.22}
O^{(1)}_{\xi(x)}[\sigma]=\xi(x)+\sum_{n=1}^\infty \; \frac{1}{n!}\;
\left[\sigma^{j_1}(x)-S^{j_1}(x)\right] \dots
\left[\sigma^{j_n}(x)-S^{j_n}(x)\right] \;
v_{j_1} \dots v_{j_n} \cdot \xi(x) \; .
\ee
Using $v_j \cdot S^k=\delta_j^k$ and commutativity
of the $v_j$, we find with $\beta^j:=\sigma^j-S^j$ that
\ba \label{3.23}
v_k\cdot O^{(1)}_{\xi(x)}[\sigma]
&=& v_k\cdot \xi+\sum_{n=1}^\infty
\left[\frac{1}{(n-1)!} \left[v_k\cdot \beta^j\right] \beta^{j_1} \dots
\beta^{j_{n-1}} \;
v_j v_{j_1} \dots v_{j_{n-1}} \cdot \xi+\frac{1}{n!}
\beta^{j_1} \dots \beta^{j_n} \;
v_k v_{j_1} \dots v_{j_n} \cdot \xi\right]
\nonumber\\
&=& v_k\cdot \xi+\left[v_k \cdot \beta^j\right] v_j\cdot \xi
+ \sum_{n=1}^\infty
\frac{1}{n!} \beta^{j_1} \dots \beta^{j_n} \;
\left[\left[v_k\cdot \beta^j\right] v_j v_{j_1} \dots v_{j_n} \cdot \xi+
v_k v_{j_1} \dots v_{j_n} \cdot \xi\right]
\nonumber\\
&=& \sum_{n=0}^\infty
\frac{1}{(n)!} \beta^{j_1} \dots \beta^{j_n} \;
\left[v_k\cdot \sigma^j\right] v_j v_{j_1} \dots v_{j_n} \cdot \xi \; .
\ea
The interpretation of (\ref{3.22}) becomes clear for the choice
$\sigma^j(x)=\sigma^j=$const., for which
(\ref{3.23}) vanishes identically. In other words, the completion
$O^{(1)}_{\xi(x)}[\sigma]$ does not depend on $x$ at all.
Hence, for this choice of $\sigma^j$, we are free to choose
$x$ in $O^{(1)}_{\xi(x)}[\sigma]$ in order to simplify (\ref{3.22}).
Since (\ref{3.22}) is a power expansion in
$\left(\sigma^{j}(x)-S^{j}\right)(x)$,
and $S^{j}$ is a diffeomorphism, we choose $x=x_\sigma$, with $x_\sigma$ being
the unique solution
of $S^j(x)=\sigma^j$. Then\footnote{We switched to the notation
$O^{(1)}_f(\sigma)$ to indicate
the choice $\sigma^j(x)=\sigma^j=$const.},
\be \label{3.24}
O^{(1)}_{\xi(x)}(\sigma)=\xi(x_\sigma)=\left[\xi(x)\right]_{S^j(x)=\sigma^j} \;
.
\ee
The completion $O^{(1)}_{\xi(x)}(\sigma)$ of $\xi(x)$ has also a simple
integral representation:
\be \label{3.25}
O^{(1)}_{\xi(x)}(\sigma)=\int_{{\cal X}} \; {\rm d}^3x\; \left|\det(\partial
S(x)/\partial x)\right|\;\delta\left(S(x),\sigma\right)\;\xi(x) \; .
\ee

The significance of choosing $\sigma^j=$const is the following:
Clearly, the choice $\sigma^j(x)=$ const. is not in the range of
$S^j(x)$, which is supposed to be a diffeomorphism.
Thus, the interpretation of $O^{(1)}_f(\sigma)$ as
the value of $f$ in the gauge $S^j=\sigma^j$ is obsolete.
However, given a function $\sigma^j(x)$, instead of solving
$S^j(x)=\sigma^j(x)$ for the values of the function $S^j$
for all $x$, we could solve it for $x$, while keeping
the function $S^j$ arbitrary. This is the appropriate interpretation
of $O^{(1)}_f(\sigma)$.
This is possible because $O^{(1)}_f[\sigma]$
is (at least formally) gauge invariant, whether or not $S^j=\sigma^j$ is a
good choice of gauge. It is fully sufficient to do this because, as shown
in \cite{3} and as we will show in appendix \ref{s3.2}, the partially reduced
phase space (with respect to the spatial diffeomorphism constraint)
is completely determined by the $O^{(1)}_f(\sigma)$, hence the
$O^{(1)}_f[\sigma]$ must be hugely redundant.

We can now compute the spatially diffeomorphism invariant extensions
for the remaining phase space variables without any additional effort, by
switching
first to variables which are spatial scalars on $\cal X$, using
$J:=\det(\partial S/\partial x)$, which we assume to be positive
(orientation preserving diffeomorphism):
\be \label{3.26}
\left(\xi,\pi/J\right)\;, \; \left(T,P/J\right)\;, \;
\left(q_{jk}\equiv q_{ab}\; S^a_j S^b_k\;,
p^{jk}\equiv S^j_{\;,a} S^k_{\; ,b}\; p^{ab}/J\right)\; .
\ee
The image of these quantities, evaluated at $x$, under the completion
$O^{(1)}_{f}(\sigma)$ simply consists in replacing $x$ by $x_\sigma$, where
$x_\sigma$ solves $S^j(x)=\sigma^j$, just as in (\ref{3.24}).
The scalars (\ref{3.26}) on ${\cal X}$ are the
pull backs of the original tensor (densities) under the diffeomorphism
$\sigma \mapsto x_\sigma$ evaluated at $\sigma$. Thus, they are tensor
(densities) of the same type, but live now on the dust space manifold 
$\cal S$.

This statement sounds
contradictory because of the following subtlety: We have
e.g. the three quantities
$P(x),\;\tilde{P}(x)=P(x)/J(x),\;\tilde{P}(\sigma)=\tilde{P}(x_\sigma)$.
On $\cal X$, $P(x)$ is a scalar density while $\tilde{P}(x)$ is a
scalar. Pulling back $P(x)$ to ${\cal S}=S({\cal X})$ by the
diffeomorphism $\sigma\mapsto S^{-1}(\sigma)$ results in
$\tilde{P}(\sigma)$. But pulling back $\tilde{P}(x)$ back to $\cal S$
results in the {\it same} quantity $\tilde{P}(\sigma)$. Since a
diffeomorphism does not change the density weight, we would get the
contradiction that $\tilde{P}(\sigma)$ has both density weights zero and
one on $\cal S$. The resolution of the puzzle is that what determines
the density weight of $P(x)$ on $\cal X$ is its transformation behaviour
under canonical transformations generated by the
total spatial diffeomorphism constraint $c_a^{\rm tot}=c_a^{\rm dust}+c_a$ where
$c_a^{\rm dust},\;c_a$ are the dust and non dust contributions respectively.
After the reduction of $c_a^{\rm tot}$, what determines the density weight
of $\tilde{P}(\sigma)$ on $\cal S$ is its transformation behaviour under
$([c_a+P
T_{,a}]S^a_j/J)(x_\sigma)=\tilde{c}_j(\sigma)+\tilde{P}(\sigma)
\tilde{T}_{,j}(\sigma)$ and this shows that $\tilde{P}(\sigma)$ has
density weight one\footnote{
In order to avoid confusion of the reader we mention that any quantity 
$f$ 
on $\cal X$ which has positive density weight is mapped to zero under
$f\mapsto O^{(1)}_f(\sigma)$. Let us again consider the example $f=P$.
We have $\tilde{P}(\sigma)=P(x_\sigma)\det(\partial 
S^{-1}(\sigma)/\partial \sigma)$ which is perfectly finite. However
by the Poisson automorphism formula 
$O^{(1)}_{P(x)}=
O^{(1)}_{J(x)\; \tilde{P}(x)}= 
O^{(1)}_{J(x)}\;\tilde{P}(\sigma)
=\det(\partial \sigma/\partial x)\;\tilde{P}(\sigma)=0$
since $\sigma=$const.}.\\
\\
We will denote the images under $f\mapsto O^{(1)}_f(\sigma)$ by
\be \label{3.27}
\left(\tilde{\xi}(\sigma),\tilde{\pi}(\sigma)\right)\;, \; 
\left(\tilde{T}(\sigma),
\tilde{P}(\sigma)\right)\;, \; \left(\tilde{q}_{ij}(\sigma)
\tilde{p}^{ij}(\sigma)\right)\; .
\ee
In appendix \ref{s3.2} we show that the quantities (\ref{3.27}) can be also
obtained through 
symplectic reduction which is an alternative method to show that the pairs in
(\ref{3.27}) are conjugate and as it was done in \cite{3}.

\subsection{Implementing invariance with respect to the Hamiltonian
constraint}
\label{s3.3}

Having completed the elementary phase space variables with respect
to the spatial diffeomorphism constraint, it remains to render those
variables invariant under the action of the Hamiltonian constraint.
This amounts to calculate the image of those variables under the map
$f\mapsto O^{(2)}_f[\sigma]$, for any $f$ in (\ref{3.27}).
For $f$ independent of $T,P$, the completion of $f$ with respect to
the Hamiltonian constraint is given by
\be \label{3.32}
O^{(2)}_f[\tau]=\sum_{n=0}^\infty\;
\frac{1}{n!}\;\left\{h[\tau],f\right\}_{(n)}\; , \hspace{0.5cm}
h[\tau]=\int_{{\cal X}} \; {\rm d}^3x\; \left(\tau(x)-T(x)\right)\;h(x) \; .
\ee
Only if we choose $\tau(x)=\tau={\rm const.}$ (\ref{3.32}) is invariant under
diffeomorphisms. Hence we choose $\tau(x)=\tau={\rm const.}$ 
which allows to rewrite (\ref{3.32}) entirely in terms of
the variables (\ref{3.27}).
As a reminder of this choice, we denote the completion by
$O^{(2)}_{f}(\tau)$.
In this case (\ref{3.32}) can be written as
\be \label{3.33}
O^{(2)}_f(\tau)=\sum_{n=0}^\infty\;
\frac{1}{n!}\;\{\tilde{h}(\tau),f\}_{(n)}\;,\hspace{0.5cm}
\tilde{h}(\tau)=\int_{{\cal S}} {\rm d}^3\sigma\;
(\tau-\tilde{T}(\sigma))\;\tilde{h}(\sigma) \;
\ee
with $\tilde{h}(\sigma)$ denoting the image of $h(x)$ under the
replacement\footnote{
The proof of this statement is based on the fact that the replacement
corresponds to a diffeomorphism and that $h(\tau)$ is the integral
of a scalar density of weight on, for $\tau=$const.}
of $\left\{\xi(x),\pi(x),q_{ab}(x),p^{ab}(x)\right\}$ by \\
$\left\{\tilde{\xi}(\sigma),\tilde{\pi}(\sigma),\tilde{q}_{jk}(\sigma),
\tilde{p}^{jk}(\sigma)\right\}$, respectively. Explicitly, denoting
\be \label{3.34} \tilde{c}(\sigma)\equiv
\left[\frac{c(x)}{J(x)}\right]_{S(x)=\sigma}\hspace{0.5cm}
\tilde{c}_j(\sigma)\equiv
\left[\frac{c_j(x)}{J(x)}\right]_{S(x)=\sigma}\; , \ee where, as
before, $c_j(x)=S^a_j(x) \; c_a(x)$, we find \be \label{3.35}
\tilde{h}(\sigma)=\sqrt{\tilde{c}^2-\tilde{q}^{jk} \; \tilde{c}_j
\tilde{c}_k}(\sigma) \; . \ee It is easy to see that \be
\label{3.36}
\frac{d}{d\tau}\;O^{(2)}_f(\tau)=\{\HF,O^{(2)}_f(\tau)\} \ee
with \be \label{3.37} \HF:=\int_{{\cal S}}\; d^3\sigma\;
\tilde{h}(\sigma) \ee is the {\sl physical Hamiltonian} (not
Hamiltonian density) of the deparametrised system.

We denote the fully gauge invariant completions of the Hamiltonian constraint,
the spatial diffeomorphism
constraints\footnote{Explicit expression for the constraints in terms of the
fully gauge invariant phase space variables are given in the next
section, see (\ref{3.54}).}
and the physical Hamilton density,
respectively, as
\ba \label{3.39}
C(\tau,\sigma)&\equiv& O^{(2)}_{\tilde{c}(\sigma)}(\tau) \hspace{0.5cm}
C_j(\tau,\sigma)\equiv O^{(2)}_{\tilde{c}_j(\sigma)}(\tau)\; ,\nonumber \\
H(\tau,\sigma)&\equiv& O^{(2)}_{\tilde{h}(\sigma)}(\tau)\; .
\ea
It is worth emphasising again that $H(\tau,\sigma)$ is the physical energy
density associated to the physical Hamiltonian when the dust fields are considered as clocks of the system.
The fully gauge invariant completions of the phase space variables
for matter and gravity are denoted by
\ba\label{3.38}
\Xi(\tau,\sigma)&\equiv& O^{(2)}_{\tilde{\xi}(\sigma)}(\tau) \hspace{0.5cm}
\Pi(\tau,\sigma)\equiv O^{(2)}_{\tilde{\pi}(\sigma)}(\tau),\;,\nonumber \\
Q_{ij}(\tau,\sigma)&\equiv& O^{(2)}_{\tilde{q}_{ij}(\sigma)}(\tau)\hspace{0.5cm}
P^{ij}(\tau,\sigma)\equiv O^{(2)}_{\tilde{p}^{ij}(\sigma)}(\tau)\;.
\ea
The matter scalar field $\Xi(\tau,\sigma)$ and its conjugate momentum
$\Pi(\tau,\sigma)$ are observable quantities since gauge invariant. The same applies to the 
three-metric $Q_{ij}(\tau,\sigma)$ and its canonical momentum field
$P^{ij}(\tau,\sigma)$. Moreover, the completion is non -- perturbative,
i.e.~full non-Abelian gauge invariance has been accomplished.

\subsection{Constants of the physical Motion}
\label{s3.4}

In the previous section we successfully constructed fully gauge invariant
quantities for a specific deparmetrising system. In some sense,
the construction frees the true degrees of freedom from the constraints,
replacing them by conservation laws which govern the physical motion
of observable quantities. Indeed, we have the following first integrals
of physical motion ({\sl conservation laws}):
\be \label{3.40}
\frac{d}{d\tau} C_j(\tau,\sigma)=0\;, \hspace{0.5cm}
\frac{d}{d\tau} H(\tau,\sigma)=0\; .
\ee
These equations express invariance under the physical evolution generated
by $\HF$, as opposed to gauge invariance. The functions
$C_j,\;H$, representing physical three -- momentum and energy,
{\sl are} already gauge invariant.

We proceed the  proof of (\ref{3.40}). Recall that the original constraints
$c^{\rm tot}(x)\;, c^{\rm tot}_j(x)$ are mutually
Poisson commuting. Using (\ref{3.3}), this means in particular,
\be \label{3.41}
\left\{c^{\rm tot}(x),c^{\rm tot}(y)\right\} =
\left\{P(x)+h(x),P(y)+h(y)\right\}=\left\{h(x),h(y)\right\}=0 
\ee
where we used that the $P(x)$ are mutually Poisson commuting and that
$h(x)$ is independent of the dust variables.
Next, consider the smeared spatial diffeomorphism generator
\be
c(u)\equiv \int_{{\cal X}} \; {\rm d}^3x \;
u^a(x) \; c_a(x) \;.
\ee
The smeared constraint acts on $h(y)$ as it should,
\be
\left\{c(u),h(y)\right\}=\left[u^a h\right]_{,a}(y) 
\ee
or, after functional differentiation with respect
to the smearing functions $u^a$:
\be \label{3.42}
\left\{c_a(x),h(y)\right\}=\partial_{y^a} \left(\delta(x,y) h(y)\right) \; .
\ee
This follows from the properties of
$c_a$, generating spatial diffeomorphims on the matter and gravity variables,
and $h$, being a scalar density of weight one and only depending on
the non -- dust variables.
Furthermore, the spatial diffeomorphsims form an algebra with
$\{c(u),c(u')\}=c([u',u])$. From this follows again by functional
differentiation
\be \label{3.43}
\left\{c_a(x),c_b(y)\right\}=
\left[\partial_{y^b} \delta(x,y)\right] c_a(y)-\left[\partial_{x^a}
\delta(y,x)\right] c_b(x)\;.
\ee
Let us investigate the implications of (\ref{3.41}--\ref{3.43})
for
\ba \label{3.44}
\tilde{h}(\sigma) &=& \left[\frac{h(x)}{J(x)}\right]_{S(x)=\sigma}=
\int_{{\cal X}}\; {\rm d}^3x\; \delta\left(S(x),\sigma\right)\; h(x) 
\nonumber\\
\tilde{c}_j(\sigma) &=& \left[\frac{c_a(x)\;
S^a_j(x)}{J(x)}\right]_{S(x)=\sigma}= \int_{{\cal X}}\; {\rm d}^3x\;
\delta\left(S(x),\sigma\right)\; S^a_j(x) c_a(x) \; . \ea First of
all, \be \label{3.45}
\left\{\tilde{h}(\sigma),\tilde{h}(\sigma')\right\}= \int_{{\cal X}}
\; {\rm d}^3x \; \int_{{\cal X}} \; {\rm d}^3y \;
\delta\left(S(x),\sigma\right) \; \delta\left(S(y),\sigma'\right) \;
\left\{h(x),h(y)\right\}=0 \ee where we used that the $S^j(x)$
are mutually commuting, as well as with the $h(y)$. Second, denoting
the pullback of the smeared diffeomorphism generator with
$\tilde{c}(\tilde{u})$ for some smearing functions
$\tilde{u}^j(\sigma)$, we have \ba \label{3.46}
\left\{\tilde{c}(\tilde{u}),\tilde{h}(\sigma')\right\} & =&
\int_{{\cal S}}\; {\rm d}^3\sigma \; \tilde{u}^j(\sigma)\;
\int_{{\cal X}} \; {\rm d}^3x \; \int_{{\cal X}} \; {\rm d}^3y \;
\delta\left(S(x),\sigma) \; \delta(S(y),\sigma'\right) \; S^a_j(x)\;
\left\{c_a(x),h(y)\right\}
\nonumber\\
&=& \int_{{\cal S}}\; {\rm d}^3\sigma \; \tilde{u}^j(\sigma)\;
\int_{{\cal X}} \; {\rm d}^3x \; \int_{{\cal X}} \; {\rm d}^3y \;
\delta\left(S(x),\sigma\right) \; \delta\left(S(y),\sigma'\right) \; S^a_j(x)\;
\partial_{y^a} \left(\delta(x,y) h(y)\right)
\nonumber\\
&=& - \int_{{\cal S}}\; {\rm d}^3\sigma \; \tilde{u}^j(\sigma)\;
\int_{{\cal X}} \; {\rm d}^3x \;
\delta\left(S(x),\sigma\right) \; \left[\partial_{x^a}
\delta\left(S(x),\sigma'\right)\right] \;
S^a_j(x)h(x)\;
\nonumber\\
& =& - \int_{{\cal S}}\; {\rm d}^3\sigma \; \tilde{u}^j(\sigma)\;
\int_{{\cal X}} \; {\rm d}^3x
\delta\left(S(x),\sigma\right) \; \left[\partial_{\tilde{\sigma}^k}
\delta\left(\tilde{\sigma},\sigma'\right)\right]_{\tilde{\sigma}=S(x)} \;
S^k_{,a}(x)
S^a_j(x)\; h(x)
\nonumber\\
& =& - \int_{{\cal S}}\; {\rm d}^3\sigma \; \tilde{u}^j(\sigma)\;
\int_{{\cal S}} \; {\rm d}^3\sigma_1 \; \delta\left(\sigma_1,\sigma\right)\;
\; \left[\partial_{\sigma_1^j} \delta\left(\sigma_1,\sigma'\right)\right] \;
\left[\frac{h(x)}{J(x)}\right]_{S(x)=\sigma_1}
\nonumber\\
& =& -
\int_{{\cal S}} \; {\rm d}^3\sigma_1 \; \tilde{u}^j(\sigma_1)\;
\; \left[\delta\left(\sigma_1,\sigma'\right)\right]_{,\sigma_1^j} \;
\tilde{h}(\sigma_1)
\nonumber\\
& =&
\int_{{\cal S}} \; {\rm d}^3\sigma_1 \;
\left[\tilde{u}^j(\sigma_1)\;\tilde{h}(\sigma_1)\right]_{,\sigma_1^j}
\; \delta\left(\sigma_1,\sigma'\right)
\nonumber\\
& =&
\left[\tilde{u}^j(\sigma')\;\tilde{h}(\sigma')\right]_{,\sigma^{\prime j}} \; .
\ea
The last implication follows from
\be \label{3.47}
\tilde{c}(\tilde{u})=c(u_S)\;, \hspace{0.5cm} u^a_S(x)=S^a_j(x)
\tilde{u}^j(S(x)) 
\ee
where the vector fields $u_S$ are phase space dependent (they depend on
$S$) and using the fact that the $S^j(x)$ and $c_a(y)$ are mutually
Poisson commuting. Then,
\ba \label{3.48}
&& \left\{\tilde{c}(\tilde{u}),\tilde{c}(\tilde{u}')\right\} =
c\left(\left[u'_S,u_S\right]\right)
\nonumber\\
&=& \int_{{\cal X}}\; {\rm d}^3x\;
\left[u^{\prime b}_S(x) u^a_{S,b}(x)
-u^b_S(x) u^{\prime a}_{S,b}(x)\right] \; c_a(x)
\nonumber\\
&=& \int_{{\cal X}}\; {\rm d}^3x\;
S^b_j(x)\;
\Bigg[\tilde{u}^{\prime j} (S(x)) \left(S^a_k(x) S^l_{,b}(x)\;
\tilde{u}^k_{,l}(S(x))-S^a_l(x) S^l_{,cb}(x) S^c_k(x)
\tilde{u}^k(S(x))\right)
\nonumber\\
&& -\tilde{u}^j(S(x)) \left(S^a_k(x) \tilde{u}^{\prime k}_{,l}(S(x))
-S^a_l(x) S^l_{,cb}(x) S^c_k(x)
\tilde{u}^{\prime k}(S(x))\right)\Bigg] \; c_a(x)
\nonumber\\
&=& \int_{{\cal X}}\; {\rm d}^3x\;
\left[\tilde{u}^{\prime j}(S(x)) \tilde{u}^k_{,j}(S(x))
-\tilde{u}^j(S(x))  \tilde{u}^{\prime k}_{,j}(S(x))\right]\;
c_a(x) S^a_k(x)
\nonumber\\
& =&
\int_{{\cal S}}\; {\rm d}^3\sigma\;
\int_{{\cal X}}\; {\rm d}^3x\; \delta\left(S(x),\sigma\right)
\left[\tilde{u}^{\prime j}(\sigma) \tilde{u}^k_{,j}(\sigma)
-\tilde{u}^j(\sigma)  \tilde{u}^{\prime k}_{,j}(\sigma)\right]\;
c_a(x) S^a_k(x)
\nonumber\\
&=& \tilde{c}\left(\left[\tilde{u}',\tilde{u}\right]\right) \; .
\ea
Hence, equations (\ref{3.41}--\ref{3.43}) are exactly reproduced 
by (\ref{3.45}), (\ref{3.46}) and (\ref{3.48}).

We can now easily finish the proof of (\ref{3.40}).
In (\ref{3.33}) we introduced $\tilde{h}(\tau)$. From (\ref{3.45})
follows that
\be \label{3.49}
\big\{\tilde{h}(\tau),\tilde{h}(\sigma)\big\}=0 \; .
\ee
This implies in particular that
\be \label{3.50}
\tilde{h}(\sigma)=H(\sigma)=O^{(2)}_{\tilde{h}(\sigma)}(\tau)
\ee
is already an observable quantity\footnote{Note, however, although
$H(\sigma)=\tilde{h}(\sigma)$, this is not true for the corresponding spatial
diffeomorphism constraints, $\tilde{c}_j(\sigma)\not=C_j(\sigma)$!}.
Hence, from the definition
of $\HF$ and (\ref{3.45}) we find $\{\HF,\tilde{h}(\sigma)\}=0$.
Furthermore,
\be \label{3.51}
\left\{\HF,C_j(\tau,\sigma)\right\}
=\{O^{(2)}_{\HF}(\tau),O^{(2)}_{\tilde{c}_j(\sigma)}(\tau)\}
=O^{(2)}_{\{\HF, \tilde{c}_j(\sigma)\}}(\tau)
=0 \; .
\ee

Alternatively, a more direct way to understand this result is to make
use of the series representation (\ref{3.33}) and of
\be
\tilde{h}(\tau)=\tau
\HF-\tilde{h}[\tilde{T}]\hspace{0.5cm}
\tilde{h}[\tilde{T}]=
\int_{{\cal S}}\; {\rm d}^3\sigma\; \tilde{T}(\sigma)\;\tilde{h}(\sigma)
\; .
\ee
Since the Hamiltonian vector fields $X_1,\;X_2$ of $\HF$ and
$\tilde{h}[\tilde{T}]$, respectively,
are commuting, we may write for (\ref{3.33})
\ba \label{3.52}
C_j(\sigma,\tau)
&=& \exp(\tau X_1-X_2)\cdot \tilde{c}_j(\sigma)=
\exp(-X_2)\cdot [\exp(\tau X_1) \cdot \tilde{c}_j(\sigma)]
\nonumber\\
&=& \exp(-X_2)\cdot \tilde{c}_j(\sigma)=
\sum_{n=0}^\infty \;\frac{(-1)^n}{n!}\;
\{\tilde{h}[\tilde{T}],\tilde{\sigma}_j(\sigma)\}_{(n)} \;
\ea
which is clearly $\tau$ -- independent.

We end this section by giving an explicit expressions for
the physical Hamiltonian in terms of purely gauge invariant quantities:
\be \label{3.53}
H(\sigma)=\sqrt{C(\tau,\sigma)^2-Q^{jk}(\tau,\sigma)\;C_j(\sigma)
\;C_k(\sigma)}
\ee
Note that $C,\;Q^{jk}$ are not independent of the physical time $\tau$.
Of course, $C(\tau,\sigma),\; C_j(\sigma)$ are
obtained from $\tilde{c}(\sigma),\; \tilde{c}_j(\sigma)$ simply by
replacing everywhere the functional dependence on
$\left\{\tilde{\xi}(\sigma),\tilde{\pi}(\sigma),\tilde{q}_{jk}(\sigma),
p^{jk}(\sigma)\right\}$
by that on
$\left\{\Xi(\tau,\sigma),\Pi(\tau,\sigma),Q_{jk}(\tau,\sigma),
P^{jk}(\tau,\sigma)\right\}$. In greater detail,
\ba \label{3.54}
C_j(\sigma) &=& \left[-2 Q_{jk} (D_k P^{kl})+\Pi (D_j\Xi)\right](\tau,\sigma)
\nonumber\\
C(\tau,\sigma) &=&
\frac{1}{\kappa}\left[\frac{1}{\sqrt{\det(Q)}}\left(Q_{jm} Q_{kn}-\frac{1}{2}
Q_{jk}
Q_{mn}\right)P^{jk} P^{mn}- \sqrt{\det(Q)} R^{(3)}[Q] +2\Lambda
\sqrt{\det(Q)}\right](\tau,\sigma)
\nonumber\\
&&+ \frac{1}{2\lambda}\left[\frac{\Pi^2}{\sqrt{\det(Q)}}+\sqrt{\det(Q)}
\left(Q^{jk}\; (D_j\Xi)\;(D_k \Xi)+v(\Xi)\right)\right](\tau,\sigma)
\nonumber\\
&\equiv& C_{\rm geo}(\tau,\sigma)+C_{\rm matter}(\tau,\sigma) 
\ea
with $D_j$ denoting the covariant differential compatible with $Q_{jk}$.

\section{Physical Equations of Motion}
\label{s4}

In this section\footnote{
For the purposes of this section we assume that $\cal X$ and, equivalently,
$\cal S$ have no boundary. In order to allow for more general topologies,
we consider boundary terms in the next section.
The calculations
of the present section are not affected by the presence of such a boundary term,
because it only cancels the boundary term that would appear in the calculation
of
this section.}
we derive the physical evolution
of the gauge invariant functions
$\{\Xi,\;\Pi,\;Q_{jk},\;P^{jk}\}$, generated by the
true Hamiltonian $\HF$, in the first order (Hamilton) and
second order (Lagrange) formulation.
In other words, we study the true evolution of matter degrees of freedom and
gravity
with respect to the physical reference system (dust).

\subsection{First Order (Hamiltonian) Formulation}
\label{s4.1}

For a generic observable $F$, we denote\footnote{
Furthermore, for notational ease we drop the dependence on $(\tau,\sigma)$
when no confusions can arise.}
its $\tau$--derivative simply by an
overdot, $\dot{F}$. Then,
\ba \label{4.1}
\dot{F} &=& \left\{\HF,F\right\}
=\int_{{\cal S}}\;{\rm d}^3\sigma\;\{H(\sigma),F\}
\\
&=& \int_{{\cal
S}}\;{\rm
d}^3\sigma\;\frac{1}{H(\sigma)}\left(C(\sigma)\left\{C(\sigma),F\right\}
-Q^{jk}(\sigma) C_k(\sigma) \left\{C_j(\sigma),F\right\}+\frac{1}{2}
Q^{im}(\sigma)C_m(\sigma)Q^{jn}(\sigma)C_{n}(\sigma) \left\{Q_{ij}(\sigma),F\right\}\right)\;\nonumber .
\ea

Let us introduce {\sl dynamical shift} and {\sl dynamical lapse} fields by
\ba \label{4.2}
N_j &\equiv& - C_j/H 
\nonumber\\
N &\equiv& C/H = \sqrt{1+Q^{jk} N_j N_k} \; .
\ea
Notice that $N_j$ is a constant of the physical motion, but
neither are $N$ nor $N^j=Q^{jk} N_k$.
Then (\ref{4.1}) can be rewritten in the familiar looking form
\be \label{4.3}
\dot{F}
= \int_{{\cal
S}}\;{\rm d}^3\sigma\;\left(N(\sigma)\left\{C(\sigma),F\right\}+N^j(\sigma)
\left\{C_j(\sigma),F\right\}+\frac{1}{2} H(\sigma) N^i(\sigma) N^{j}(\sigma)
\left\{Q_{ij}(\sigma),F\right\}\right) \; .
\ee
The first two terms in (\ref{4.3}) are {\sl exactly the same} as those
in the gauge variant derivation of the equation of motion,
derived with respect the primary Hamiltonian
\be \label{4.4}
\HF_{\rm primary}(N,\vec{N})=\int_{{\cal S}} \; {\rm d}^3\sigma\;
\left(N(\sigma)
C(\sigma)+N^j(\sigma) C_j(\sigma)\right) \;.
\ee
Here, $N,\;N^j$ are viewed as phase space independent functions. The third
term in (\ref{4.3}), on the other hand, is a genuine correction to the
gauge variant formalism. However, it enters only in the physical evolution
equation
of $P^{jk}$. Hence,
\ba \label{4.5}
\dot{\Xi} &=&  \frac{N}{\sqrt{\det(Q)}}\Pi+{\cal L}_{\vec{N}} \Xi 
\nonumber\\
\dot{\Pi} &=&\partial_j \left[N \sqrt{\det(Q)}\; Q^{jk} \Xi_{,k}\right]
-\frac{N}{2}\sqrt{\det(Q)}\; v'(\Xi)
+ {\cal L}_{\vec{N}} \Pi 
\nonumber\\
\dot{Q}_{jk} &=& \frac{2N}{\sqrt{\det(Q)}}\; G_{jkmn} \;
P^{mn}+ \left({\cal L}_{\vec{N}} Q\right)_{jk} 
\nonumber\\
\dot{P}^{jk} &=&
 N\left[-\frac{Q_{mn}}{\sqrt{\det{Q}}}\left(2 P^{jm} P^{kn}-P^{jk}
P^{mn}\right)
+\frac{\kappa}{2}Q^{jk}\;C -
\Q\;Q^{jk}\left(2\Lambda +\frac{\kappa}{2\lambda}\left(\Xi^{,m}\Xi_{,m}
+v(\Xi)\right)\right)\right]
\nonumber\\
&& +\Q \left[G^{-1}\right]^{jkmn}\left((D_mD_n N)-NR_{mn}[Q]\right)
+\frac{\kappa}{2\lambda}N\Q\;\Xi^{,j}\Xi^{,k}
\nonumber\\
&& -\frac{\kappa}{2}\; H \; Q^{jm} Q^{kn} N_m N_n
 +({\cal L}_{\vec{N}} P)^{jk} 
\ea
with ${\cal L}_{\vec{N}}$ denoting the Lie derivative\footnote{
For the explicit calculation of the Lie derivative it is
important to note that $\Pi\;, P^{jk}$ are tensor densities
of weight one in dust space.}
with respect
to the vector field $\vec{N}$ with components $N^j=Q^{jk} N_k$,
and we have defined the DeWitt metric on symmetric tensors as
\be
\label{Gjkmn}
G_{jkmn}\equiv
\frac{1}{2}\left(Q_{jm}Q_{nk}+Q_{jn}Q_{mk}-Q_{jk} Q_{mn}\right) 
\ee
which has the inverse
\be
\label{GInvjkmn}
\left[G^{-1}\right]^{jkmn}=\frac{1}{2}\left(Q^{jm} Q^{nk}+Q^{jn}Q^{mk}-2Q^{jk}
Q^{mn}\right)
\ee
that is $G_{jkmn} \left[G^{-1}\right]^{nmpq}=\delta^p_{(j} \delta^q_{k)}$.
The Ricci tensor of $Q$ is denoted by $R_{jk}[Q]$ and
$C=C_{\rm geo}+C_{\rm matter}$ denotes the split of the Hamiltonian constraint
(with the dust reference system excluded) into gravitational\footnote{
Note that we have included a cosmological constant term
$+2\sqrt{\det(q)})\Lambda$
in $C_{\rm geo}$}
and matter
contribution, as shown explicitly in (\ref{3.54}).

It is already evident that the dust model we utilised
as a physical reference system has the great advantage that, remarkably,
equations (\ref{4.5}) are almost exactly of the same form as
the corresponding equations in the gauge variant formalism,
the only difference being the last term on the right hand side
of the physical evolution equation for $\dot{P}^{jk}$.
In other words, introducing a physical reference system must necessarily
lead to corrections compared to gauge fixing, because the physical
reference system will communicate via gravitational interaction
with the original system under consideration. In the sense described
above, the dust reference system creates only a minimal modification ---
it is the minimal extension of the original gravity -- matter system
that extracts the true degrees of freedom and
allows for their physical evolution.

The other difference is that instead of having constraints imposed
on the phase space variables, $C=C_j=0$, now the dynamics of the
true degrees of freedom is subject to conservation laws
$\dot{H}=\dot{C}_j=0$. Thus, in solving (\ref{4.5}) we may prescribe
arbitrary functions $\epsilon(\sigma),\;\epsilon_j(\sigma)$ which
play the role of the (constant in $\tau$--time) energy and momentum
density, respectively. The  substitution\footnote{ The letter
$\epsilon$ is chosen to indicate that these values are small,
appropriate for {\sl test clocks and rods}. In this way it can be
guaranteed that the dust, although gravitationally coupled with the
original system, will not alter the dynamics of the original system
in an uncontrolled fashion.} $H=\epsilon,\;C_j=-\epsilon_j$ will be
crucial in what follows. In fact, in order to derive the second
order equations of motion, $\dot{\Xi},\dot{Q}_{jk}$ in (\ref{4.5})
has to be solved for $\Pi,P^{jk}$. Without the conservation laws,
this would be impossible, since $\Pi,P^{jk}$ enter the expressions
for $H,C,C_j$ in a non trivial way, i.e.~solving for them would lead
to algebraic equations of higher than fourth order. The substitution
will also be crucial for the derivation of the effective Lagrangian,
by the inverse Legendre transform, corresponding to $\HF$, see
appendix \ref{s6}.


\subsection{Second Order (Lagrangian) Formulation}
\label{s4.2} In this section we will use the first order
(Hamiltonian) equations of motion and derive the corresponding
second order (Lagrangian) equations of motion for the configuration
variables $\Xi$ and $Q_{jk}$, respectively. We will sketch the main
steps of these calculations in section \ref{DeriveEOM}. The reader
who is just interested in the results should skip this section and
go directly to section \ref{SumEOM} where the final equations are
summarised.
\subsection{Derivation of the Second Order Equations of Motion}
\label{DeriveEOM}
In this section we want to derive the second order equations of
motion for $\Xi$ and $Q_{jk}$, respectively. These second order
equations will be functions of the configuration variables
$\Xi,Q_{jk}$ and their corresponding velocities
$\dot{\Xi},\dot{Q}_{jk}$, respectively. This can be achieved by
solving for the conjugate momenta $\Pi$ and $P^{jk}$ in terms of
their corresponding velocities $\dot{\Xi}$ and $\dot{Q}_{jk}$ via the equation
of motion. The
relation between the conjugate momenta and their velocities is
given through the first order Hamiltonian equations which were
displayed in the last section in equation (\ref{4.5}).
\\
\\
We begin with the matter equation for $\Xi$. First, we have to take
the time derivative of the first order equation for $\Xi$ given in
equation (\ref{4.5}). This yields \ba \label{DDXi}
\ddot{\Xi}&=&\Big[\frac{\dot{N}}{\Q}-N\frac{(\Q)^{\bf\dot{}}}{\det{Q}}\Big]
\Pi+\frac{N}{\Q}\dot{\Pi}+{\cal L}_{\dot{\vec{N}}}\Xi+{\cal
L}_{\vec{N}}\dot{\Xi}. \ea As discussed in section \ref{s3.4}, the
shift vector $N_j:=-C_j/H$ is a constant of motion since
$\dot{C}_j=\dot{H}=0$. Therefore for the Lie derivative with respect
to $\dot{\vec{N}}$ the only non vanishing contribution is the one
including $\dot{Q}^{ij}$, \be \big({\cal
L}_{\dot{\vec{N}}}\Xi\big)=(Q^{ij}N_j)^{\bf\dot{}}\Xi_{,i}=\dot{Q}^{ij}N_j\Xi_{,
i}.
\ee We will use this result later on, but for now we will work with
the compact form of the Lie derivatives as written in equation
(\ref{DDXi}). Solving for $\Pi$ in terms of $\dot{\Xi}$ we get from
equation (\ref{4.5}) \be \label{PiVel}
\Pi(\Xi,\dot{\Xi},Q_{jk})=\frac{\Q}{N}\Big(\dot{\Xi}-{\cal
L}_{\vec{N}}\Xi\Big) \ee and thus have expressed $\Pi$ as a function
of the velocity $\dot{\Xi}$. In order to stress that $\Pi$ has to be
understood as a function of $\dot{\Xi}$, we have explicitly written
the function's arguments in this section. 
Notice that strictly speaking $\Pi$ also appears in
$N=\sqrt{1+Q^{jk}C_jC_k/H^2}$ and $N^j=-Q^jkC_k/H$. However, $C_j$ and $H$ are
treated as constants of motion as discussed before. The same applies to $P^{jk}$
below.
Next we insert this result
into equation (\ref{4.5}), obtaining \ba \label{DPiVel}
\dot{\Pi}(\Xi,\dot{\Xi})&=&[N \sqrt{\det(Q)}\; Q^{jk}
\Xi_{,k}]_{,j}-\frac{N}{2}\sqrt{\det(Q)}\; v'(\Xi) + {\cal
L}_{\vec{N}}\Big(\frac{\Q}{N}\big(\dot{\Xi}-{\cal
L}_{\vec{N}}\Xi\big)\Big)\\
&=&[N \sqrt{\det(Q)}\; Q^{jk}
\Xi_{,k}]_{,j}-\frac{N}{2}\sqrt{\det(Q)}\; v'(\Xi) +
\big(\dot{\Xi}-{\cal L}_{\vec{N}}\Xi\big)\Big({\cal
L}_{\vec{N}}\frac{\Q}{N}\Big) + \frac{\Q}{N}\Big({\cal
L}_{\vec{N}}\big(\dot{\Xi}-{\cal L}_{\vec{N}}\Xi\big)\Big).\nonumber
\ea The final second order equation of motion for $\Xi$ can be
derived by inserting equation (\ref{PiVel}) and (\ref{DPiVel}) into
equation (\ref{DDXi}). The result is \ba \ddot{\Xi}&=&\Big[
\frac{\dot{N}}{N}-\frac{(\Q)^{\bf\dot{}}}{\Q}+\frac{N}{\Q}\Big({\cal
L}_{\vec{N}}\frac{\Q}{N}\Big)\Big]\big(\dot{\Xi}-{\cal
L}_{\vec{N}}\Xi\big)
+Q^{jk} \Xi_{,k}\Big[\frac{N}{\Q}\big[N\Q]_{,j}\Big]\nonumber\\
&& +N^2\big[\Delta\Xi +[Q^{jk}]_{,j} \Xi_{,k}-\frac{1}{2}\;
v'(\Xi)\big] +2\big({\cal L}_{\vec{N}}\dot{\Xi}\big)+\big({\cal
L}_{\dot{\vec{N}}}\Xi\big) -\big({\cal L}_{\vec{N}}\big({\cal
L}_{\vec{N}}\Xi\big)\big). \ea The same procedure has to be repeated
for the gravitational equations now. Applying another time
derivative to the first order equation of $Q_{jk}$ in equation
(\ref{4.5}) yields \be \label{DDQ}
\ddot{Q}_{jk}=\Big(2\Big[\frac{\dot{N}}{\Q}
-N\frac{(\Q)^{\bf\dot{}}}{\det{Q}}\Big]G_{jkmn}+\frac{2N}{\Q}\dot{G}_{jkmn}
\Big)P^{mn}+\frac{2N}{\Q}G_{jkmn}\dot{P}^{mn} +\big({\cal
L}_{\vec{\dot{N}}}Q\big)_{jk}+\big({\cal
L}_{\vec{N}}\dot{Q}\big)_{jk}. \ee In order to solve for $P^{jk}$ in
terms of $\dot{Q}_{jk}$ we use the inverse of the tensor $G_{jkmn}$
denoted by $[G^{-1}]^{jkmn}$ and defined in equation
(\ref{GInvjkmn}). This results in \be \label{PVel}
P^{jk}(Q_{jk},\dot{Q}_{jk})=\frac{\Q}{2N}[G^{-1}]^{jkmn}\big(\dot{Q}_{mn}-\big({
\cal L}_{\vec{N}}Q\big)_{mn}\big). \ee Since the equation for
$\dot{P}^{jk}$ in (\ref{4.5}) contains $C$ which includes the
geometry as well as the matter part of the Hamiltonian constraint
(see equation (\ref{3.54}) for its explicit definition), it is a
function of the variables $Q_{jk},P^{jk},\Xi$ and $\Pi$. This was
different for $\Pi$ where its time derivative involved the matter
momentum only. Thus, in order to express $\dot{P}^{jk}$ as a
function of configuration variables and velocities, we use equation
(\ref{PiVel}) and (\ref{PVel}) and replace the momenta occurring in
$\dot{P}^{jk}$. Rewriting $C_{\mathrm{geo}}$ by means of the DeWitt
bimetric $G_{jkmn}$ we get \be
C_{\mathrm{geo}}=\frac{1}{\kappa}\Big[\frac{1}{\Q}G_{jkmn}P^{jk}P^{mn}
+\Q\big(2\Lambda-R)\Big]. \ee Using the relation in equation
(\ref{PVel}) and the fact that
$G_{jkmn}[G^{-1}]^{jkrs}=\delta^{r}_{(m}\delta^{s}_{n)}$, we obtain
\be \label{CgeoVel}
C_{\mathrm{geo}}(Q_{jk},\dot{Q}_{jk})=\frac{1}{\kappa}\Big[\frac{\Q}{4N^2}[G^{-1
}]^{jkmn}\big(\dot{Q}_{jk}-\big({\cal
L}_{\vec{N}}Q\big)_{jk}\big)\big(\dot{Q}_{mn}-\big({\cal
L}_{\vec{N}}Q\big)_{mn}\big) +\Q\big(2\Lambda-R)\Big]. \ee For the
matter part of the Hamiltonian constraint we obtain by means of
equation (\ref{PiVel}) \be \label{CmatVel}
C_{\mathrm{matter}}(\Xi,\dot{\Xi},Q_{jk})
=\frac{1}{2\lambda}\Big[\frac{\Q}{N^2}\big(\dot{\Xi}-{\cal
L}_{\vec{N}}\Xi\big)^2+\Q\big(Q^{jk}\Xi_{,j}\Xi_{,k}+v(\Xi)\big)\Big].
\ee There are two other terms in $\dot{P}^{jk}$ which include the
conjugate momenta $P^{jk}$. One is the first term on the right hand
side of equation (\ref{4.5}) being quadratic in $P^{jk}$ and the
second is the Lie derivative of $P^{jk}$. Reinserting into those
terms the relation shown in equation (\ref{PVel}), we end up with
the following expression for $\dot{P}^{jk}$ as a function of
configuration and velocity variables: \ba \label{DPVel}
\dot{P}^{jk}(Q_{jk},\dot{Q}_{jk},\Xi,\dot{\Xi}) &=&
 -\frac{\Q}{2N}Q_{mn}\Big([G^{-1}]^{jmrs}[G^{-1}]^{kntu}-\frac{1}{2}[G^{-1}]^{
mnrs}[G^{-1}]^{jktu}\Big)
\big(\dot{Q}_{rs}-\big({\cal L}_{\vec{N}}Q\big)_{rs}\big)\nonumber\\
&&
\big(\dot{Q}_{tu}-\big({\cal L}_{\vec{N}}Q\big)_{tu}\big)
+N\Big[\frac{\kappa}{2}Q^{jk}\;C
-\Q\;Q^{jk}\big(2\Lambda
+\frac{\kappa}{2\lambda}\big(\Xi^{m}\Xi_{m}+v(\Xi)\big)\big)\Big]
\nonumber\\
&& +\Q [G^{-1}]^{jkmn}\Big((D_mD_n
N)-NR_{mn}[Q]\Big)+\frac{\kappa}{2\lambda}N\Q\;\Xi^{,j}\Xi^{,k}
\nonumber\\
&&
+\Big(({\cal L}_{\vec{N}} \frac{\Q}{2N})[G^{-1}]^{jkmn}+\frac{\Q}{2N}({\cal
L}_{\vec{N}} [G^{-1}]\big)^{jkmn}\Big)\big(\dot{Q}_{mn}-\big({\cal
L}_{\vec{N}}Q\big)_{mn}\big)
\nonumber\\
&& +\Big(({\cal L}_{\vec{N}}\dot{Q}\big)_{mn}-\big({\cal
L}_{\vec{N}}\big({\cal
L}_{\vec{N}}Q\big)\big)_{mn}\Big)\frac{\Q}{2N}[G^{-1}]^{jkmn}
-\frac{\kappa}{2}H Q^{jm} Q^{kn} N_m N_n.\ea Next we insert
the expressions for $P^{jk}$ and $\dot{P}^{jk}$ in equations
(\ref{PVel}) and (\ref{DPVel}), respectively, into equation
(\ref{DDQ}) for $\ddot{Q}_{jk}$. Keeping in mind that
$G_{jkmn}Q^{mn}=-1/2Q_{jk}$, we end up with \ba \ddot{Q}_{jk}&=&
\Big[\frac{\dot{N}}{N}-\frac{(\Q)^{\bf\dot{}}}{\Q}+\frac{N}{\Q}\big({\cal
L}_{\vec{N}}\frac{\Q}{N}\big)\Big]
\big(\dot{Q}_{jk} -\big({\cal L}_{\vec{N}}Q\big)_{jk}\big)\nonumber\\
&&-\big(\dot{Q}_{rs} -\big({\cal L}_{\vec{N}}Q\big)_{rs}\big)G_{jkmn}\nonumber\\
&&
\Big[[\dot{G}^{-1}]^{mnrs}+\big({\cal L}[G^{-1}]\big)^{mnrs}
+Q_{tu}\big(\dot{Q}_{vw} -\big({\cal L}_{\vec{N}}Q\big)_{vw}\big)
\Big([G^{-1}]^{mtrs}[G^{-1}]^{nuvw}-\frac{1}{2}[G^{-1}]^{mnvw}[G^{-1}]^{turs}
\Big)\Big]\nonumber\\
&&-Q_{jk}\Big[\frac{N^2\kappa}{2\Q}\;C - N^2\big(2\Lambda
+\frac{\kappa}{2\lambda}v(\Xi)\big)\Big]
+N^2\Big[\frac{\kappa}{\lambda}\Xi_{,j}\Xi_{,k}-2R_{jk}\Big]
+2N\big(D_jD_kN\big)
-\frac{NH\kappa}{\Q}G_{jkmn}N^mN^n
\nonumber\\
&&+2\big({\cal L}_{\vec{N}}\dot{Q}\big)_{jk} +\big({\cal
L}_{\dot{\vec{N}}}Q\big)_{jk}-\big({\cal L}_{\vec{N}}\big({\cal
L}_{\vec{N}}Q\big)\big)_{jk}.
\ea Here we used that
$\dot{G}_{jkmn}[G^{1-}]^{mnrs}=-G_{jkmn}[\dot{G}^{-1}]^{jkrs}$ which
follows from $\big(G_{jkmn}[G^{-1}]^{mnrs}\big)^{\bf\dot{}}=0$ and
\ba \lefteqn{
\frac{2N}{\Q}G_{jkmn}\Big[-N\Q\;Q^{mn}\frac{\kappa}{2\lambda}\Xi^{,r}\Xi_{,r}
+N\Q\;\frac{\kappa}{2\lambda}\Xi^{,m}\Xi^{,n}\Big]}\nonumber\\
&=&\frac{2N}{\Q}\Big[+\frac{N\Q\;\kappa}{4\lambda}\;Q_{jk}\Xi^{,r}\Xi_{,r}
+\frac{N\Q\;\kappa}{2\lambda}\;\Xi_{,j}\Xi_{,k}-\frac{N\Q\kappa}{4\lambda}\;Q_{
jk}\Xi^{,r}\Xi_{,r}\Big]\nonumber\\
&=&N^2\frac{\kappa}{\lambda}\Xi_{,j}\Xi_{,k}. \ea A straightforward,
but tedious calculation shows that the second term on the right-hand
side of the equation for $\ddot{Q}_{jk}$ simplifies greatly: \ba
 \lefteqn{-\big(\dot{Q}_{rs} -\big({\cal
L}_{\vec{N}}Q\big)_{rs}\big)G_{jkmn}}\nonumber\\
&&\Big[[\dot{G}^{-1}]^{mnrs}+\big({\cal L}[G^{-1}]\big)^{mnrs}
+Q_{tu}\big(\dot{Q}_{vw} -\big({\cal L}_{\vec{N}}Q\big)_{vw}\big)
\Big([G^{-1}]^{mtrs}[G^{-1}]^{nuvw}-\frac{1}{2}[G^{-1}]^{mnvw}[G^{-1}]^{turs}
\Big)\Big]\nonumber\\
&=&
Q^{mn}\Big(\dot{Q}_{mj}-\big({\cal
L}_{\vec{N}}Q\big)_{mj}\Big)\Big(\dot{Q}_{nk}-\big({\cal
L}_{\vec{N}}Q\big)_{nk}\Big).
\ea
Consequently, the final form for the second order equation of motion for
$Q_{jk}$ is given by
\ba
\ddot{Q}_{jk}&=&
\Big[\frac{\dot{N}}{N}-\frac{(\Q)^{\bf\dot{}}}{\Q}+\frac{N}{\Q}\big({\cal
L}_{\vec{N}}\frac{\Q}{N}\big)\Big]
\big(\dot{Q}_{jk} -\big({\cal L}_{\vec{N}}Q\big)_{jk}\big)\nonumber\\
&&
+Q^{mn}\Big(\dot{Q}_{mj}-\big({\cal
L}_{\vec{N}}Q\big)_{mj}\Big)\Big(\dot{Q}_{nk}-\big({\cal
L}_{\vec{N}}Q\big)_{nk}\Big)
\nonumber\\
&&+Q_{jk}\Big[-\frac{N^2\kappa}{2\Q}\;C + N^2\big(2\Lambda
+\frac{\kappa}{2\lambda}v(\Xi)\big)\Big]
+N^2\Big[\frac{\kappa}{\lambda}\Xi_{,j}\Xi_{,k}-2R_{jk}\Big]
+2N\big(D_kD_kN\big)
\nonumber\\
&&+2\big({\cal L}_{\vec{N}}\dot{Q}\big)_{jk} +\big({\cal
L}_{\dot{\vec{N}}}Q\big)_{jk}-\big({\cal L}_{\vec{N}}\big({\cal
L}_{\vec{N}}Q\big)\big)_{jk} -\frac{NH\kappa}{\Q}G_{jkmn}N^mN^n. \ea This
finishes our derivation of the (general) second order equation of
motion for $\Xi$ and $Q_{jk}$, respectively.
\subsection{Summary of Second Order Equations of Motion}
\label{SumEOM}
The second order equations of motion for the (manifestly) gauge
invariant matter scalar field $\Xi$ and the (manifestly) gauge
invariant three metric $Q_{jk}$ have the following form: \ba
\label{ResDDXi} \ddot{\Xi}&=&\Big[
\frac{\dot{N}}{N}-\frac{(\Q)^{\bf\dot{}}}{\Q}+\frac{N}{\Q}\Big({\cal
L}_{\vec{N}}\frac{\Q}{N}\Big)\Big]\big(\dot{\Xi}-{\cal
L}_{\vec{N}}\Xi\big)
+Q^{jk} \Xi_{,k}\Big[\frac{N}{\Q}\big[N\Q]_{,j}\Big]\nonumber\\
&& +N^2\big[\Delta\Xi +[Q^{jk}]_{,j} \Xi_{,k}-\frac{1}{2}\; v'(\Xi)\big]
+2\big({\cal L}_{\vec{N}}\dot{\Xi}\big)+\big({\cal L}_{\dot{\vec{N}}}\Xi\big)
-\big({\cal L}_{\vec{N}}\big({\cal L}_{\vec{N}}\Xi\big)\big)
\ea
and
\ba
\label{ResDDQ}
\ddot{Q}_{jk}&=&
\Big[\frac{\dot{N}}{N}-\frac{(\Q)^{\bf\dot{}}}{\Q}+\frac{N}{\Q}\big({\cal
L}_{\vec{N}}\frac{\Q}{N}\big)\Big]
\big(\dot{Q}_{jk} -\big({\cal L}_{\vec{N}}Q\big)_{jk}\big)\nonumber\\
&&
+Q^{mn}\Big(\dot{Q}_{mj}-\big({\cal
L}_{\vec{N}}Q\big)_{mj}\Big)\Big(\dot{Q}_{nk}-\big({\cal
L}_{\vec{N}}Q\big)_{nk}\Big)
\nonumber\\
&&+Q_{jk}\Big[-\frac{N^2\kappa}{2\Q}\;C + N^2\big(2\Lambda
+\frac{\kappa}{2\lambda}v(\Xi)\big)\Big]
+N^2\Big[\frac{\kappa}{\lambda}\Xi_{,j}\Xi_{,k}-2R_{jk}\Big]
+2N\big(D_jD_kN\big)
\nonumber\\
&&+2\big({\cal L}_{\vec{N}}\dot{Q}\big)_{jk} +\big({\cal
L}_{\dot{\vec{N}}}Q\big)_{jk}-\big({\cal L}_{\vec{N}}\big({\cal
L}_{\vec{N}}Q\big)\big)_{jk} -\frac{NH\kappa}{\Q}G_{jkmn}N^mN^n. \ea The
term $C=C_{\mathrm{geo}}+C_{\mathrm{matter}}$ occurring in the
equation for $\ddot{Q}_{jk}$ has to be understood as a function of
configuration and velocity variables and is explicitly given by \ba
\label{CgeoCmat}
C_{\mathrm{geo}}&=&\frac{1}{\kappa}\Big[\frac{\Q}{4N^2}[G^{-1}]^{jkmn}\big(\dot{
Q}_{jk}-\big({\cal
L}_{\vec{N}}Q\big)_{jk}\big)\big(\dot{Q}_{mn}-\big({\cal
L}_{\vec{N}}Q\big)_{mn}\big)
+\Q\big(2\Lambda-R)\Big]\nonumber\\
C_{\mathrm{matter}}
&=&\frac{1}{2\lambda}\Big[\frac{\Q}{N^2}\big(\dot{\Xi}-{\cal
L}_{\vec{N}}\Xi\big)^2+\Q\big(Q^{jk}\Xi_{,j}\Xi_{,k}+v(\Xi)\big)\Big].
\ea Apart from the fact that we have a dynamical lapse function
given by $N=C/H$, as well as a dynamical shift vector defined as
$N_j=-C_j/H$, the only deviation from the standard Einstein equations
in canonical form is the last term on the right-hand side in
equation (\ref{ResDDQ}). This term, being quadratic in $N_j$ and
therefore quadratic in
$C_j=C^{\mathrm{geo}}_j+C^{\mathrm{matter}}_j$, vanishes for
instance when FRW spacetimes are considered. In our companion paper
\cite{8a}, we specialize these equations precisely to this context
and show that the equations above reproduce the correct FRW
equations.


\section{Physical Hamiltonian, Boundary Term and ADM Hamiltonian}
\label{s5}

As long as the dust space $\cal S$ (and therefore $\cal X$) has no boundary,
$\HF$
is functionally differentiable, which we always assumed so far.
However, for more general topologies we are forced to consider
boundary conditions. In this section we show how to deal with asymptotically
flat boundary conditions for illustrative purposes. More general
situations can be considered analogously.

Recall \cite{30} that asymptotically flat initial data are subject
to the following boundary conditions
\ba \label{5.1}
q_{ab}(x) &=& \delta_{ab}+\frac{f_{ab}(\Omega)}{r}+{\cal O}(r^{-2})\; ,
\nonumber\\
p^{ab}(x) &=& \frac{g^{ab}(\Omega)}{r}+{\cal O}(r^{-3})\;,
\nonumber\\
\xi(x) &=& \frac{f(\Omega)}{r^2} + {\cal O}(r^{-3})\; ,
\nonumber\\
\pi(x) &=& \frac{g(\Omega)}{r^2} + {\cal O}(r^{-3})\; .
\ea
Here, $x^a$ is an asymptotic coordinate system, $r^2\equiv \delta_{ab} x^a
x^b$ and $\Omega$ denotes the angular dependence of the unit vector $x^a/r$
(on the asymptotic sphere). The functions
$f_{ab},\;f,\; g^{ab},\; g$ are assumed to be smooth. Moreover, $f_{ab}$ is
an even function under reflection at spatial infinity on the sphere,
while $g^{ab}$ is odd ($f,g$ do not underly parity restrictions). Notice
that these boundary conditions directly translate into analogous ones
for the substitutions
$Q_{jk}\leftrightarrow q_{ab},\;
P^{jk}\leftrightarrow p^{ab},\;
\Xi\leftrightarrow \xi,\;
\Pi\leftrightarrow \pi$, because switching from $\cal X$ to $\cal S$ is
just a diffeomorphism.

The part of the differential of $\HF$ that gives rise to a boundary term is
\be \label{5.2}
\delta_{{\rm boundary}}\HF=\int_{{\cal S}} \; {\rm d}^3\sigma \; \left[N
\delta_{{\rm boundary}}C+N^j \delta_{{\rm boundary}}C_j\right]\; ,
\ee
which coincides precisely with the boundary terms produced by
the canonical Hamiltonian\footnote{With the substitutions
${\cal S}\leftrightarrow {\cal X},\;
Q_{jk}\leftrightarrow q_{ab},\;
P^{jk}\leftrightarrow p^{ab},\;
\Xi\leftrightarrow \xi,\;
\Pi\leftrightarrow \pi$ implied.}
\be \label{5.3}
H_{{\rm canonical}}=\int_{{\cal X}} \; {\rm d}^3x \; \left[nc+n^a c_a\right]
\ee
without dust. Here, the lapse and shift functions $n,n^a$ are Lagrange
multiplier
rather than dynamical quantities like $N,N^j$.
Therefore, the usual boundary terms~\cite{30} can be adopted, once the
asymptotic behaviour of the dynamical lapse and shift functions $N,N^j$
have been determined, which, in turn, is completely specified by $N_j$ because
$N^j=Q^{jk}
N_k,\;N=\sqrt{1+Q^{jk} N_j N_k}$.
The scalar field contribution to
$C,C_j$ decays as $1/r^4$, thus it is sufficient to consider the geometrical
contribution. $C^{\rm geo}_j=-2 D_k P^k_j$ decays as $1/r^3$ and is even
asymptotically. The term quadratic in $P^{jk}$ and the term quadratic in
the Christoffel symbols in $C^{\rm geo}$ decays as $1/r^4$, while the
term linear in the Cristoffel symbols decays as $1/r^3$ and is even. We
conclude that $N_j=-\frac{C_j}{H}$ (with $H=\sqrt{C^2-Q^{jk} C_j C_k}$)
is
asymptotically constant and even.
The same is true for $N$, which is anyway bounded from below by unity.
The usual computation \cite{15b,30} yields
\ba \label{5.4}
\delta_{{\rm boundary}}\HF &=& -\delta B'(N)-\delta \vec{B}'(\vec{N}')
\nonumber\\
\kappa \delta B'(N) &=&
\int_{S^2}\; \sqrt{\det(Q)} Q^{jk} Q^{mn}
\left[(D_j N) {\rm d}S_k \delta(Q_{mn}-\delta_{mn})
-(D_m N) dS_n \delta\left(Q_{jk}-\delta_{jk}\right)\right]
\nonumber\\
&& +
\int_{S^2}\; \sqrt{\det(Q)} Q^{jk} N \left[-{\rm d}S_j \delta\Gamma^m_{mk}+
S_k
\delta \Gamma^k_{jk}\right]
\nonumber\\
\kappa \delta\vec{B}'(\vec{N}) &=& 2\int_{S^2}\; {\rm d}S_j\;  N^k \delta P^j_k
\; ,
\ea
where ${\rm d}S_j =x^j/r {\rm d}\Omega\;,\; {\rm d}\Omega$ the volume element of
$S^2$.
The prime is to indicate that, contrary to what the notation suggests,
the terms shown are, a priori, not total differentials. In the usual formalism
they are, because lapse and shift functions are Lagrange multipliers and
do not depend on phase space. Here, however, they are dynamical and
we must worry about the variation $\delta N_j$.

It turns out that the boundary conditions need to be refined in order to
make $\HF$ differentiable. We will not analyze the most general boundary
conditions in this paper, but just make a specific choice that will be
sufficient for our purposes. Namely, we will impose in addition that
$C_j^{\rm geo}=-2 D_k P^k_j$
decays as $1/r^{3+\epsilon},\; (\epsilon>0)$ at spatial infinity. Then
also $C_j$ falls off as
$1/r^{3+\epsilon}$. Since $C$ decays like $1/r^3$, it
follows that also
$H=\sqrt{C^2-Q^{jk} C_j C_k}$ decays as $1/r^3$, whence
$N_j=-C_j/H$
decays as $1/r^\epsilon$. This implies that $\delta N_j$ decays
as \footnote{We could more generally have imposed
that $\delta N_j$ falls off like $1/r^\epsilon$ at spatial infinity.}
$1/r^\epsilon$.
Thus, $\delta N=[Q^{jk} N_k \delta N_j+N_j N_k \delta Q^{jk}/2]/N$ decays
as \footnote{The choice $\epsilon=1/2$ seems to be
appropriate in order to reproduce the asymptotic Schwarzschild decay for
dynamical lapse and shift. However, this is not the case here, because we are
automatically in a freely falling (dust) frame and $\tau$ is eigentime.
Hence, $g_{\tau\tau}=-N^2+Q^{jk} N_j N_k-1$, whatever choice for $N_j$
is made, it is independent of $\tau$.}
$1/r^{2\epsilon}$.

It follows that $\delta\vec{B}'(\vec{N})=0$ and
\be \label{5.5}
\kappa \delta B'(N) = \kappa \delta E_{\rm ADM},\;
E_{\rm ADM} =
\int_{S^2}\; \sqrt{\det(Q)} Q^{jk}  \left[-{\rm d}S_j \; \Gamma^m_{mk}+S_k
\delta \Gamma^k_{jk}\right]
\ee
reduces to the variation of the ADM energy.
The correct Hamiltonian is given by
\be \label{5.6}
\HF=E_{ADM}+\int_{{\cal S}}\; {\rm d}^3\sigma\; H(\sigma) \;.
\ee
It is reassuring that in the asymptotically flat context, we have
automatically a boundary term in the Hamiltonian, which is just the ADM
energy. The additional bulk term comes from the dust energy density and
{\sl does not vanish} on the constraint surface.
\\
\\
 Before we close this section, let us also mention the physical 
 Hamiltonian system under consideration there exists a Lagrangian from 
 which derives by Legendre transformation on the phase space defined by 
 the physical observables. Curiously, the corresponding effective action 
turns out to be 
 non-local in dust space but local in dust time. It can be computed via 
 a fixed point equation to any order in the spatial derivatives of the 
fields. Details can be found in appendix \ref{s6}. 
%


\section{Linear, Manifestly Gauge Invariant Perturbation Theory}
\label{s7} In section \ref{s4} we derived the (manifestly) gauge
invariant second order equation of motion for the scalar field $\Xi$
and the three metric $Q_{jk}$. Now we want to consider small
perturbations around a given background whose corresponding
quantities will be denoted by a bar,  namely $\ov{\Xi},\ov{Q}_{jk}$.
The linear perturbations are then defined as
$\delta\Xi:=\Xi-\ov{\Xi}$ and $\delta Q_{jk}:=Q_{jk}-\ov{Q}_{jk}$,
respectively. Note that these perturbations are also (manifestly)
gauge invariant because they are defined as a difference of two
gauge invariant quantities. Consequently, any power of $\delta\Xi$
and $\delta Q_{jk}$ will also be a (manifestly) gauge invariant
quantity such that in the framework introduced in this paper gauge
invariant perturbation theory up to arbitrary order is not only
possible, but also straightforward. This is a definite strength of
our approach compared to the traditional one, see section \ref{s8}
for a detailed discussion.
\\
\\
However, in this section we will restrict ourselves to linear
(manifestly) gauge invariant perturbation theory. That is any
function $F(Q_{jk},\Xi)$ will be expanded up to linear order in the
perturbations $\delta\Xi$ and $\delta Q_{jk}$. We denote by $\delta
F$ the linear order term in the Taylor expansion of the expression
$F(Q_{jk},\Xi)-F(\ov{Q}_{jk},\ov{\Xi})$. A calculation of higher
order terms will be the content of a future publication. Usually, in
cosmological perturbation theory one chooses an FRW background and
considers small perturbations around it. Here we will derive the
equations of motion for the linear perturbations assuming an
 arbitrary background. In our companion paper, we will specialise
the results derived here to the case of an FRW background and show
that we can reproduce the standard results as presented, e.g., in
\cite{2}. The reader who is only interested in the final form of
the perturbed equation of motions should go directly to section
\ref{SummPer}, where a summary of the results is provided.
\subsection{Derivation of the Equation of Motion for $\delta\Xi$}
Let us go back to the second order equation of motion for $\Xi$
shown in equation (\ref{ResDDXi}). Since its perturbation involves
several terms we will, for reasons of book keeping, split this
equation into several parts and consider the perturbation of those
parts separately. Comparing the equation of $\Xi$ with the one
for $Q_{ij}$ in equation (\ref{ResDDQ}), we realise that in both equations the first
term on the right-hand side includes an identical term, namely the
expression in the square brackets. Therefore it is convenient to
derive its perturbation in a closed form such that the result can
then also be used for the derivation of the equation of motion of
$\delta Q_{ij}$. Throughout this section we will repeatedly need the
perturbation of $\Q$ and its inverse, given by \ba \label{DQ}
\delta\Q&=&\frac{1}{2}\Qb\;\ov{Q}^{ij}\delta Q_{ij}\nonumber\\
\delta\big(\frac{1}{\Q}\big)&=&-\frac{1}{2}\frac{1}{\Qb}\;\ov{Q}^{ij}\delta
Q_{ij}. \ea Next, considering the definition of the lapse functions
given by $N=C/H$ and the definition of $H=\sqrt{C^2-Q^{ij}C_iC_j}$
we obtain $N=\sqrt{1+Q^{ij}N_iN_j}$. Thus, as mentioned before, $N$
is not an independent variable, but can be expressed in terms of
$Q_{ij}$ and the shift vector, which itself is a function of the
elementary variables $\Xi,Q_{ij}$. However, as shown in our
companion paper \cite{8a}, the perturbation of $N_j$ is again a
constant of motion, that is $\delta\dot{N}_j=0$. For this reason it
is convenient to work with $\delta Q_{jk},\delta\Xi$ and $\delta
N_j$, although, strictly speaking, $\delta N_j$ is not an
independent perturbation. Keeping this in mind the perturbation of
$N$ yields \be \label{DN} \delta
N=\Big[-\frac{\ov{N}}{2}\Big(\frac{\ov{N}^j\ov{N}^{k}}{\ov{N}^2}\Big)\Big]\delta
Q_{jk} +\Big[\ov{N}\Nv{j}\Big]\delta N_j .\ee From the explicit form
of $\delta N$ and $\delta\Q$ we can derive the following expressions
which we will need below: \ba \delta\Big(\frac{\dot{N}}{N}\Big) &=&
\Big[-\frac{\partial}{\partial\tau}\frac{1}{2}\Big(\frac{\ov{N}^j\ov{N}^k}{\ov{N
}^2}\Big)\Big]\delta Q_{jk}
+\Big[\Big(\frac{\ov{N}^j}{\ov{N}^2}\Big)^{\bf\dot{}}\;\;\Big]\delta N_j\\
\delta\Big(-\frac{\big(\Q\big)^{\bf\dot{}}}{\Q}\Big)&=&
\Big[-\frac{\partial}{\partial\tau}\frac{1}{2}\ov{Q}^{jk}\Big]\delta
Q_{jk}.\nonumber \ea Here the derivative with respect to $\tau$ is
understood to act on everything to its right, including the
perturbations. We also used that $\delta N_j$ is a constant of
motion, so the term proportional to $\delta\dot{N}_j$ does not
contribute. In order to calculate the perturbation of the Lie
derivative term occurring in the square brackets in eqn.
(\ref{ResDDXi}), we compute \ba \delta\Big(\frac{N}{\Q}\Big)&=&
\Big[-\Big(\frac{\ov{N}}{\Qb}\Big)\QN{j}{k}\Big]\delta Q_{jk}
+\Big[\Nv{j}\Big]\delta N_j\nonumber\\
\delta\Big(\frac{\Q}{N}\Big)&=&
\Big[\Big(\frac{\Qb}{\ov{N}}\Big)\QN{j}{k}\Big]\delta Q_{jk}
+\Big[-\frac{\Qb}{\ov{N}}\Nv{j}\Big]\delta N_j.
\ea
The Lie derivative term yields then
\ba
\delta\Big(\big({\cal L}_{\N}\frac{\Q}{N}\big)\Big)
&=&
\Big[{\cal L}_{\Nb}\Big(\frac{\Qb}{\ov{N}}\QN{j}{k}\Big)\Big]\delta Q_{jk}\nonumber\\
&&
+\Big[- {\cal L}_{\Nb}\frac{\Qb}{\ov{N}}\Big(\Nv{j}\Big)\Big]\delta N_j\nonumber\\
&&
+\Big( {\cal L}_{\delta\N}\frac{\Qb}{\ov{N}}\Big).
\ea Similar to the $\tau$-derivative, the Lie derivative ${\cal
L}_{\N}$ acts on all terms to its right, including the linear
perturbations. For the moment we keep the Lie derivative with
respect to $\delta\N$ in the compact form above. At a later stage we
will write down these terms explicitly and express them in terms of
$\delta Q_{jk}$ and $\delta N_j$. Having calculated the variations
of all components of the first square bracket in eqn.
(\ref{ResDDXi}), we can now give the final result: \ba \label{ResSB}
\delta\Big[
\frac{\dot{N}}{N}-\frac{(\Q)^{\bf\dot{}}}{\Q}+\frac{N}{\Q}\Big({\cal
L}_{\vec{N}}\frac{\Q}{N}\Big)\Big]
&=&\Big[-\Big(\frac{\partial}{\partial\tau}-{\cal
L}_{\Nb}\Big)\Big(\QN{j}{k}\Big)\Big]\delta Q_{jk}\nonumber\\
&&
+\Big[\Big(\frac{\partial}{\partial\tau}-{\cal
L}_{\Nb}\Big)\Nv{j}\Big]\delta N_j\nonumber\\
&&
+\frac{\ov{N}}{\Qb}\Big({\cal
L}_{\delta\N}\frac{\Qb}{\ov{N}}\Big). \ea The last term in
this equation can be written explicitly in terms of $\delta Q_{mn}$
and $\delta N_m$ as \ba \label{SBLie} \frac{\ov{N}}{\Qb}\Big({\cal
L}_{\delta\N}\frac{\Qb}{\ov{N}}\Big) &=&
\Big[-\frac{\ov{N}}{\Qb}\frac{\partial}{\partial
x^m}\Big(\frac{\Qb}{\ov{N}}\ov{Q}^{jm}\ov{N}^k\Big)\Big] \delta
Q_{jk}\nonumber\\
&&
 +\Big[\frac{\ov{N}}{\Qb}\frac{\partial}{\partial
x^k}\Big(\frac{\Qb}{\ov{N}}\ov{Q}^{jk}\Big)\Big]
\delta N_{j}.\ea As the terms in equation (\ref{ResSB}) are multiplied by
$(\dot{\Xi} -({\cal L}_{\N}\Xi))$ in equation (\ref{ResDDXi}), we
also need the perturbation of the latter term. It is given by \ba
\label{SBfac} \delta\big(\dot{\Xi}-\big({\cal L}_{\N}\Xi\big)\big)
&=&\Big(\frac{\partial}{\partial\tau}-{\cal
L}_{\Nb}\Big)\delta\Xi-\big({\cal L}_{\delta\N}\ov{\Xi}\big). \ea
Next we determine the perturbation of $Q^{jk}\Xi_{,k}$ which enters
the second term on the right-hand side of equation (\ref{ResDDXi}):
\ba \delta\big(Q^{jk}\Xi_{,k}\Big)&=&
\Big[-\ov{Q}^{jm}\ov{Q}^{kn}\Xi_{,k}\Big]\delta
Q_{mn}+\Big[\ov{Q}^{jk}\frac{\partial}{\partial x^k}\Big]\delta\Xi.
\ea The perturbation of the term that is multiplied with
$Q^{jk}\Xi_{,k}$ yields \ba
\lefteqn{\delta\Big[\frac{N}{\Q}\big(N\Q)_{,j}\Big]}\\
&=&
\Big[-\frac{\ov{N}}{\Qb}\QN{j}{k}\big[\ov{N}\Qb]_{,j}+\frac{\ov{N}}{\Qb}\frac{
\partial}{\partial
x^j}\Big(\frac{1}{2}\Big(\ov{Q}^{jk}-\frac{\ov{N}^j\ov{N}^k}{\ov{N}^2}\Big)\ov{N
}\Qb\Big)\Big]\delta Q_{jk}\nonumber\\
&&+\Big[\frac{\ov{N}}{\Qb}\Nv{j}\big[\ov{N}\Qb]_{,j}+\frac{\ov{N}}{\Qb}\frac{
\partial}{\partial x^j}\Big(\Nv{j}\ov{N}\Qb\Big)\Big]\delta N_j\nonumber.
\ea We will split the third term occurring on the right-hand side of
equation (\ref{ResDDXi}) into $N^2$ and the remaining part given by
$(\Delta\Xi+[Q^{jk}]_{,j}\Xi_{,k}-\frac{1}{2}v^{\prime}(\Xi))$.
Their respective perturbations are \ba \delta N^2&=&2\ov{N}\delta
N=\Big[-\ov{N}^j\ov{N}^k\Big]\delta Q_{jk}
+\Big[2\ov{N}^j\Big]\delta N_j\\
\delta\Big(\Delta\Xi+[Q^{jk}]_{,j}\Xi_{,k}-\frac{1}{2}v^{\prime}(\Xi)\Big)
&=& \Big[-\frac{\partial}{\partial
x^n}\Big(\ov{Q}^{jm}\ov{Q}^{kn}\Xi_{,m}\Big)\Big]\delta Q_{jk}\nonumber\\
&&
+\Big[\Delta + [\ov{Q}^{mn}]_{,n}\frac{\partial}{\partial
x^m}-\frac{1}{2}v^{\prime\prime}(\ov{\Xi})\Big]\delta\Xi.\nonumber
\ea Finally, we have to calculate the perturbation of the last three
terms in equation (\ref{ResDDXi}), involving Lie derivatives of
$\Xi$ and $\dot{\Xi}$, respectively. We obtain \ba
\delta\Big(2\big({\cal L}_{\vec{N}}\dot{\Xi}\big)+\big({\cal
L}_{\dot{\vec{N}}}\Xi\big) -\big({\cal L}_{\vec{N}}\big({\cal
L}_{\vec{N}}\Xi\big)\big)\Big) &=&\Big[{\cal
L}_{\Nb}\big(\frac{\partial}{\partial\tau}-{\cal L}_{\Nb}\Big)
+\frac{\partial}{\partial\tau}{\cal L}_{\Nb}\Big]\delta\Xi\\
&&+\Big({\cal L}_{\delta N}\big(\frac{\partial}{\partial\tau}-{\cal
L}_{\Nb}\big)+\big(\frac{\partial}{\partial\tau}-{\cal
L}_{\Nb}\big){\cal L}_{\delta N}\Big)\Big[\ov{\Xi}\Big]\nonumber.
\ea Adding up all the contributions and factoring out the linear
perturbations $\delta\Xi,\delta Q_{jk}$ and $\delta N_j$, we can
rewrite the second order EOM as \ba \label{ResDDdXi} \delta\ddot{\Xi}&=& \Big[
\Big[\frac{\dot{\ov{N}}}{\ov{N}}-\frac{(\Qb)^{\bf\dot{}}}{
\Qb}+\frac{\ov{N}}{\Qb}\Big({\cal
L}_{\Nb}\frac{\Qb}{\ov{N}}\Big)\Big]\big(\frac{\partial}{\partial\tau}-{\cal
L}_{\Nb}\big)
+\ov{Q}^{jk}\Big[\frac{\ov{N}}{\Qb}\big(\ov{N}\Qb)_{,j}\Big]\frac{\partial}{
\partial x^k}\\
&&
\quad
+\ov{N}^2\Big(\Delta+[\ov{Q}^{mn}]_{,n}\frac{\partial}{\partial
x^m}-\frac{1}{2}v^{\prime\prime}(\ov{\Xi})\Big)
+\Big({\cal L}_{\Nb}\big(\frac{\partial}{\partial\tau}-{\cal L}_{\Nb}\Big)
+\frac{\partial}{\partial\tau}{\cal L}_{\Nb}\Big)
\Big]\delta\Xi\nonumber\\
&&
+\Big[
- \big(\dot{\ov{\Xi}}-\big({\cal
L}_{\Nb}\ov{\Xi}\big)\big)
\Big[
\Big(\frac{\partial}{\partial\tau}-{\cal
L}_{\Nb}\Big)\Big(\QN{j}{k}\Big)
+\frac{\ov{N}}{\Qb}\frac{\partial}{\partial
x^m}\Big(\frac{\Qb}{\ov{N}}\ov{Q}^{jm}\ov{N}^k\Big)
\Big]\nonumber\\
&&
\quad
-\ov{Q}^{jm}\ov{Q}^{kn}\ov{\Xi}_{,n}\big[\frac{\ov{N}}{\Q}\big(\ov{N}\Q)_{,m}
\big]
+\Big[\frac{\dot{\ov{N}}}{\ov{N}}-\frac{(\Qb)^{\bf\dot{}}}{
\Qb}+\frac{\ov{N}}{\Qb}\Big({\cal
L}_{\Nb}\frac{\Qb}{\ov{N}}\Big)\Big]\Big(\ov{Q}^{jm}\ov{\Xi}_{,m}\ov{N}
^k\Big)\Big)
\nonumber\\
&&
\quad
+\big(\ov{Q}^{mn}\ov{\Xi}_{,n}\big)\Big(-\frac{\ov{N}}{\Qb}\QN{j}{k}\big[\ov{N}
\Qb]_{,m}\nonumber\\
&&
\quad
+\frac{\ov{N}}{\Qb}\frac{\partial}{\partial
x^m}\Big(\frac{1}{2}\Big(\ov{Q}^{jk}-\frac{\ov{N}^j\ov{N}^k}{\ov{N}^2}\Big)\ov{N
}\Qb\Big)\Big)
\nonumber\\
&&
\quad
-\ov{N}^2\Big(\Delta\ov{\Xi}+[\ov{Q}^{mn}]_{,m}\ov{\Xi}_{,n}-\frac{1}{2}v^{
\prime}(\ov{\Xi})\Big)\Big(\frac{\ov{N}^j\ov{N}^k}{\ov{N}^2}\Big)
-\ov{N}^2\frac{\partial}{\partial
x^n}\Big(\ov{Q}^{jm}\ov{Q}^{kn}\ov{\Xi}_{,m}\Big)\nonumber\\
&&
\quad
-\ov{Q}^{jm}\ov{Q}^{kn}\ov{N}_m\big[\dot{\ov{\Xi}}-\big({\cal
L}_{\Nb}\ov{\Xi}\big)\big]_{,n}
-\big(\frac{\partial}{\partial\tau}-{\cal
L}_{\Nb}\big)\Big(\ov{Q}^{jm}\ov{Q}^{kn}\ov{N}_m\ov{\Xi}_{,n}\Big)
\Big]\delta Q_{jk}\nonumber\\
&&
+\Big[
\big(\dot{\ov{\Xi}}-\big({\cal L}_{\Nb}\ov{\Xi}\big)\big)
\Big[\Big(\frac{\partial}{\partial\tau}-{\cal L}_{\Nb}\Big)\Nv{j}
+\frac{\ov{N}}{\Qb}\frac{\partial}{\partial
x^k}\Big(\frac{\Qb}{\ov{N}}\ov{Q}^{jk}\Big)
\Big]\nonumber\\
&&
\quad
+2\ov{N}^2\Big(\Delta\ov{\Xi}+[\ov{Q}^{mn}]_{,m}\ov{\Xi}_{,n}-\frac{1}{2}v^{
\prime}(\ov{\Xi})\Big)\Nv{j}
\nonumber\\
&&
\quad
+\big(\ov{Q}^{mn}\ov{\Xi}_{,n}\big)\Big(\frac{\ov{N}}{\Qb}\Nv{j}\big[\ov{N}\Qb]_
{,m}
+\frac{\ov{N}}{\Qb}\frac{
\partial}{\partial x^j}\Big(\Nv{j}\ov{N}\Qb\Big)\Big)\nonumber\\
&& \quad -\Big[\frac{\dot{\ov{N}}}{\ov{N}}-\frac{(\Qb)^{\bf\dot{}}}{
\Qb}+\frac{\ov{N}}{\Qb}\Big({\cal
L}_{\Nb}\frac{\Qb}{\ov{N}}\Big)\Big]\Big(\ov{Q}^{jk}\ov{\Xi}_{,k}\Big)
\nonumber\\
&&
\quad
+\ov{Q}^{jk}\big[\dot{\ov{\Xi}}-\big({\cal
L}_{\Nb}\ov{\Xi}\big)\big]_{,k}
+\big(\frac{\partial}{\partial\tau}-{\cal
L}_{\Nb}\big)\big(\ov{Q}^{jk}\ov{\Xi}_{,k}\big) \Big]\delta
N_j.\nonumber \ea Here we used that the last term occurring in
equation (\ref{ResDDXi}), which involves the Lie derivative with
respect to $\delta\vec{N}$, can again be expressed in terms of the
perturbations $\delta Q_{jk}$ and $\delta N_j$, explicitly given by
\ba 
\lefteqn{\Big({\cal L}_{\delta N}\big(\frac{\partial}{\partial\tau}-{\cal
L}_{\Nb}\big)+\big(\frac{\partial}{\partial\tau}-{\cal
L}_{\Nb}\big){\cal L}_{\delta N}\Big)\Big[\ov{\Xi}\Big]}\\ &=&
\Big[\ov{Q}^{jk}\big[\dot{\Xi}-\big({\cal L}_{\Nb}\Xi\big)\big]_{,k}
+\big(\frac{\partial}{\partial\tau}-{\cal
L}_{\Nb}\big)\big(\ov{Q}^{jk}\ov{\Xi}_{,k}\big)\Big]\delta N_j\nonumber\\
&&\Big[-\ov{Q}^{jm}\ov{Q}^{kn}\ov{N}_m\big[\dot{\ov{\Xi}}-\big({\cal
L}_{\Nb}\ov{\Xi}\big)\big]_{,n}
-\big(\frac{\partial}{\partial\tau}-{\cal
L}_{\Nb}\big)\Big(\ov{Q}^{jm}\ov{Q}^{kn}\ov{N}_m\ov{\Xi}_{,n}\Big)\Big]\delta
Q_{jk}.\nonumber \ea Analogously, the last term on the right hand
side of equation (\ref{ResSB}), which also involves a Lie derivative
with respect to $\delta\vec{N}$, can be expressed as \be
\frac{\ov{N}}{\Qb}\Big({\cal
L}_{\delta\vec{N}}\frac{\Qb}{\ov{N}}\Big)=\Big[-\frac{\ov{N}}{\Qb}\frac{\partial}{\partial
x^m}\Big(\frac{\Qb}{\ov{N}}\ov{Q}^{jm}\ov{N}^k\Big)\Big]\delta
Q_{jk} +\Big[\frac{\ov{N}}{\Qb}\frac{\partial}{\partial
x^k}\Big(\frac{\Qb}{\ov{N}}\ov{Q}^{jk}\Big)\Big]\delta N_j. \ee
Moreover $\big({\cal L}_{\delta\vec{N}}\ov{\Xi}\big)$ occurring in
equation (\ref{SBfac}) is given by \be \big({\cal
L}_{\delta\vec{N}}\ov{\Xi}\big)=\Big[-\ov{Q}^{jm}\ov{\Xi}_{,m}\ov{N}^k\Big]
\delta
Q_{jk}+\Big[\ov{Q}^{jk}\Xi_{,k}\Big]\delta N_j. \ee We would like to
emphasise again that the partial and Lie derivatives act on
everything to their right and therefore also on the perturbations.
That is the reason why for instance $\delta\dot{\Xi}$ and
$\delta\dot{Q}_{jk}$ have not been factored out separately. This
finishes our derivation of the second order equation of motion for
the scalar field perturbation $\delta\Xi$. In the next section we
will discuss the corresponding equation for the metric perturbation
$\delta Q_{jk}$.
\subsection{Derivation of the Equation of Motion for $\delta Q_{jk}$}
\label{DerDQ}
Similar to the derivation of the second order equation of motion for
$\delta\Xi$ we will also split the equation for $\ddot{Q}_{jk}$ in
equation (\ref{ResDDQ}) into several terms whose perturbations are
then considered separately. More precisely, we will split the
equation into seven separate terms, with the three terms involving
Lie derivatives in the last line counted as one. Starting with the
first term on the right-hand side, we recall that the perturbation
of the sum of terms in the square brackets has already been computed
during the derivation of the equation for $\delta\ddot{\Xi}$. Thus,
we can take over those results, as shown in equation (\ref{ResSB}).
The perturbation of the term $(\dot{Q}_{jk}-({\cal L}_{\N}Q)_{jk})$
is given by \ba \label{Per1final} \delta\big(\dot{Q}_{jk}-\big({\cal
L}_{\N}Q\big)_{jk}\big) &=& \big(\frac{\partial}{\partial\tau}-{\cal
L}_{\Nb}\big)\delta Q_{jk}-\big({\cal L}_{\delta N}\ov{Q}\big)_{jk},
\ea whereby \ba \label{Per1Lie} \big({\cal L}_{\delta N}Q\big)_{jk}
&=& \Big[-\ov{Q}^{mr}\ov{N}^n\Big[Q_{jk}\Big]_{,r}
-\ov{Q}_{km}\frac{\partial}{\partial
x^j}\Big(\ov{Q}^{mr}\ov{N}^n\Big)
-\ov{Q}_{jm}\frac{\partial}{\partial
x^k}\Big(\ov{Q}^{mr}\ov{N}^n\Big)
\Big]\delta Q_{mn}\\
&& +\Big[\ov{Q}^{mn}\Big[Q_{jk}\Big]_{,n}
+\ov{Q}_{kn}\frac{\partial}{\partial x^j}\Big(\ov{Q}^{mn}\Big)
+\ov{Q}_{jn}\frac{\partial}{\partial x^k}\Big(\ov{Q}^{mn}\Big)
\Big]\delta N_{m}.\nonumber \ea Consequently, we have all
ingredients needed for the perturbation of the first term. However,
we will not present the final expression in the main text since it
is quite lengthly. Nevertheless, the interested reader can find the
explicit form in appendix \ref{LPA} in equation (\ref{Per1final}).
The perturbation of the second term on the right-hand side yields
\ba \label{Per2final}
\lefteqn{\delta\Big(Q^{mn}\Big(\dot{Q}_{mj}-\big({\cal
L}_{\vec{N}}Q\big)_{mj}\Big)
\Big(\dot{Q}_{nk}-\big({\cal L}_{\vec{N}}Q\big)_{nk}\Big)\Big)}\\
&=&
\Big[-\ov{Q}^{mr}\ov{Q}^{ns}\Big(\dot{\ov{Q}}_{rj}-\big({\cal
L}_{\Nb}\ov{Q}\big)_{rj}\Big)\Big(\dot{\ov{Q}}_{sk}
-\big({\cal L}_{\Nb}\ov{Q}\big)_{sk}\Big)\Big]\delta Q_{mn}
+\Big[2\ov{Q}^{mn}\Big(\dot{\ov{Q}}_{n(k}-\big({\cal
L}_{\Nb}\ov{Q}\big)_{n(k}\Big)\big(\frac{\partial}{\partial\tau}-{\cal
L}_{\Nb}\big)\Big]
\delta Q_{j)m}\nonumber\\
&& +\Big[-2\ov{Q}^{mn}\Big(\dot{\ov{Q}}_{n(k}-\big({\cal
L}_{\Nb}\ov{Q}\big)_{n(k}\Big)\Big]\big({\cal L}_{\delta
N}\ov{Q}\big)_{j)m}.\nonumber \ea The perturbation of the third term
occurring on the right-hand side requires a bit more work, because
it involves $C=C^{\mathrm{geo}}+C^{\mathrm{matter}}$. Thus we will
perform this calculation in two steps. First we will ignore the
explicit form of $\delta C$ in terms of $\delta Q_{jk},\delta\Xi$
and $\delta N_j$. The resulting expression looks like \ba
\label{Per31}
\lefteqn{\delta\Big(Q_{jk}\Big[-\frac{N^2\kappa}{2\Q}\;C +
N^2\big(2\Lambda +\frac{\kappa}{2\lambda}v(\Xi)\big)\Big]\Big)
}\\
&=&
\Big[\ov{Q}_{jk}\Big[\QN{m}{n}\frac{\kappa\ov{N}^2}{2\Qb}\ov{C}+\frac{\ov{N}
^m\ov{N}^n}{2\ov{N}^2}\ov{N}^2\Big(\frac{\kappa}{2\Qb}\ov{C}
-2\big(2\Lambda+\frac{\kappa}{2\lambda}
v(\Xi)\big)\Big)\Big)\Big]\Big]\delta Q_{mn}\nonumber\\
&&
+\Big[\ov{N}^2\Big(-\frac{\kappa}{2\Qb}\ov{C}+2\Lambda+\frac{\kappa}{2\lambda}
v(\ov{\Xi})\Big)\Big]\delta Q_{jk}
+\Big[2\ov{Q}_{jk}\ov{N}^2\Big(-\frac{\kappa}{2\Qb}\ov{C} +2\Lambda
+\frac{\kappa}{2\lambda}v(\ov{\Xi})\Big)\Big(\frac{\ov{N}^m}{\ov{N}^2}\Big)\Big]
\delta N_m\nonumber\\
&& +\Big[-\ov{Q}_{jk}\ov{N}^2\frac{\kappa}{2\Qb}\Big]\delta C
+\Big[\ov{Q}_{jk}\ov{N}^2\frac{\kappa}{2\lambda}v^{\prime}(\ov{\Xi})\Big]
\delta\Xi.\nonumber
\ea Due to its length, the explicit calculation of $\delta C$ can be
found in appendix \ref{LPA}. However, when actually inserting the
expression of $\delta C$ into equation (\ref{Per31}), some of the
terms in the expression for $\delta C$ cancel with existing terms in
equation (\ref{Per31}). As a result, the final expression for the
perturbation of the third term in equation (\ref{Per3final}) is
slightly less involved. It is given by \ba \label{Per3final}
\lefteqn{\delta\Big(Q_{jk}\Big[-\frac{N^2\kappa}{2\Q}\;C +
N^2\big(2\Lambda +\frac{\kappa}{2\lambda}v(\Xi)\big)\Big]\Big)
}\\
&=&
\Big[
-\ov{Q}_{jk}\Big(
\frac{1}{2}\frac{\ov{N}^m\ov{N}^n}{\ov{N}^2}\ov{N}^2\Big(2\Lambda-\ov{R}+\frac{
\kappa}{2\lambda}\Big(v(\ov{\Xi})-\ov{Q}^{rs}\ov{\Xi}_{,r}\ov{\Xi}_{,s}\Big)
\Big)
+\frac{\ov{N}^2}{2}\Big(\ov{R}^{mn}
-[\ov{G}^{-1}]^{mnrs}\ov{D}_r\ov{D}_s\Big)\nonumber\\
&&
\quad\quad\quad
+\frac{\kappa}{2\lambda}\ov{Q}^{mr}\ov{Q}^{ns}\big(\dot{\Xi}-\big({\cal
L}_{\Nb}\ov{\Xi}\big)\big)\ov{N}_r\ov{\Xi}_{,s}
-\frac{\kappa\ov{N}^2}{2\lambda}\ov{Q}^{rm}\ov{Q}^{sn}\ov{\Xi}_{,r}\ov{\Xi}_{,s}
\nonumber\\
&&
\quad\quad\quad
+\frac{1}{4}\Big(\dot{\ov{Q}}_{rs}-\big({\cal L}_{\Nb}\ov{Q}\big)_{rs}\Big)
\Big[[\ov{G}^{-1}]^{mnrs}\Big(\frac{\partial}{\partial\tau}-{\cal L}_{\Nb}\Big)
-\Big(\dot{\ov{Q}}_{tu}-\big({\cal
L}_{\Nb}\ov{Q}\big)_{tu}\Big)\ov{Q}^{sn}[\ov{G}^{-1}]^{turm}\Big]\nonumber\\
&&
\quad\quad\quad
+\frac{1}{4}\Big(\dot{\ov{Q}}_{rs}-\big({\cal L}_{\Nb}Q\big)_{rs}\Big)
[\overline{G}^{-1}]^{rsvw}
\Big[\ov{Q}^{mt}\ov{N}^n[\ov{Q}_{vw}]_{,t}
+2\overline{Q}_{tw}\frac{\partial}{\partial x^v}\Big(\ov{Q}^{mt}\ov{N}^n\Big)\Big]
\Big)
\Big]\delta Q_{mn}\nonumber\\
&&
\Big[
-\ov{Q}_{jk}\Big(
-\Nv{m}\ov{N}^2\Big(2\Lambda-\ov{R}-\frac{\kappa}{2\lambda}\big(v(\ov{\Xi})-\ov{
Q}^{rs}\ov{\Xi}_{,r}\ov{\Xi}_{,s}\big)\Big)
-\frac{\kappa}{2\lambda}\big(\dot{\Xi}-\big({\cal
L}_{\Nb}\ov{\Xi}\big)\big)\ov{Q}^{mn}\ov{\Xi}_{,n}
\nonumber\\
&&
\quad\quad\quad
-\frac{1}{4}\Big(\dot{\ov{Q}}_{rs}-\big({\cal L}_{\Nb}Q\big)_{rs}\Big)
[\ov{G}^{-1}]^{rsvw}
\Big[\ov{Q}^{mt}[\ov{Q}_{vw}]_{,t}
+2\ov{Q}_{tw}\frac{\partial}{\partial x^v}\Big(\ov{Q}^{mt}\Big)\Big]\Big)\delta
N_m\nonumber\\
&&
+\Big[-\ov{Q}_{jk}
\Big(
\frac{\ov{N}^2\kappa}{2\lambda}\Big(\frac{1}{\ov{N}^2}\Big(\dot{\ov{\Xi}}-\big({
\cal L}_{\Nb}\ov{\Xi}\big)\Big)\Big(\frac{\partial}{\partial\tau}-{\cal
L}_{\Nb}\Big)
+\ov{Q}^{mn}\ov{\Xi}_{,n}\frac{\partial}{\partial x^m}
-\frac{1}{2}v^{\prime}(\ov{\Xi})\Big)\Big)
\Big]\delta\Xi\nonumber\\
&&
+\Big[\ov{N}^2\Big(-\frac{\kappa}{2\Qb}\ov{C}+2\Lambda+\frac{\kappa}{2\lambda}
v(\ov{\Xi})\Big)\Big]\delta Q_{jk}.
\nonumber
\ea Proceeding with the next term on the right-hand side of equation
(\ref{ResDDQ}), we obtain for its perturbation the following result:
\ba \label{Per4final}
\delta\Big(N^2\big(\frac{\kappa}{\lambda}\Xi_{,j}\Xi_{,k}
-2R_{jk}\Big)\Big) &=& \Big[
\frac{1}{2}\frac{\ov{N}^m\ov{N}^n}{\ov{N}^2}\Big(-2\ov{N}^2\big(\frac{\kappa}{
\lambda}\ov{\Xi}_{,j}\ov{\Xi}_{,k}-2\ov{R}_{jk}\big)\Big)+\ov{N}^2\ov{D}_j\ov{D}
_k\ov{Q}^{mn}\Big]\delta Q_{mn}\\
&&
+\Big[\ov{N}^2\ov{D}_m\ov{D}_n\ov{Q}^{mn}\Big]\delta Q_{jk}
+\Big[-2\ov{N}^2\ov{D}_n\ov{D}_{(k}\ov{Q}^{mn}\Big]\delta Q_{j)m}\nonumber\\
&&
+\Big[\Nv{m}\Big(2\ov{N}^2\big(\frac{\kappa}{\lambda}\ov{\Xi}_{,j}\ov{\Xi}_{,k}
-2\ov{R}_{jk}\big)\Big)\Big]\delta
N_m
+\Big[4\ov{N}^2\frac{\kappa}{\lambda}\ov{\Xi}_{(,k}\frac{\partial}{\partial
x^{j)}}\Big]\delta\Xi.\nonumber \ea For this calculation we used the
fact that the perturbation of the Ricci tensor can be expressed in
terms of the perturbations $\delta Q_{jk}$. The explicit relation
reads \be \delta
R_{jk}=\Big[-\frac{1}{2}\ov{D}_m\ov{D}_n\ov{Q}^{mn}\Big]\delta
Q_{jk} +\Big[-\frac{1}{2}\ov{D}_j\ov{D}_k\ov{Q}^{mn}\Big]\delta
Q_{mn} +\Big[\ov{D}_n\ov{D}_{(k}\ov{Q}^{mn}\Big]\delta Q_{j)m}. \ee
The next term in equation (\ref{ResDDQ}) includes covariant
derivatives. Therefore we will have to consider the perturbation of
the Christoffel symbols $\Gamma^m_{jk}$. These can again be written
in terms of metric perturbations as shown below \be \delta
\Gamma^{m}_{jk}=\frac{\ov{Q}^{mn}}{2}\Big(\big(\ov{D}_j\delta
Q_{nk}\big) +\big(\ov{D}_k\delta Q_{jn}\big) -\big(\ov{D}_n\delta
Q_{jk}\big)\Big). \ee Using this we end up with \ba
\label{Per5final} \delta\big(2ND_jD_kN) &=&
\Big[\frac{1}{2}\frac{\ov{N}^m\ov{N}^n}{\ov{N}^2}\Big(-4\ov{N}\big(\ov{D}_j\ov{D
}_k\ov{N}\big)\Big)-2\ov{N}^2\ov{D}_j\ov{D}_k\Big(\frac{1}{2}\frac{\ov{N}
^m\ov{N}^n}{\ov{N}^2}\Big)\Big]\delta Q_{mn}\\
&&
+\Big[\Nv{m}\Big(4\ov{N}\big(\ov{D}_j\ov{D}_k\ov{N}\big)\Big)+2\ov{N}^2\ov{
D}_j\ov{D}_k\Big(\Nv{m}\Big)\Big]\delta N_m\nonumber\\
&&
+\Big[2\ov{N}\big(\ov{D}_n\ov{N}\big)\ov{D}_{(k}\ov{Q}^{mn}\Big]\delta
Q_{j),m}
+\Big[-\ov{N}^2\big(\ov{D}_{n}\ov{N}\big)\ov{D}_{m}\ov{Q}^{mn}\Big]\delta
Q_{jk}.\nonumber \ea Note that covariant derivatives surrounded by
round bracket such as $(D_jD_k...)$ act on the elements inside the
round brackets only. By contrast, covariant derivatives not
surrounded by round brackets act on everything to their right,
including also the perturbations.

Next we deal with the three terms on the right-hand side of equation
(\ref{ResDDQ}) that include Lie derivatives with respect to
$\delta\vec{N}$. Those Lie derivatives are again functions of
$\delta Q_{jk}$, $\delta N_j$, and partial derivatives thereof,
because we have $\delta N^m=-\ov{Q}^{mr}\ov{Q}^{ns}\ov{N}_s\delta
Q_{rn}+\ov{Q}^{mn}\delta N_n$. In this section, however, we will
work with the following compact form only:
 \ba \label{Per6final}
\delta\Big(2\big({\cal L}_{\vec{N}}\dot{Q}\big)_{jk}+\big({\cal
L}_{\dot{\vec{N}}}Q\big)_{jk} -\big({\cal L}_{\vec{N}}\big({\cal
L}_{\vec{N}}Q\big)\big)_{jk}\Big) &=&\Big[{\cal
L}_{\Nb}\big(\frac{\partial}{\partial\tau}-{\cal L}_{\Nb}\Big)
+\frac{\partial}{\partial\tau}{\cal L}_{\Nb}\Big]\delta Q_{jk}\\
&& +\Big({\cal L}_{\delta N}\big(\frac{\partial}{\partial\tau}-{\cal
L}_{\Nb}\big)+\big(\frac{\partial}{\partial\tau}-{\cal
L}_{\Nb}\big){\cal L}_{\delta N}\Big)\Big[\ov{Q}_{jk}\Big].\nonumber
\ea In the next section, we will rewrite the second order equation
of motion in a more concise form, using coefficient functions. That
will allow us to include the Lie derivative term with respect to
$\delta\vec{N}$ in the following, more explicit form: \ba
\label{Per6Lie} \lefteqn{\Big({\cal L}_{\delta
N}\big(\frac{\partial}{\partial\tau}-{\cal
L}_{\Nb}\big)+\big(\frac{\partial}{\partial\tau}-{\cal
L}_{\Nb}\big){\cal
L}_{\delta N}\Big)\Big[\ov{Q}_{jk}\Big]}\\
&=&
\Big[
-\ov{Q}^{mr}\ov{N}^n\Big[\dot{Q}_{jk}-\big({\cal
L}_{\Nb}\ov{Q}\big)_{jk}\Big]_{,r}
-\big(\dot{\ov{Q}}_{km}-\big({\cal
L}_{\Nb}\ov{Q}\big)_{kr}\big)\frac{\partial}{\partial
x^j}\Big(\ov{Q}^{mr}\ov{N}^n\Big)
-\big(\dot{\ov{Q}}_{jm}-\big({\cal
L}_{\Nb}\ov{Q}\big)_{jr}\big)\frac{\partial}{\partial
x^k}\Big(\ov{Q}^{mr}\ov{N}^n\Big)\nonumber\\
&&
\quad
+\Big(\frac{\partial}{\partial\tau}-{\cal L}_{\Nb}\Big)
\Big(-\ov{Q}^{mr}\ov{N}^n\Big[Q_{jk}\Big]_{,r}
-\ov{Q}_{km}\frac{\partial}{\partial x^j}\Big(\ov{Q}^{mr}\ov{N}^n\Big)
-\ov{Q}_{jm}\frac{\partial}{\partial x^k}\Big(\ov{Q}^{mr}\ov{N}^n\Big)\Big)
\Big]\delta Q_{mn}\nonumber\\
&&
+
\Big[
\ov{Q}^{mn}\Big[\dot{\ov{Q}}_{jk}-\big({\cal L}_{\Nb}\ov{Q}\big)_{jk}\Big]_{,n}
+\big(\dot{\ov{Q}}_{kn}-\big({\cal
L}_{\Nb}\ov{Q}\big)_{kn}\big)\frac{\partial}{\partial x^j}\Big(\ov{Q}^{mn}\Big)
+\big(\dot{\ov{Q}}_{jn}-\big({\cal
L}_{\Nb}\ov{Q}\big)_{jn}\big)\frac{\partial}{\partial
x^k}\Big(\ov{Q}^{mn}\Big)\nonumber\\
&&
\quad
+\Big(\frac{\partial}{\partial\tau}-{\cal L}_{\Nb}\Big)
\Big(\ov{Q}^{mn}\Big[Q_{jk}\Big]_{,n}
+\ov{Q}_{kn}\frac{\partial}{\partial x^j}\Big(\ov{Q}^{mn}\Big)
+\ov{Q}_{jn}\frac{\partial}{\partial x^k}\Big(\ov{Q}^{mn}\Big)\Big)\nonumber
\Big]\delta N_{m}
\ea

Another term which includes Lie derivatives with respect to
$\delta\vec{N}$ appeared previously as part of equation
(\ref{Per2final}). Performing the Lie derivative also in this case,
we end up with \ba \label{Per2Lie} \lefteqn{
\Big[-2\ov{Q}^{mn}\Big(\dot{Q}_{n(k}-\big({\cal
L}_{\Nb}\ov{Q}\big)_{n(k}\Big)\Big]\big({\cal L}_{\delta N}Q\big)_{j)m}}\\
&=& \Big[\Big(-2\ov{Q}^{tu}\Big(\dot{\ov{Q}}_{t(k}-\big({\cal
L}_{\Nb}\ov{Q}\big)_{t(k}\Big)\Big)
\Big(-\ov{Q}^{mr}\ov{N}^n\Big[\ov{Q}_{j)u}\Big]_{,r}
-\ov{Q}_{um}\frac{\partial}{\partial
x^{j)}}\Big(\ov{Q}^{mr}\ov{N}^n\Big)
-\ov{Q}_{j)m}\frac{\partial}{\partial
x^u}\Big(\ov{Q}^{mr}\ov{N}^n\Big)\Big) \Big]\delta Q_{mn}
\nonumber\\
&& \Big[\Big(-2\ov{Q}^{tu}\Big(\dot{Q}_{t(k}-\big({\cal
L}_{\Nb}\ov{Q}\big)_{t(k}\Big)\Big)
\Big(\ov{Q}^{mn}\Big[\ov{Q}_{j)u}\Big]_{,n}
+\ov{Q}_{un}\frac{\partial}{\partial x^{j)}}\Big(\ov{Q}^{mn}\Big)
+\ov{Q}_{j)n}\frac{\partial}{\partial x^u}\Big(\ov{Q}^{mn}\Big)\Big)
\Big]\delta N_{m}.\nonumber \ea The other two terms that include Lie
derivatives with respect to $\delta\vec{N}$ are parts of the
perturbation of the first term.  Their explicit expressions are
given in equation (\ref{SBLie}) and equation (\ref{Per1Lie}),
respectively. Note that these terms are written out explicitly in
appendix \ref{LPA}, where the final
form of the perturbation of the first term is calculated.
\\
Finally, we consider the last remaining term from equation
(\ref{ResDDQ}). It involves the Hamiltonian density $H$, which we
found out to be a constant of motion in section (3.4). In our
companion paper \cite{8a} we show that also $\delta N_j$ and $\delta
H$ are constants of motion. Therefore, in complete analogy to the
case of $\delta N_j$, we will factor out $\delta H$. We thus obtain
\ba \label{Per7final} \delta\Big(-\frac{N\kappa}{\Q}HG_{jkmn}N^mN^n\Big)
&=& \Big[-\frac{\ov{N}\kappa}{\Qb}\ov{N}^n\ov{N}^m\ov{G}_{jkmn}\Big]\delta
H
+\Big[-\frac{\ov{N}\kappa}{\Qb}\ov{H}\ov{Q}^{rm}\ov{N}^n\ov{N}^s\ov{G}_{snr(k}\Big]
\delta Q_{j)m}\nonumber\\
&&
+\Big[\QN{m}{n}\frac{\ov{N}\kappa}{\Qb}\ov{H}\ov{N}^r\ov{N}^s\ov{G}_{jkrs}\nonumber\\
&&
\quad
+\frac{\ov{N}\kappa}{\Qb}\ov{H}\ov{Q}^{sm}\ov{N}^r\ov{N}^n\ov{G}_{jkrs}\Big]\delta
Q_{mn}\nonumber\\
&&
+\Big[-\Nv{m}\frac{\ov{N}\kappa}{\Qb}\ov{H}\ov{N}^r\ov{N}^s\ov{G}_{jkrs}\Big]\delta
N_m.
\ea Here we used that \be \delta G_{jkmn}=-G_{jkrs}G_{mntu}\delta
[G^{-1}]^{rstu}. \ee Now we have finally derived all the individual
parts that are needed in order to write down the equation of motion
for $\delta Q_{jk}$. However, by just looking at the various
individual terms, it is clear that they are already considerably
more complicated than for the corresponding case of the matter
equation of motion for $\delta\Xi$. Nevertheless, we decided to
present the final equation in detailed form on the next page, in
particular to convey a sense of how much more involved the
geometrical part of the perturbed equations is compared to the
matter part. In the next section we will rewrite both equations of
motion, the one for $\delta\Xi$ as well as the one for $\delta
Q_{jk}$, in a more transparent form where all the complicated
background coefficients in front of the perturbations are hidden in
certain coefficient functions. In our companion paper, we will then
specialise those coefficient functions to the case of an FRW
spacetime and show that the general equations derived in the last
two subsection are (up to small correction caused by our dust clock)
in agreement with the well-known perturbation equations as
discussed, e.g., in \cite{2}. Note that we still kept the compact
form of the Lie derivative with respect to $\delta N$ in equation
(\ref{DDdQfinal}), because we wanted to present this equation on one
page only.
\newpage
\ba
\label{DDdQfinal}
\delta\ddot{Q}_{jk}
&=&
\Big[-\ov{Q}_{jk}
\Big(
\frac{\ov{N}^2\kappa}{2\lambda}\Big(\frac{1}{\ov{N}^2}\Big(\dot{\ov{\Xi}}-\big({
\cal L}_{\Nb}\ov{\Xi}\big)\Big)\Big(\frac{\partial}{\partial\tau}-{\cal
L}_{\Nb}\Big)
+\ov{Q}^{mn}\ov{\Xi}_{,n}\frac{\partial}{\partial x^m}
-\frac{1}{2}v^{\prime}(\ov{\Xi})\Big)\Big)
+\Big(4\ov{N}^2\frac{\kappa}{\lambda}\ov{\Xi}_{(,k}\frac{\partial}{\partial
x^{j)}}\Big)
\Big]\delta\Xi\nonumber\\
&&
\Big[\Big[
\frac{\dot{\ov{N}}}{\ov{N}}-\frac{(\Qb)^{\bf\dot{}}}{\Qb}+\frac{\ov{N}}{\Qb}
\Big({\cal
L}_{\Nb}\frac{\Qb}{\ov{N}}\Big)\Big]\Big(\ov{Q}^{rn}\ov{N}_n[\ov{Q}_{jk}]_{,r}
+\ov{Q}_{rk}\frac{\partial}{\partial
x^j}\Big(\ov{Q}^{rn}\ov{N}^m\Big)+\ov{Q}_{rj}\frac{\partial}{\partial
x^k}\Big(\ov{Q}^{rn}\ov{N}^m\Big)\Big)\nonumber\\
&&
\quad
-\Big(\dot{\ov{Q}}_{jk}-\big({\cal L}_{\Nb}\ov{Q}\big)_{jk}\Big)
\Big(\big(\frac{\partial}{\partial\tau}-{\cal L}_{\Nb}\big)\Big(\QN{m}{n}\Big)
+\frac{\ov{N}}{\Qb}\frac{\partial}{\partial
x^k}\big(\frac{\Qb}{\ov{N}}\ov{N}^m\ov{Q}^{rn}\big)\Big)\nonumber\\
&&
\quad
+\Big(-\ov{Q}^{mr}\ov{Q}^{ns}\Big(\dot{\ov{Q}}_{rj}-\big({\cal
L}_{\Nb}\ov{Q}\big)_{rj}\Big)\Big(\dot{\ov{Q}}_{sk}
-\big({\cal L}_{\Nb}\ov{Q}\big)_{sk}\Big)\Big)\nonumber\\
&&
\quad
+\Big(
-\ov{Q}_{jk}\Big[
\frac{1}{2}\frac{\ov{N}^m\ov{N}^n}{\ov{N}^2}\ov{N}^2\Big(2\Lambda-\ov{R}+\frac{
\kappa}{2\lambda}\Big(v(\ov{\Xi})-\ov{Q}^{rs}\ov{\Xi}_{,r}\ov{\Xi}_{,s}\Big)
\Big)
+\frac{\ov{N}^2}{2}\Big(\ov{R}^{mn}
-[\ov{G}^{-1}]^{mnrs}\ov{D}_r\ov{D}_s\Big)\nonumber\\
&&
\quad\quad\quad\quad\quad
+\frac{\kappa}{2\lambda}\ov{Q}^{mr}\ov{Q}^{ns}\big(\dot{\Xi}-\big({\cal
L}_{\Nb}\ov{\Xi}\big)\big)\ov{N}_r\ov{\Xi}_{,s}
-\frac{\kappa\ov{N}^2}{2\lambda}\ov{Q}^{rm}\ov{Q}^{sn}\ov{\Xi}_{,r}\ov{\Xi}_{,s}
\nonumber\\
&&
\quad\quad\quad\quad\quad
+\frac{1}{4}\Big(\dot{\ov{Q}}_{rs}-\big({\cal L}_{\Nb}\ov{Q}\big)_{rs}\Big)
\Big[[\ov{G}^{-1}]^{mnrs}\Big(\frac{\partial}{\partial\tau}-{\cal L}_{\Nb}\Big)
-\Big(\dot{\ov{Q}}_{tu}-\big({\cal
L}_{\Nb}\ov{Q}\big)_{tu}\Big)\ov{Q}^{sn}[\ov{G}^{-1}]^{turm}\Big]\nonumber\\
&&
\quad\quad\quad\quad\quad
+\frac{1}{4}\Big(\dot{\ov{Q}}_{rs}-\big({\cal L}_{\Nb}Q\big)_{rs}\Big)
[\overline{G}^{-1}]^{rsvw}
\Big[\ov{Q}^{mt}\ov{N}^n[\ov{Q}_{vw}]_{,t}
+2\overline{Q}_{tw}\frac{\partial}{\partial x^v}\Big(\ov{Q}^{mt}\ov{N}^n\Big)\Big]\Big]\Big)\nonumber\\
&&
\quad
+\Big(
\frac{1}{2}\frac{\ov{N}^m\ov{N}^n}{\ov{N}^2}
\Big(-2\ov{N}^2\big(\frac{\kappa}{\lambda}\ov{\Xi}_{,j}\ov{\Xi}_{,k}-2\ov{R}_{jk
}\big)-4\ov{N}\big(\ov{D}_j\ov{D}_k\ov{N}\big)\Big)+\ov{N}^2\ov{D}_j\ov{D}_k\ov{
Q}^{mn}-2\ov{N}^2\ov{D}_j\ov{D}_k\frac{1}{2}\frac{\ov{N}^m\ov{N}^n}{\ov{N}
^2}\Big)\nonumber\\
&&
\quad
+\Big(\QN{m}{n}\frac{\ov{N}\kappa}{\Qb}\ov{H}\ov{N}^r\ov{N}^s\ov{G}_{jkrs}
+\frac{\ov{N}\kappa}{\Qb}\ov{H}\ov{Q}^{sm}\ov{N}^r\ov{N}^n\ov{G}_{jkrs}\Big)
\Big]\delta Q_{mn}
\nonumber\\
&&
+\Big[
-\Big[\frac{\dot{\ov{N}}}{\ov{N}}-\frac{(\Qb)^{\bf\dot{}}}{\Qb}+\frac{\ov{N}}{
\Qb}\Big({\cal
L}_{\Nb}\frac{\Qb}{\ov{N}}\Big)\Big]
\Big(\ov{Q}^{mn}[\ov{Q}_{jk}]_{,n}+\ov{Q}_{nk}\frac{\partial}{\partial
x^j}\Big(\ov{Q}^{mn}\Big)+\ov{Q}_{nj}\frac{\partial}{\partial
x^k}\Big(\ov{Q}^{mn}\Big)\nonumber\\
&&
\quad
+\Big(\dot{\ov{Q}}_{jk}-\big({\cal L}_{\Nb}\ov{Q}\big)_{jk}\Big)
\Big(\big(\frac{\partial}{\partial\tau}-{\cal L}_{\Nb}\big)\Big(\Nv{m}\Big)
+\frac{\ov{N}}{\Qb}\frac{\partial}{\partial
x^n}\big(\frac{\Qb}{\ov{N}}\ov{Q}^{mn}\big)\Big)
\nonumber\\
&&
\quad
+\Big(
-\ov{Q}_{jk}\Big[
-\Nv{m}\ov{N}^2\Big(2\Lambda-\ov{R}-\frac{\kappa}{2\lambda}\big(v(\ov{\Xi})-\ov{
Q}^{rs}\ov{\Xi}_{,r}\ov{\Xi}_{,s}\big)\Big)
-\frac{\kappa}{2\lambda}\big(\dot{\Xi}-\big({\cal
L}_{\Nb}\ov{\Xi}\big)\big)\ov{Q}^{mn}\ov{\Xi}_{,n}
\nonumber\\
&&
\quad\quad\quad\quad\quad
-\frac{1}{4}\Big(\dot{\ov{Q}}_{rs}-\big({\cal L}_{\Nb}Q\big)_{rs}\Big)
[\ov{G}^{-1}]^{rsvw}
\Big[\ov{Q}^{mt}[\ov{Q}_{vw}]_{,t}
+2\ov{Q}_{tw}\frac{\partial}{\partial x^v}\Big(\ov{Q}^{mt}\Big)\Big]\Big]\Big)
\nonumber\\
&&
\quad
+\Big(-\Nv{m}\frac{\ov{N}\kappa}{\Qb}\ov{H}\ov{N}^r\ov{N}^s\ov{G}_{jkrs}\Big)
\nonumber\\
&&
\quad
+\Big(\Nv{m}\Big(2\ov{N}^2\big(\frac{\kappa}{\lambda}\ov{\Xi}_{,j}\ov{\Xi}_{,k}
-2\ov{R}_{jk}\big)
+4\ov{N}\big(\ov{D}_j\ov{D}_k\ov{N}\big)\Big)+2\ov{N}^2\ov{D}_j\ov{D}_k\Big(\Nv{
m}
\Big)
\Big]\delta N_{m}\nonumber\\
&&
+\Big[
\Big[\frac{\dot{\ov{N}}}{\ov{N}}-\frac{(\Qb)^{\bf\dot{}}}{\Qb}+\frac{\ov{N}}{\Qb
}\Big({\cal
L}_{\Nb}\frac{\Qb}{\ov{N}}\Big)\Big]\Big(\frac{\partial}{\partial\tau}-{\cal
L}_{\Nb}\Big)
+\Big(\ov{N}^2\Big(-\frac{\kappa}{2\Qb}\ov{C}+2\Lambda+\frac{\kappa}{2\lambda}
v(\ov{\Xi})\Big)\Big)\nonumber\\
&&
\quad
+\Big(\ov{N}^2\ov{D}_m\ov{D}_n\ov{Q}^{mn}\Big)+\Big(-\ov{N}^2\big(\ov{D}_{n}\ov{
N}\big)\ov{D}_{m}\ov{Q}^{mn}\Big)
+\Big(\Big({\cal L}_{\Nb}\big(\frac{\partial}{\partial\tau}-{\cal L}_{\Nb}\Big)
+\frac{\partial}{\partial\tau}{\cal L}_{\Nb}\Big)
\Big]\delta Q_{jk}\nonumber\\
&&
+\Big[\Big(2\ov{Q}^{mn}\Big(\dot{Q}_{n(k}-\big({\cal
L}_{\Nb}\ov{Q}\big)_{n(k}\Big)\big(\frac{\partial}{\partial\tau}-{\cal
L}_{\Nb}\big)\Big)
+ \Big(-2\ov{N}^2\ov{D}_n\ov{D}_{(k}\ov{Q}^{mn}\Big)
+\Big(2\ov{N}\big(\ov{D}_n\ov{N}\big)\ov{D}_{(k}\ov{Q}^{mn}\Big)\nonumber\\
&&
\quad
+\Big(-\frac{\ov{N}\kappa}{\Qb}\ov{H}\ov{Q}^{rm}\ov{N}^n\ov{N}^s\ov{G}_{snr(k}\Big)
\Big]\delta Q_{j)m}
+\Big[-\frac{\ov{N}\kappa}{\Qb}\ov{N}^n\ov{N}^m\ov{G}_{jkmn}\Big]\delta H\nonumber\\
&& +\Big({\cal L}_{\delta N}\big(\frac{\partial}{\partial\tau}-{\cal
L}_{\Nb}\big)+\big(\frac{\partial}{\partial\tau}-{\cal
L}_{\Nb}\big){\cal L}_{\delta N}\Big)\Big[\ov{Q}_{jk}\Big]
+\Big[-2\ov{Q}^{mn}\Big(\dot{Q}_{n(k}-\big({\cal
L}_{\Nb}\ov{Q}\big)_{n(k}\Big)\Big]\big({\cal L}_{\delta
N}\ov{Q}\big)_{j)m}. \ea
\subsection{Summary of the Equations of Motion for $\delta\Xi$ and $\delta
Q_{jk}$}
\label{SummPer} In the last two sections we derived the second order
equations of motion for $\delta\Xi$ and $\delta Q_{jk}$. The results
of our calculations can  be found in equation (\ref{ResDDdXi}) and
equation (\ref{DDdQfinal}), respectively. However, these equations
are quite complex and not very transparent in that form. For this
reason, we want to rewrite them in a form where we can still
recognize their general form but where they take a much simpler
form. We will hide the precise details of the various background
quantities that occur as coefficients in front of the linear
perturbations in certain coefficient functions that will be
introduced below. These coefficients will be operator valued since
they also involve objects such as partial or Lie derivatives, as can
be seen from eqns. (\ref{ResDDdXi}) and (\ref{DDdQfinal}). As
explained before, apart from the elementary perturbations $\delta
Q_{jk}$ and $\delta\Xi$, the second order equations of motion
contain also the perturbation of the shift vector $\delta N_j$ and
the (physical) Hamiltonian density $\delta H$. Both are functions of
$\delta Q_{jk}$ and $\delta\Xi$, and therefore not independent
perturbations. However, it turns out that these two quantities are
constants of motion, so it is convenient to keep them in the
equations.
\\
Starting with the second order equation for $\delta\Xi$, its general
structure is of the kind \be \label{SimpXi} \fbox{\parbox{6cm}{
$\big[C_{\Xi}\big]\delta\Xi=\big[C_{\Xi}\big]^{jk}\delta Q_{jk} +
\big[C_{\Xi}\big]^j\delta N_j$}} \ee where the coefficients are
given by \ba \label{XiCoeffXi} [C_{\Xi}]&:=&
\left[\frac{\partial^2}{\partial\tau^2} -\Big({\cal
L}_{\Nb}\big(\frac{\partial}{\partial\tau}-{\cal L}_{\Nb}\Big)
+\frac{\partial}{\partial\tau}{\cal L}_{\Nb}\Big)\right.
\\
&&
-\Big[\frac{\dot{\ov{N}}}{\ov{N}}-\frac{(\Qb)^{\bf\dot{}}}{
\Qb}+\frac{\ov{N}}{\Qb}\Big({\cal
L}_{\Nb}\frac{\Qb}{\ov{N}}\Big)\Big]\big(\frac{\partial}{\partial\tau}-{\cal
L}_{\Nb}\big)\nonumber\\
&&
-\ov{Q}^{jk}\Big[\frac{\ov{N}}{\Qb}\big(\ov{N}\Qb)_{,j}\Big]\frac{\partial}{
\partial x^k}\
-\ov{N}^2\Big(\Delta +[\ov{Q}^{mn}]_{,n}\frac{\partial}{\partial
x^m}\Big)+\frac{1}{2}\ov{N}^2v^{\prime\prime}(\ov{\Xi})
\left.\right]\nonumber
\ea
and
\ba
[C_{\Xi}]^{jk}&:=&
\label{XiCoeffQ}
\Big[
- \big(\dot{\ov{\Xi}}-\big({\cal
L}_{\Nb}\ov{\Xi}\big)\big)
\Big[
\Big(\frac{\partial}{\partial\tau}-{\cal
L}_{\Nb}\Big)\Big(\QN{j}{k}\Big)
+\frac{\ov{N}}{\Qb}\frac{\partial}{\partial
x^m}\Big(\frac{\Qb}{\ov{N}}\ov{Q}^{jm}\ov{N}^k\Big)
\Big]\\
&&
\quad
-\big(\frac{\partial}{\partial\tau}-{\cal
L}_{\Nb}\big)\Big(\ov{Q}^{jm}\ov{Q}^{kn}\ov{N}_m\ov{\Xi}_{,n}\Big)
-\ov{N}^2\Big(\Delta\ov{\Xi}
+[\ov{Q}^{mn}]_{,m}\ov{\Xi}_{,n}-\frac{1}{2}v^{
\prime}(\ov{\Xi})\Big)\Big(\frac{\ov{N}^j\ov{N}^k}{\ov{N}^2}\Big)
\nonumber\\
&&
\quad
-\ov{N}^2\frac{\partial}{\partial
x^n}\Big(\ov{Q}^{jm}\ov{Q}^{kn}\ov{\Xi}_{,m}\Big)
-\ov{Q}^{jm}\ov{Q}^{kn}\ov{N}_m\big[\dot{\ov{\Xi}}-\big({\cal
L}_{\Nb}\ov{\Xi}\big)\big]_{,n}
-\ov{Q}^{jm}\ov{Q}^{kn}\ov{\Xi}_{,n}\big[\frac{N}{\Q}\big(N\Q)_{,m}\big]
\nonumber\\
&&
\quad
+\Big[\frac{\dot{\ov{N}}}{\ov{N}}-\frac{(\Qb)^{\bf\dot{}}}{
\Qb}+\frac{\ov{N}}{\Qb}\Big({\cal
L}_{\Nb}\frac{\Qb}{\ov{N}}\Big)\Big]\Big(\ov{Q}^{jm}\ov{\Xi}_{,m}\ov{N}^k\Big)
\nonumber\\
&&\quad
+\big(\ov{Q}^{mn}\ov{\Xi}_{,n}\big)\Big(-\frac{\ov{N}}{\Qb}\QN{j}{k}\big[\ov{N}
\Qb]_{,j} 
\nonumber\\
&&
\quad
+\frac{\ov{N}}{\Qb}\frac{\partial}{\partial
x^m}\Big(\frac{1}{2}\Big(\ov{Q}^{jk}-\frac{\ov{N}^j\ov{N}^k}{\ov{N}^2}\Big)\ov{N
}\Qb\Big)\Big) \Big].\nonumber \ea Note that all partial and Lie
derivatives act on all terms to their right, including the linear
perturbations. This is also the reason why terms as for instance
$\ddot{\Xi}$ or $\dot{Q}_{jk}$ do not occur in the simple form of
equation (\ref{SimpXi}).
\\
The third coefficient is given by \ba \label{XiCoeffNj}
[C_{\Xi}]^j&:=& \Big[\big(\frac{\partial}{\partial\tau}-{\cal
L}_{\Nb}\big)\big(\ov{Q}^{jk}\ov{\Xi}_{,k}\big)
\big(\dot{\ov{\Xi}}-\big({\cal L}_{\Nb}\ov{\Xi}\big)\big)
+\big(\dot{\ov{\Xi}}-\big({\cal
L}_{\Nb}\ov{\Xi}\big)\big)\Big[\Big(\frac{\partial}{\partial\tau}-{\cal
L}_{\Nb}\Big)\Nv{j} +\frac{\ov{N}}{\Qb}\frac{\partial}{\partial
x^k}\Big(\frac{\Qb}{\ov{N}}\ov{Q}^{jk}\Big)
\Big]\nonumber\\
&&
\quad
+2\ov{N}^2\Big(\Delta\ov{\Xi}+[\ov{Q}^{jk}]_{,j}\ov{\Xi}_{,k}-\frac{1}{2}v^{
\prime}(\ov{\Xi})\Big)\Nv{j}
-\Big[\frac{\dot{\ov{N}}}{\ov{N}}-\frac{(\Qb)^{\bf\dot{}}}{
\Qb}+\frac{\ov{N}}{\Qb}\Big({\cal
L}_{\Nb}\frac{\Qb}{\ov{N}}\Big)\Big]\Big(\ov{Q}^{jk}\ov{\Xi}_{,k}\Big)
\nonumber\\
&&
\quad
+\big(\ov{Q}^{mn}\ov{\Xi}_{,n}\big)\Big(\frac{\ov{N}}{\Qb}\Nv{j}\big[\ov{N}\Qb]_
{,m}
+\frac{\ov{N}}{\Qb}\frac{
\partial}{\partial x^j}\Big(\Nv{j}\ov{N}\Qb\Big)\Big)
+\ov{Q}^{jk}\big[\dot{\Xi}-\big({\cal L}_{\Nb}\Xi\big)\big]_{,k}
\Big]\nonumber.\\
 \ea

As was to be expected, the corresponding equation for the
perturbation of the three metric $\delta Q_{jk}$ includes more than
just three terms. It is of the form \be \label{SimpQ}
\fbox{\parbox{13.5cm}{$ [C_{Q}]\delta
Q_{jk}=[A_Q]_{jk}\delta\Xi+[B_{Q}]_{jk}\delta H +
[C_Q]^{m}_{(k}\delta Q_{j)m} +[C_Q]^{mn}_{jk}\delta Q_{mn}
+[C_{Q}]^m_{jk}\delta N_m$}} \ee The various coefficients
introduced in the equation above are given as follows: \ba
\label{QCoeffQjk} [C_Q]&:=&
\Big[\frac{\partial^2}{\partial\tau^2} -\Big(\Big({\cal
L}_{\Nb}\big(\frac{\partial}{\partial\tau}-{\cal L}_{\Nb}\Big)
+\frac{\partial}{\partial\tau}{\cal L}_{\Nb}\Big) -\Big(\Big({\cal
L}_{\Nb}\big(\frac{\partial}{\partial\tau}-{\cal L}_{\Nb}\Big)
+\frac{\partial}{\partial\tau}{\cal L}_{\Nb}\Big)
-\Big(\ov{N}^2\ov{D}_m\ov{D}_n\ov{Q}^{mn}\Big)\nonumber
\\
&&
\quad
-\Big[\frac{\dot{\ov{N}}}{\ov{N}}-\frac{(\Qb)^{\bf\dot{}}}{\Qb}+\frac{\ov{N}}{
\Qb}\Big({\cal
L}_{\Nb}\frac{\Qb}{\ov{N}}\Big)\Big]\Big(\frac{\partial}{\partial\tau}-{\cal
L}_{\Nb}\Big)\nonumber\\
&& \quad
-\Big(\ov{N}^2\Big(-\frac{\kappa}{2\Qb}\ov{C}+2\Lambda+\frac{\kappa}{2\lambda}
v(\ov{\Xi})\Big)\Big)
-\Big(-\ov{N}^2\big(\ov{D}_{n}\ov{N}\big)\ov{D}_{m}\ov{Q}^{mn}\Big)
\Big].
 \ea
 Here we use the notation that covariant derivatives
surrounded by round bracket such as $(D_jD_k...)$ act on the
elements inside the round brackets only. By contrast, covariant
derivatives not surrounded by round brackets act on everything to
their right, including also the perturbations. The coefficient for
$\delta\Xi$ can be explicitly written as \ba \label{QCoeffXi}
[A_Q]_{jk}&:=& \Big[-\ov{Q}_{jk} \Big(
\frac{\ov{N}^2\kappa}{2\lambda}\Big(\frac{1}{\ov{N}^2}\Big(\dot{\ov{\Xi}}-\big({
\cal
L}_{\Nb}\ov{\Xi}\big)\Big)\Big(\frac{\partial}{\partial\tau}-{\cal
L}_{\Nb}\Big) +\ov{Q}^{mn}\ov{\Xi}_{,n}\frac{\partial}{\partial x^m}
-\frac{1}{2}v^{\prime}(\ov{\Xi})\Big)\Big)
+\Big(4\ov{N}^2\frac{\kappa}{\lambda}\ov{\Xi}_{(,k}\frac{\partial}{\partial
x^{j)}}\Big)
\Big].\nonumber\\
\ea For the coefficient belonging to the linear perturbation of the
Hamiltonian density $\delta H$ we get: \ba \label{QCoeffH}
[B_Q]_{jk}&:=&\Big[-\frac{\ov{N}\kappa}{\Qb}\ov{N}^n\ov{N}^m\ov{G}_{jkmn}\Big].
\ea The coefficient $[C_Q]^m_{(k}$ takes the form: \ba
\label{QCoeffQjm} [C_Q]^m_{(k}&:=&
\Big[\Big(2\ov{Q}^{mn}\Big(\dot{Q}_{n(k}-\big({\cal
L}_{\Nb}\ov{Q}\big)_{n(k}\Big)\big(\frac{\partial}{\partial\tau}-{\cal
L}_{\Nb}\big)\Big) +
\Big(-2\ov{N}^2\ov{D}_n\ov{D}_{(k}\ov{Q}^{mn}\Big)
+\Big(2\ov{N}\big(\ov{D}_n\ov{N}\big)\ov{D}_{(k}\ov{Q}^{mn}\Big)\nonumber\\
&& \quad
+\Big(-\frac{\ov{N}\kappa}{\Qb}\ov{H}\ov{Q}^{rm}\ov{N}^n\ov{N}^s\ov{G}_{snr(k}\Big)
\Big].
 \ea The last two coefficients are the ones for $\delta Q_{mn}$
and $\delta N_m$. These are the most complicated ones for the second
order equation of motion of $\delta Q_{jk}$. We will list them
below: \ba \label{QCoeffQmn} [C_Q]^{mn}_{jk}&:=& \left[\left[
\frac{\dot{\ov{N}}}{\ov{N}}-\frac{(\Qb)^{\bf\dot{}}}{\Qb}+\frac{\ov{N}}{\Qb}
\Big({\cal
L}_{\Nb}\frac{\Qb}{\ov{N}}\Big)\right]\right.\nonumber\\
&&\quad\quad\quad
\Big(\ov{Q}^{rn}\ov{N}_n[\ov{Q}_{jk}]_{,r}+\ov{Q}_{rk}\frac{\partial}{\partial
x^j}\Big(\ov{Q}^{rn}\ov{N}^m\Big)+\ov{Q}_{rj}\frac{\partial}{\partial
x^k}\Big(\ov{Q}^{rn}\ov{N}^m\Big)\nonumber\\
&&
\quad
-\Big(\dot{\ov{Q}}_{jk}-\big({\cal L}_{\Nb}\ov{Q}\big)_{jk}\Big)
\Big(\big(\frac{\partial}{\partial\tau}-{\cal L}_{\Nb}\big)\Big(\QN{m}{n}\Big)
+\frac{\ov{N}}{\Qb}\frac{\partial}{\partial
x^k}\big(\frac{\Qb}{\ov{N}}\ov{N}^m\ov{Q}^{rn}\big)\Big)\nonumber\\
&&
+\Big(-\ov{Q}^{mr}\ov{Q}^{ns}\Big(\dot{\ov{Q}}_{rj}-\big({\cal
L}_{\Nb}\ov{Q}\big)_{rj}\Big)\Big(\dot{\ov{Q}}_{sk}
-\big({\cal L}_{\Nb}\ov{Q}\big)_{sk}\Big)\Big)\nonumber\\
&&
+\Big(
-\ov{Q}_{jk}\Big[
\frac{1}{2}\frac{\ov{N}^m\ov{N}^n}{\ov{N}^2}\ov{N}^2\Big(2\Lambda-\ov{R}+\frac{
\kappa}{2\lambda}\Big(v(\ov{\Xi})-\ov{Q}^{rs}\ov{\Xi}_{,r}\ov{\Xi}_{,s}\Big)
\Big)
+\frac{\ov{N}^2}{2}\Big(\ov{R}^{mn}
-[\ov{G}^{-1}]^{mnrs}\ov{D}_r\ov{D}_s\Big)\nonumber\\
&&
\quad\quad\quad\quad\quad
+\frac{\kappa}{2\lambda}\ov{Q}^{mr}\ov{Q}^{ns}\big(\dot{\Xi}-\big({\cal
L}_{\Nb}\ov{\Xi}\big)\big)\ov{N}_r\ov{\Xi}_{,s}
-\frac{\kappa\ov{N}^2}{2\lambda}\ov{Q}^{rm}\ov{Q}^{sn}\ov{\Xi}_{,r}\ov{\Xi}_{,s}
\nonumber\\
&&
\quad\quad\quad\quad\quad
+\frac{1}{4}\Big(\dot{\ov{Q}}_{rs}-\big({\cal L}_{\Nb}\ov{Q}\big)_{rs}\Big)
\Big[[\ov{G}^{-1}]^{mnrs}\Big(\frac{\partial}{\partial\tau}-{\cal L}_{\Nb}\Big)
-\Big(\dot{\ov{Q}}_{tu}-\big({\cal
L}_{\Nb}\ov{Q}\big)_{tu}\Big)\ov{Q}^{sn}[\ov{G}^{-1}]^{turm}\Big]\nonumber\\
&&
\quad\quad\quad\quad\quad
+\frac{1}{4}\Big(\dot{\ov{Q}}_{rs}-\big({\cal L}_{\Nb}Q\big)_{rs}\Big)
[\overline{G}^{-1}]^{rsvw}
\Big[\ov{Q}^{mt}\ov{N}^n[\ov{Q}_{vw}]_{,t}
+2\overline{Q}_{tw}\frac{\partial}{\partial x^v}\Big(\ov{Q}^{mt}\ov{N}^n\Big)\Big]\Big]\Big)\nonumber\\
&&
+\Big(
\frac{1}{2}\frac{\ov{N}^m\ov{N}^n}{\ov{N}^2}\Big(-2\ov{N}^2\big(\frac{\kappa}{
\lambda}\ov{\Xi}_{,j}\ov{\Xi}_{,k}-2\ov{R}_{jk}\big)-4\ov{N}\big(\ov{D}_j\ov{D}
_k\ov{N}\big)\Big)
+\ov{N}^2\ov{D}_j\ov{D}_k\ov{Q}^{mn}
-2\ov{N}^2\ov{D}_j\ov{D}_k\Big(\frac{1}{2}\frac{\ov{N}^m\ov{N}^n}{\ov{N}^2}
\Big)\Big)\nonumber\\
&&
+\Big(\QN{m}{n}\frac{\ov{N}\kappa}{\Qb}\ov{H}\ov{N}^r\ov{N}^s\ov{G}_{jkrs}
+\frac{\ov{N}\kappa}{\Qb}\ov{H}\ov{Q}^{sm}\ov{N}^r\ov{N}^n\ov{G}_{jkrs}
\Big)\nonumber\\
&&
\quad
+\Big(
-\ov{Q}^{mr}\ov{N}^n\Big[\dot{\ov{Q}}_{jk}-\big({\cal
L}_{\Nb}\ov{Q}\big)_{jk}\Big]_{,r}
-\big(\dot{\ov{Q}}_{km}-\big({\cal
L}_{\Nb}\ov{Q}\big)_{kr}\big)\frac{\partial}{\partial
x^j}\Big(\ov{Q}^{mr}\ov{N}^n\Big)\nonumber\\
&&
\quad
-\big(\dot{\ov{Q}}_{jm}-\big({\cal
L}_{\Nb}\ov{Q}\big)_{jr}\big)\frac{\partial}{\partial
x^k}\Big(\ov{Q}^{mr}\ov{N}^n\Big)\nonumber\\
&&
\quad
+\Big(-2\ov{Q}^{tu}\Big(\dot{\ov{Q}}_{t(k}-\big({\cal
L}_{\Nb}\ov{Q}\big)_{t(k}\Big)\Big)
\Big(-\ov{Q}^{mr}\ov{N}^n\Big[Q_{j)u}\Big]_{,r}
-\ov{Q}_{um}\frac{\partial}{\partial x^{j)}}\Big(\ov{Q}^{mr}\ov{N}^n\Big)
-\ov{Q}_{j)m}\frac{\partial}{\partial
x^u}\Big(\ov{Q}^{mr}\ov{N}^n\Big)\Big)\nonumber\\
&&\left. \quad +\Big(\frac{\partial}{\partial\tau}-{\cal L}_{\Nb}\Big)
\Big(-\ov{Q}^{mr}\ov{N}^n\Big[\ov{Q}_{jk}\Big]_{,r}
-\ov{Q}_{km}\frac{\partial}{\partial
x^j}\Big(\ov{Q}^{mr}\ov{N}^n\Big)
-\ov{Q}_{jm}\frac{\partial}{\partial
x^k}\Big(\ov{Q}^{mr}\ov{N}^n\Big)\Big) \Big) \right].
 \ea
Note that the coefficient for $\delta Q_{mn}$ in equation
(\ref{DDdQfinal}) and the one above are different due to the
presence of the last two lines in the equation above. The reason for
this is that now we used the explicit expression for the Lie
derivatives with respect to $\delta\vec{N}$, which were derived in
equation (\ref{Per2Lie}) and (\ref{Per6Lie}) and lead to additional
terms in $\delta Q_{mn}$ and $\delta N_m$.
\\
Finally, we present the coefficient for $\delta N_m$. Similarly to
the case of $\delta Q_{nm}$, we also get additional terms coming
from the Lie derivatives in the last line of equation
(\ref{DDdQfinal}).
 \ba \label{QCoeffNm} [C_Q]^{m}_{jk}&:=& \Big[
-\Big[\frac{\dot{\ov{N}}}{\ov{N}}-\frac{(\Qb)^{\bf\dot{}}}{\Qb}+\frac{\ov{N}}{
\Qb}\Big({\cal
L}_{\Nb}\frac{\Qb}{\ov{N}}\Big)\Big]
\Big(\ov{Q}^{mn}[\ov{Q}_{jk}]_{,n}+\ov{Q}_{nk}\frac{\partial}{\partial
x^j}\Big(\ov{Q}^{mn}\Big)+\ov{Q}_{nj}\frac{\partial}{\partial
x^k}\Big(\ov{Q}^{mn}\Big)\Big)\nonumber\\
&&
\quad
+\Big(\dot{\ov{Q}}_{jk}-\big({\cal L}_{\Nb}\ov{Q}\big)_{jk}\Big)
\Big(\big(\frac{\partial}{\partial\tau}-{\cal L}_{\Nb}\big)\Big(\Nv{m}\Big)
+\frac{\ov{N}}{\Qb}\frac{\partial}{\partial
x^n}\big(\frac{\Qb}{\ov{N}}\ov{Q}^{mn}\big)\Big)
\nonumber\\
&&
\quad
+\Big(
-\ov{Q}_{jk}\Big[
-\Nv{m}\ov{N}^2\Big(2\Lambda-\ov{R}-\frac{\kappa}{2\lambda}\big(v(\Xi)-\ov{Q}^{
rs}\ov{\Xi}_{,r}\ov{\Xi}_{,s}\big)\Big)
-\frac{\kappa}{2\lambda}\big(\dot{\Xi}-\big({\cal
L}_{\Nb}\ov{\Xi}\big)\big)\ov{Q}^{mn}\ov{\Xi}_{,n}
\nonumber\\
&&
\quad\quad\quad\quad\quad\quad
-\frac{1}{4}\Big(\dot{\ov{Q}}_{rs}-\big({\cal L}_{\Nb}Q\big)_{rs}\Big)
[\ov{G}^{-1}]^{rsvw}
\Big[\ov{Q}^{mt}[\ov{Q}_{vw}]_{,t}
+2\ov{Q}_{tw}\frac{\partial}{\partial x^v}\Big(\ov{Q}^{mt}\Big)\Big]\Big]\Big)
\nonumber\\
&&
\quad
+\Big(-\Nv{m}\frac{\ov{N}\kappa}{\Qb}\ov{H}\ov{N}^r\ov{N}^s\ov{G}_{jkrs}\Big)
\nonumber\\
&&
\quad
+\Big(\Nv{m}\Big(2\ov{N}^2\big(\frac{\kappa}{\lambda}\ov{\Xi}_{,j}\ov{\Xi}_{,k}
-2\ov{R}_{jk}\big)
+4\ov{N}\big(\ov{D}_j\ov{D}_k\ov{N}\big)\Big)+2\ov{N}^2\ov{D}_j\ov{D}
_k\Big(\Nv{m}\Big)
\Big)\nonumber\\
&&
\quad
+\Big(
\ov{Q}^{mn}\Big[\dot{\ov{Q}}_{jk}-\big({\cal L}_{\Nb}\ov{Q}\big)_{jk}\Big]_{,n}
+\big(\dot{\ov{Q}}_{kn}-\big({\cal
L}_{\Nb}\ov{Q}\big)_{kn}\big)\frac{\partial}{\partial x^j}\Big(\ov{Q}^{mn}\Big)
+\big(\dot{\ov{Q}}_{jn}-\big({\cal
L}_{\Nb}\ov{Q}\big)_{jn}\big)\frac{\partial}{\partial
x^k}\Big(\ov{Q}^{mn}\Big)\Big)\nonumber\\
&&
\quad
+\Big(\frac{\partial}{\partial\tau}-{\cal L}_{\Nb}\Big)
\Big(\ov{Q}^{mn}\Big[\ov{Q}_{jk}\Big]_{,n}
+\ov{Q}_{kn}\frac{\partial}{\partial x^j}\Big(\ov{Q}^{mn}\Big)
+\ov{Q}_{jn}\frac{\partial}{\partial x^k}\Big(\ov{Q}^{mn}\Big)\Big)\nonumber\\
&& \quad +\Big(-2\ov{Q}^{tu}\Big(\dot{\ov{Q}}_{t(k}-\big({\cal
L}_{\Nb}\ov{Q}\big)_{t(k}\Big)\Big)
\Big(\ov{Q}^{mn}\Big[\ov{Q}_{j)u}\Big]_{,n}
+\ov{Q}_{un}\frac{\partial}{\partial x^{j)}}\Big(\ov{Q}^{mn}\Big)
+\ov{Q}_{j)n}\frac{\partial}{\partial x^u}\Big(\ov{Q}^{mn}\Big)\Big)
\Big].
 \ea
Although the form of the perturbed equations is quite complicated,
they simplify drastically for special backgrounds of interest. For
the case of FRW, for instance, all terms proportional to $\ov{N}_j$
vanish, since $\ov{N}_j=-\ov{C}_j/\ov{H}=0$ for FRW. This is due to the
geometry and matter parts of the diffeomorphism constraint vanishing
both separately in that case. Furthermore, all terms in the
coefficients that contain spatial derivatives applied to background
quantities vanish also. Other backgrounds where considerable
simplification will occur include Schwarzschild spacetime.

\section{Comparison with Other Frameworks}
\label{s8}

We now proceed to compare our work with other approaches to general
relativistic perturbation theory found in the literature. In the
following we will restrict ourselves to discussing works that treat
perturbation theory around general backgrounds. Approaches which
deal exclusively with cosmological perturbation theory will be
looked at in our second paper, specifically dedicated to that topic.

The central point of comparison is how gauge-invariance is handled
in the various approaches. As that notion often acquires different
meanings, especially in the context of general-relativistic
perturbation theory, it seems prudent to recall the precise
mathematical setting underlying most works. It was developed by
Sachs \cite{oliver1} and Stewart and Walker \cite{oliver2}, and
recently given a very general and elegant formulation by Bruni,
Sonego and collaborators \cite{oliver4,oliver5,oliver6}. The
starting point consists of two spacetime manifolds $M_0$ and $M$,
which represent the background spacetime around which one perturbs
and the actual physical spacetime, respectively. It is important to
keep in mind that $M_0$ is only an artificial construct.
Perturbations of geometric quantities are then defined by comparing
their values in $M_0$ and $M$, respectively. This procedure is
highly ambiguous, however, in that there is a great freedom inherent
in the choice of points where one compares background and "real"
quantities to define the perturbations. Note that this freedom is in
addition to the usual coordinate gauge freedom in general
relativity. For that reason, Sachs termed it gauge invariance of the
second type \cite{oliver1}. Making such a choice, which
mathematically amounts to choosing a so-called point identification
map between $M_0$ and $M$, is therefore nothing but a choice of
gauge. Correspondingly, gauge-invariant perturbations are those
quantities whose values do not depend on the choice of point
identification map. The actual condition for a perturbed quantity to
be gauge-invariant to first order in this sense was already derived
in \cite{oliver1}, fully proved in \cite{oliver2} and finally
generalized to arbitrary order $n$ in \cite{oliver4}. The result can
be succinctly summarised as follows: a geometric quantity $T$ - such
as a tensor - is invariant to order $n$ iff all its perturbations to
order $n-1$ are either vanishing, (spacetime) constant scalars or a
combination of Kronecker deltas. This result is often known as the
Stewart\&Walker lemma. Clearly, the only case of actual interest is
the first. An example is given by curvature tensors in linear
perturbation theory around Minkowski space. As they vanish in the
background, they are gauge-invariant to first order. These insights
have been the backbone of most attempts to construct gauge-invariant
quantities to various orders in perturbation theory or even
non-perturbatively. We will now discuss two of them.

In a series of papers \cite{oliver7,oliver8,oliver9} Nakamura has
used these principles to develop formulas for gauge-invariant
quantities to second and third order around an arbitrary background.
They encompass the invariant parts of various metric and curvature,
as well as matter perturbations. These general formulas are,
however, implicit only to the extent that Nakamura derives them from
the assumption that the corresponding linear order perturbations can
be decomposed into gauge-invariant and gauge-variant parts.
Consequently, while the construction is, in principle, valid for
arbitrary backgrounds, in practice only backgrounds with sufficient
symmetries to perform that split at linear order explicitly can be
used. Luckily, that applies, of course, to several cases of great
interest, such as cosmological and spherically symmetric
backgrounds. The latter case is explicitly treated in
\cite{oliver9}.

A distinctly different approach - let us call it the EB approach -
is based on seminal work by Ellis and Bruni \cite{oliver10}, which
has since inspired a multitude of other works
\cite{oliver11,oliver12,oliver13}. Although the discussions in these
papers is geared towards applications in cosmology, the framework
itself can be applied to arbitrary spacetimes, in principle, which
is why we decided to discuss it here. The basic idea is to use a
$1+3$-approach by employing covariant quantities connected to a
family of flow lines or "fundamental observers". The prime reason is
that these quantities are much more closely related to what one
actually measures in astrophysics. Furthermore, by a simple
application of the Stewart-Walker lemma, they are automatically
gauge-invariant if the corresponding quantities in the background
spacetime vanish. Unlike in the more common metric-based formalism,
these quantities are defined in the physical, perturbed spacetime.
As a result, they are fully non-perturbative and in that sense their
gauge-invariance extends to all orders. The connection to the
standard perturbative approach based on perturbations in the
background spacetime can be made by suitably expanding the physical
quantities to the desired order, see \cite{oliver11}. This approach
enjoys a clear geometric and physical interpretation of the
quantities used, as well as the advantage of basing perturbation
theory on non-perturbative variables.

Comparing the works mentioned so far to ours, a first obvious
difference is that we work in a canonical setting versus the
covariant setting used by the others. The motivation is first that
gauge issues become particularly clear in the canonical picture and
second our view towards quantization. The more important difference,
however, is our use of dust as a dynamically implemented coordinate
system. Our dust clocks serve a twofold purpose: on the one hand
they enable us to construct background observables and therefore to
solve the standard gauge problem in general relativity. On the other
hand they also serve as the point identification map and thus
eliminate the gauge freedom of "second type". In that sense they
represent a logical extension to perturbation theory of the initial
conceptual idea by Brown and Kuchar \cite{3} to use dust as a
physical and therefore preferred coordinate system.

We should also point out that while our framework employs the metric
and its perturbations as fundamental variables, we could equally
well use the dust clocks to build a gauge-invariant perturbation
theory based on the same variables used in the EB approach. In fact,
it seems worthwhile to look a bit more closely at the relationship
between the two approaches. Both are non-perturbative in the
following sense: they construct quantities which are gauge-invariant
(with respect to gauge transformations of the second type). Only
then perturbation theory is applied, which means that
gauge-invariance is then automatically guaranteed in each order of
perturbation. The difference arises when one looks at the role of
gauge transformations of the first kind. The EB approach uses
idealized observers that are comoving with the physical matter in
the model. Thus the theory is not deparametrized and gauge freedom
with respect to the background spacetime remains, as illustrated by
the presence of constraints as part of the equations of motion. In
our case, the observers are represented by the dust, a component
added to the physical matter content of the theory. They are thus
dynamically included in the theory via the dust contribution to the
Lagrangian, in addition to all the other matter. This allows for a
complete deparametrization of the non-dust system with respect to
the dust. Time evolution for this subsystem becomes unconstrained
and a physical Hamiltonian emerges. The price to pay for this is
that the dust contributes to the energy-momentum tensor of the
deparametrized system, the size of which is small, however. One
might well argue that, for practical purposes at the classical
level, the choice between the two approaches is a matter of taste.
Our approach, however, offers a clear advantage if one is interested
in quantization, eventually. All programs aiming at a quantization
of gravity that have been pursued, so far, use the metric as a
fundamental building block. It is not obvious to us how to attempt a
quantization based on the covariant variables used in the EB
approach.
\\
Another argument in favour of our framework is the following: 
A test observer which by definition does not have any impact is 
only a mathematical idealisation. Physically much more realistic is 
a dynamically coupled observer fluid like the dust considered here which 
in particular takes into account the gravitational backreaction.

To summarise: the crucial difference between our approach and the
others discussed here is that the latter deal only with gauge
freedom of the second type. This can be seen from the fact that they
use background variables which are not gauge-invariant and is
evidenced, e.g., by the presence of constraint equations in addition
to evolution equations. Our treatment, by contrast, deals with all
variables at all orders in a unified manner. Furthermore, in our
opinion the framework developed here allows for a much more
straightforward implementation at higher orders. Another advantage
is that it allows to treat arbitrary backgrounds in practice,
without the high degree of symmetry necessary for approaches based
on the Stewart-Walker lemma to work. Recall that the latter require
finding non-trivial quantities that vanish in the background
manifold. Only in symmetric backgrounds such as homogeneous
spacetimes is that a fairly tractable problem.

Finally, we should briefly discuss the recent work in \cite{18},
which is close to ours in terms of motivation and conceptual
underpinnings. The authors there also use the general gauge
invariant framework of \cite{6}, however, with two differences:
First of all, they do not use dust matter to achieve gauge invariant
completions of geometry and matter variables. This prevents them
from bringing the constraints into a deparametrised form \cite{8}
and thus there is no time independent physical Hamiltonian.
Secondly, while they can develop higher order cosmological
perturbation theory, their perturbations of gauge invariant
quantities are still expanded in terms of the perturbations of the
the gauge variant quantities which is what we never do. Therefore
the basic perturbation variables are different in the two schemes:
In our scheme we never care how our gauge invariant variables are
assembled from gauge variant ones, they and their perturbations are
fundamental for us and n'th order quantities are n'th order
expressions in those. In contrast, in \cite{18} n'th order means n'th
order in the gauge variant quantities. In particular, the n'th order
perturbed variables are only invariant with respect to the n'th
order constraints up to terms of higher order. In contrast, our
perturbed variables are always first order and always fully gauge
invariant, it is only in the Hamiltonian that higher orders of gauge
invariant variables appear. It would no doubt be fruitful to
translate the schemes into each other and to see which differences
and
similarities arise.

\section{Conclusions and Open Questions}
\label{s9}

This is a long and technically involved paper. The reader rightfully 
will ask why one should dive into its details and what exactly is novel 
as compared to the existing literature. The following remarks are in 
order: 
\begin{itemize}
\item[1.] {\it Non perturbative gauge invariance}\\
To the best of our knowledge, there exists no generally accepted 
notion of gauge invariance at n-th order of perturbation theory 
in General Relativity. Moreover, at each order of perturbation 
theory one has to repeat the analysis for how to preserve gauge 
invariance to the given order.
Given those difficulties, it is natural to try to invent a scheme 
which separates the issue of gauge invariance from the perturbation 
theory. Hence, one must treat gauge invariance non perturbatively.
This is exactly what we managed to do in this paper.
 
Thus, one only deals with the exact observables of the theory. 
All the equations of the theory have to be written in terms of those 
gauge invariant quantities. Given such a gauge invariant function $O$
on the full phase space, we evaluate it on a certain background 
(data) which 
is an exact solution to our equations of motion and get a certain value 
$O_0$. The perturbation of $O$ is then defined as $\delta O=O-O_0$. We 
never care to expand $\delta O$ in terms of the perturbations of the 
gauge variant degrees of freedom (although we could). However,
all the equations are 
expanded directly in terms of the perturbations of those physical 
observables. 
\item[2.] {\it Material reference system}\\
In General Relativity it is well known that in order to meaningfully talk 
about the Einstein equations and to have them describe something 
observable or measurable, one has to suppose that spacetime is inhabited 
by (geodesic) test observers. By definition, a test observer has no 
effect whatsoever on the system. This is of course mathematically 
convenient but physically worrysome because a test observer is a 
mathematical idealisation. Every real observer interacts at least 
gravitationally and does leave its fingerprint on the system. One of the 
achievements of the seminal work \cite{3} of Brown \& Kucha{\v r}, which 
in 
our mind is insufficiently appreciated in the literature, is to have 
overcome this shortcoming. The authors of \cite{3} have identified a 
generally covariant
Lagrangian which comes as close as possible to describing a non -- self 
interacting, perfect and geodesically moving fluid that fills out 
spacetime (congruence). It does leave its fingerprint on the system and 
thus is 
physically more realistic than the test observer fluid.

In this paper we have driven the work of \cite{3} to its logical 
frontier and have asked the question whether the dust when added to 
the geometry -- matter system really accomplishes the goal of keeping the 
approximate validity of the usual interpretation of the Einstein 
equations. We have verified that it does which in our mind is an
intriguing result. 
\item[3.] {\it Solving the Problem of Time}\\
Since General Relativity is a generally covariant or reparametrisation 
invariant theory, it is not equipped with a natural Hamiltonian. Rather,
the ``dynamics'' of the non observables is described by a linear
combination of constraints which really generate gauge transformations
rather than physical evolution.
Observable quantities are gauge invariant and therefore do not evolve with 
respect to the ``gauge dynamics''. Therefore it is conceptually unclear 
what to do with those observables. The achievement of \cite{5,6} is to 
have invented a scheme that in principle unfreezes the observables from 
their non motion. However, in general that physical motion is far from 
uniquely selected, there are in general infinitely many such physical 
notions of time and none of them is preferred. Moreover, the associated 
Hamiltonians are generically neither preserved nor positive or at least 
bounded from below.

When combining the frameworks of \cite{3} and \cite{5,6} we find the 
remarkable result that there is a preferred Hamiltonian which is 
manifestly positive, not explicitly dependent on physical time and gauge 
invariant. It maps a conceptually complicated gauge systems into the 
safe realm of a conservative Hamiltonian system. The physical 
Hamiltonian drives the evolution 
of 
the physical observables. It reproduces the Einstein equations for the 
gauge invariant observables up to corrections which describe the influence 
of the dust. 
\item[4.] {\it Counting of the Physical Degrees of Freedom}\\
The price to pay is that one has to assume the existence of the 
dust as an additional matter species next to those of the standard 
model\footnote{
Curiously, what we have done in this paper bears some resemblance with the 
Stueckelberg formalism. There one also adds 
additional matter to Maxwell theory. One can then make the 
longitudinal mode gauge invariant and thus finds a theory with one 
more degree of freedom. This is one way to arrive at the Proca theory
and more generally at massive vector boson theories via the Higgs 
mechanism.}.
It would maybe be more desirable to have matter species of the standard 
model or geometry modes playing the role of a dynamical test 
observer\footnote{This would be similar to technicolour theories 
which declare the Higgs scalar field not as an independent degree of 
freedom but as a compound object built from the bosons of the electroweak
theory. Here one would build four independent scalars e.g. from the 
geometry field.}
. In principle this is possible, however, the resulting formalism is 
much more complicated and it does not lead to 
deparametrisation. Thus a conserved physical Hamiltonian would then not 
be available and the equations of motion would become intractable. 

It is true 
that the dust variables disappear in the final description of the 
observables which are complicated aggregates built out of all fields.
However, the theory has fundamentally four more physical degrees of 
freedom than without dust and that might eventually rule out our theory
if those additional degrees of freedom are not observed. 

The truly remarkable feature of the dust is that
it replaces the initial value constraints of General Relativity which are 
reponsible for having less physical degrees of freedom than one would 
naively expect, by four conservation laws. That is, in any given solution 
of our equations of motion, the physical observables must physically 
evolve with respect to each other in such a way that the conserved 
quantities do not change. This effectively acts like a constraint and 
therefore reduces the number of independently evolving observables 
by four, in agreement with the counting of the degrees of freedom 
without dust. Thus, at least as long as the value of those conserved 
quantities is sufficiently small, we will not be able to see those 
additional degrees of freedom. It is this fact which makes it possible 
that one effectively does not see more degrees of freedom than in the 
standard treatment.

As a final objection against our formalism one might raise the fact that 
the dust contributes 
with the wrong sign to the matter energy momentum tensor. However, the 
formalism not only forces us to do this as we would otherwise have a 
negative definite Hamiltonian, moreover, as already remarked in the 
introduction, it is completely acceptable since here we talk about 
the energy momentum tensor of non observables. The energy momentum tensor 
of the observables in the final Einstein equations does satisfy the 
usual 
energy conditions.
\item[5.] {\it Complexity of the Equations of Motion}\\
Since General Relativity is a highly non linear, complicated self -- 
interacting theory, experience from much less complicated integrable 
systems suggests that its Invariants, that is, the 
gauge invariant observables, satisfy a tremendously complicated Poisson 
algebra and that the equations of motion are intractable. Surprisingly, 
this is not at all the case. The observable algebra is almost as simple
as the algebra of non observables and the equations of motion can be 
solved almost as easily as in the usual gauge variant formalism.
Key to that is the presence of the already mentioned conserved current.
\end{itemize}
In this first paper we have developed the general gauge invariant 
formalism 
and linear perturbation theory about general backgrounds. In the companion
paper we apply the general results to flat and FRW backgrounds and find 
agreement with usual linear perturbation theory for linearly invariant 
observables. This is a first consistency test that our theory has passed.  
Thus we hope to have convinced the reader that the present framework has 
conceptual advantages over previous ones and that it is nonetheless 
technically not much more complicated. Actually, the pay -- off for having a 
manifestly gauge invariant approach will really come in at higher order 
where we believe that our equations of motion will be simpler. 

There are many lines of investigations that one can follow from
here. An obvious one, the specialization to the all-important case
of an FRW background is presented in a companion paper as already 
mentioned.
Investigating perturbations around backgrounds of astrophysical
interest, such as Schwarzschild spacetime, is also valuable. Again, for 
all
these cases it should also prove very interesting to go beyond the
linear to higher orders. Predictions from second-order perturbation
theory, e.g. the issue of non-Gaussianity in cosmological
perturbations, are the topic of current research and could soon be
testable by future experiments such as PLANCK, see \cite{oliver20}.
On a more technical and conceptual level, our framework might prove
useful to settle the issue of under what conditions
general-relativistic perturbation theory is consistent and stable.
Finally, with a view towards facilitating the search for physically
relevant predictions from approaches to quantum gravity, a
quantization of our gauge-invariant formulation of general
relativity, together with the development of the corresponding
perturbation theory at the quantum level strikes us as a highly
desirable goal. See \cite{17} for first steps.\\
\\
\\
\\
{\large\bf Acknowledgements}\\
\\
K.G. is grateful to Perimeter Institute for hospitality and
financial support. Research performed at Perimeter Institute
Theoretical Physics is supported in part by the Government of Canada
through NSERC and by the Province of Ontario through MRI. O.W. was
partially supported by the Province of Ontario through an ERA award,
ER06-02-298.

\begin{appendix}

\section{Second Class Constraints of the Brown -- Kucha{\v r} Theory}
\label{sa}

 In this section we provide the calculational details of the constraint 
 analysis of the Brown -- Kucha{\v r} theory discussed in section 
 \ref{s2}. In particular we want to show that no tertiary constraints 
 arise. Our starting point is equation (\ref{2.21}) which  
we display once again below for the convenience of the reader
\ba
\label{2ndCons}
z_{,t} &=& \{H_{\rm primary},p\}=-c^{\rm tot}
\nonumber\\
z_{a,t} &=& \{H_{\rm primary},p_a\}=-c_a^{\rm tot}
\nonumber\\
Z_{,t} &=& 
\{H_{\rm primary},I\}=\frac{n}{2}[-\frac{P^2}{\sqrt{\det(q)}\rho^2}
+\sqrt{\det(q)}(q^{ab} U_a U_b+1)]=:\tilde{c}
\nonumber\\
Z^j_{,t} &=& \{H_{\rm primary},I^j\}=-\mu^j P -n\rho \sqrt{\det(q)} 
q^{ab} U_a S^j_{,b} + P S^j_{,a} n^a
\nonumber\\
Z_{j,t} &=& \{H_{\rm primary},P_j+W_j P\}=\mu_j P - (n^a-\frac{n \rho 
\sqrt{\det(q)}}{P} q^{ab} U_b) P W_{j,a} 
\ea
\\
\\
 The last two equations  involve the Lagrange multipliers $\mu^j$ and 
$\mu_j$ and can be solved for them. In contrast we observe that the first three
equation are independent of the 
 Lagrange multipliers, they are secondary constraints.
 We will now proceed with the constraint analysis and show that when the 
 Poisson brackets between the primary Hamiltonian $H_{\rm primary}$ and the 
secondary constraints are consideres no new constraints are generated.
Recall that the primary Hamiltonian density was given by
\ba
 h_{\rm primary} &=& \mu^j Z_j+\mu Z+\mu_j Z^j+\nu z+\nu^a z_a+n^{\prime} c^{\rm
tot}+n^{\prime a} 
c^{\rm tot}_a
\ea
whereby the single constraints are shown below
\ba 
\label{AllCons}
c^{\rm tot} &=& c^{\rm geo}+c^{matter}+c^{\rm D} \nonumber\\
c_a^{\rm tot} &=& c^{\rm geo}_a+c^{matter}_a+c^{\rm D}_a \nonumber\\
\kappa c^{\rm geo} &=& \frac{1}{\sqrt{\det(q)}}[P^{ab} P_{ab} -\frac{1}{2} 
(P_a^a)^2]-\sqrt{\det(q)}\; R^{(3)}+2\Lambda \sqrt{\det(q)} 
\nonumber\\
\lambda c^{matter} 
&=& \frac{1}{2}[\frac{\pi^2}{\sqrt{\det(q)}}+\sqrt{\det(q)}\big(q^{ab}\xi_{,a} 
\xi_{,b}+v(\xi)\big)]
\nonumber\\
c^{\rm D} &=& 
\frac{1}{2}[\frac{P^2}{\sqrt{\det(q)}\rho}+\sqrt{\det(q)}\rho(q^{ab} U_a 
U_b+1)] \nonumber\\
\kappa c_a^{\rm geo} &=& -2D_b P^b_a
\nonumber\\
\lambda c_a^{matter} &=& \pi \xi_{,a}
\nonumber\\
c^{\rm D}_a &=& P[T_{,a}-W_j S^j_{,a}]
\ea 
 We begin with the calculation of the Poisson bracket of $H_{\rm primary}$ 
and the smeared constraint 
\be
\vec{c}^{\rm tot}(\vec{n}):=\int \;d^3x \;n^a(x) c_a^{\rm tot}(x)
\ee
\ba
\label{Hprimcatot}
\{H_{\rm primary},\vec{c}^{\rm tot}(\vec{n})\}
&=&
\int d^3x\int d^3y\Big(
\{\mu^j(y)Z_j(y),n^a(x) c_a^{\rm tot}(x)\}
+\{\mu(y)Z(y),n^a(x) c_a^{\rm tot}(x)\}\nonumber\\
&&
+\{\mu_j(y)Z^j(y),n^a(x) c_a^{\rm tot}(x)\}
+\{\nu(y)z(y),n^a(x) c_a^{\rm tot}(x)\}
+\{\nu^b(y)z_b(y),n^a(x) c_a^{\rm tot}(x)\}\nonumber\\
&&
+\{n^{\prime}(y)c^{\rm tot}(y),n^a(x) c_a^{\rm tot}(x)\}
+\{n^{\prime b}(y)c^{\rm tot}_b(y),n^a(x) c_a^{\rm tot}(x)\}\Big)
\ea
 For the single Poisson brackets that occur in the equation above we 
obtain the following result:
\ba
 \int d^3x\int d^3y\{\mu^j(y)Z_j(y),n^a(x) c_a^{\rm tot}(x)\}&=&\int d^3x 
n^aP\mu^jW_{j,a}
\\
\int d^3x\int d^3y\{\mu(y)Z(y),n^a(x) c_a^{\rm tot}(x)\}&=&0\nonumber\\
 \int d^3x\int d^3y\{\mu_j(y)Z^j(y),n^a(x) c_a^{\rm tot}(x)\}&=&-\int d^3x 
P\mu_jS^j_{,a}n^a\nonumber\\
\int d^3x\int d^3y\{\nu(y)z(y),n^a(x) c_a^{\rm tot}(x)\}&=&0\nonumber\\
 \int d^3x\int d^3y\{\nu^b(y)z_b(y),n^a(x) 
c_a^{\rm tot}(x)\}&=&\vec{c}^{\rm tot}(\vec{\nu})
\nonumber\\
 \int d^3x\int d^3y\{n^{\prime}(y)c^{\rm tot}(y),n^a(x) 
c_a^{\rm tot}(x)\}&=&c^{\rm tot}({\cal L}_{\vec{n}}n^{\prime})
+\tilde{c}({\cal L}_{\vec n^{\prime}}\rho)
 +\int d^3x \rho\, n^{\prime}\, n^a 
\sqrt{\det(q)}q^{bc}U_b\big(W_{j,a}S^j_c-W_{j,c}S^j_{,a}\big)
\nonumber\\
\int d^3x\int d^3y\{n^{\prime b}(y)c^{\rm tot}_b(y),n^a(x) c_a^{\rm tot}(x)\}&=&
 \vec{c}({\cal L}_{\vec{n}}\vec{n}^{\prime})-\int 
d^3x\big(n^{\prime a}n^b - n^an^{\prime b}\big)PS^j_{,a}W_{j,b}
\nonumber
\ea
 Consequently we can rewrite equation (\ref{Hprimcatot}) as
\ba
\{H_{\rm primary},\vec{c}^{\rm tot}(\vec{n})\}
&=&
c^{\rm tot}({\cal L}_{\vec{n}}n^{\prime})+\tilde{c}({\cal L}_{\vec{n}}\rho)
 +\int d^3x \rho\, n\, n^a 
\sqrt{\det(q)}q^{bc}U_b\big(W_{j,a}S^j_c-w_{j,c}S^j_a\big)\nonumber\\
&&
 -\int d^3x \mu_jPS^j_an^a +\int d^3x \mu^jPw_{j,a}n^a +\vec{c}({\cal 
L}_{\vec{n}}\vec{n}^{\prime})+\vec{c}^{\rm tot}(\vec{\nu})\nonumber\\
&&
-\int d^3x \big(n^{\prime b}n^a-n^{\prime a}n^b\big)PW_{j,a}S^j_{,b}\nonumber\\
&\approx&
 \int d^3x\, n^{a}\Big(S^j_a\big(Pn^{\prime b}W_{j,b}-\rho\, 
n\,\sqrt{\det(q)}q^{bc}U_bw_{j,c}-\mu_jP\big)\nonumber\\
&&\quad\quad
 +W_{j,a}\big(-Pn^{\prime b}S^j_{,b}+\rho\, n\, 
\sqrt{\det(q)}q^{bc}U_bS^j_{,c}+\mu^jP\big)\Big)
\ea
 Hence, the result above involves the Langrange multipliers $\mu_j$ and 
 $\mu^j$ and can be solved for them such that no new constraints arise 
from $c_a^{\rm tot}$. 
Proceeding with $c^{\rm tot}$ whereby the smeared constraint is given by
\be
c^{\rm tot}(n):=\int d^3x n(x) c^{\rm tot}(x)
\ee
Thus we get
\ba
\label{Hprimc}
\{H_{\rm primary},c^{\rm tot}(n)\}
&=&
\int d^3x\int d^3y\Big(
\{\mu^j(y)Z_j(y),n(x) c^{\rm tot}(x)\}
+\{\mu(y)Z(y),n(x) \tilde{c}(x)\}\nonumber\\
&&
+\{\mu_j(y)Z^j(y),n(x) c^{\rm tot}(x)\}
+\{\nu(y)z(y),n(x) c^{\rm tot}(x)\}
+\{\nu^b(y)z_b(y),n(x) c^{\rm tot}(x)\}\nonumber\\
&&
+\{n^{\prime}(y)c^{\rm tot}(y),n(x) c^{\rm tot}(x)\}
+\{n^{\prime b}(y)c^{\rm tot}_b(y),n(x) c^{\rm tot}(x)\}\Big)
\ea
 For the single Poisson brackets that occur in the equation above we 
obtain the following result
\ba
 \int d^3x\int d^3y\{\mu^j(y)Z_j(y),n(x) c^{\rm tot}(x)\}&=&-\int 
d^3x\sqrt{\det(q)} n\, \rho\, \mu^jq^{ab}U_bW_{j,a}
\\
 \int d^3x\int d^3y\{\mu(y)Z(y),n(x) c^{\rm tot}(x)\}&=&\int d^3x 
\mu\frac{n}{n^{\prime}}\tilde{c}\nonumber\\
 \int d^3x\int d^3y\{\mu_j(y)Z^j(y),n(x) c^{\rm tot}(x)\}&=&\int d^3x 
\mu_{j}\rho\, n\sqrt{\det(q)}q^{bc}U_bS^j_{,c}
\nonumber\\
\int d^3x\int d^3y\{\nu(y)z(y),n(x) c^{\rm tot}(x)\}&=&c^{\rm
tot}(\nu)\nonumber\\
\int d^3x\int d^3y\{\nu^b(y)z_b(y),n(x) c^{\rm tot}(x)\}&=&0
\nonumber\\
 \int d^3x\int d^3y\{n^{\prime}(y)c^{\rm tot}(y),n(x) 
c^{\rm tot}(x)\}&=&\vec{c}^{\rm tot}(q^{-1}[n^{\prime}dn-ndn^{\prime}])
\nonumber\\
\int d^3x\int d^3y\{n^{\prime b}(y)c^{\rm tot}_b(y), n(x) c^{\rm tot}(x)\}&=&
-c^{\rm tot}({\cal L}_{\vec{n}^{\prime}}n)
-\tilde{c}({\cal L}_{\vec n^{\prime}}\rho)\nonumber\\
&&
 -\int d^3x \rho\, n\, n^{\prime a} 
\sqrt{\det(q)}q^{bc}U_b\big(W_{j,a}S^j_c-W_{j,c}S^j_{,a}\big)
\nonumber
\ea
Reinserting these results into equation (\ref{Hprimc}) yields
\ba
\{H_{\rm primary},c^{\rm tot}(n)\}
 &=&-\int d^3x\sqrt{\det(q)} n \rho \mu^jq^{ab}U_bW_{j,a}+\int d^3x 
\mu\frac{n}{n^{\prime}}\tilde{c}
+\int d^3x \rho\, n\, \mu_j\sqrt{\det(q)}q^{bc}U_bS^j_{,c}
\nonumber\\
&&
+c^{\rm tot}(\nu)
+\vec{c}^{\rm tot}[q^{-1}(n^{\prime}dn-ndn^{\prime}])
-c^{\rm tot}({\cal L}_{\vec{n}^{\prime}}n)
-\tilde{c}({\cal L}_{\vec n^{\prime}}\rho)
\nonumber\\
&&
-\int d^3x \rho\, n\, n^{\prime a} 
\sqrt{\det(q)}q^{bc}U_b\big(W_{j,a}S^j_c-W_{j,c}S^j_{,a}\big)
\nonumber\\
&\approx&
 \int d^3x \sqrt{\det(q)}n \rho q^{bc}U_b\Big(
 W_{j,c}\big(n^{\prime a}S^j_{,a}-\mu^j-\frac{\rho\, 
n^{\prime}}{P}\sqrt{\det(q)}q^{de}U_dS^j_{,e}\big)\nonumber\\
&&
\quad\quad\quad\quad\quad\quad\quad\quad\quad\quad\quad
 S^j_{,c}\big(\mu_j-W_{j,a}n^{\prime a}+\frac{\rho\, 
n^{\prime}}{P}q^{de}U_d W_{j,e}\big)\Big)
\nonumber\\
&&
+\int d^3x\sqrt{\det(q)}n \rho q^{bc}U_b\frac{\rho
n^{\prime}}{P}\sqrt{\det(q)}q^{de}U_d
\Big(W_{j,c}S^j_{,e}-W_{j,e}S^j_{,c}\Big)\nonumber\\
&=&
 \int d^3x \sqrt{\det(q)}n \rho q^{bc}U_b\Big(
 W_{j,c}\big(n^{\prime a}S^j_{,a}-\mu^j-\frac{\rho\, 
n^{\prime}}{P}\sqrt{\det(q)}q^{de}U_dS^j_{,e}\big)\nonumber\\
&&
\quad\quad\quad\quad\quad\quad\quad\quad\quad\quad\quad
 S^j_{,c}\big(\mu_j-W_{j,a}n^{\prime a}+\frac{\rho\, 
n^{\prime}}{P}q^{de}U_d W_{j,e}\big)\Big)
\ea
 Here we used in the last step that the last integral in the line  
before the 
 last line one vanishes,  because $W_{j,c}S^j_{,e}-W_{j,e}S^j_{,c}$ is 
 antisymmetric in $e,c$ and multiplied by $q^{bc}q^{de}U_bU_d$ which is 
 symmetric in the indeces $c,e$.
\\
These are again the equation involving the Lagrange multipliers that we have
seen before in the calculations for $c^{\rm tot}_a$.
\\
\\
Finally, let us consider the Poisson bracket of $H_{\rm primary}$ and the 
secondary constraint $\tilde{c}$ whose smeared version is given by
\be
\tilde{c}(u):=\int d^3x u(x) \tilde{c}(x)
\ee
where $u$ is an appropiate smearing function. We obtain
\ba
\label{Hprimctilde}
\{H_{\rm primary},\tilde{c}(u)\}
&=&
\int d^3x\int d^3y\Big(
\{\mu^j(y)Z_j(y),u(x) \tilde{c}(x)\}
+\{\mu(y)Z(y),u(x) \tilde{c}(x)\}\nonumber\\
&&
+\{\mu_j(y)Z^j(y),u(x) \tilde{c}(x)\}
+\{\nu(y)z(y),u(x) \tilde{c}(x)\}
+\{\nu^b(y)z_b(y),u(x) \tilde{c}(x)\}\nonumber\\
&&
+\{n^{\prime}(y)c^{\rm tot}(y),u(x) \tilde{c}(x)\}
+\{n^{\prime b}(y)c^{\rm tot}_b(y),u(x) \tilde{c}(x)\}\Big)
\ea
In this case we do not need to compute all the individual Poisson bracket in
order to convince ourselves that no constraints arise, because the Poisson
bracket of $Z(\mu)$ and $\tilde{c}(u)$ yields
\ba
 \int d^3x\int d^3y\{\mu(y)Z(y),u(x) \tilde{c}(x)\}&=&\int d^3x \mu\, 
u\frac{nP^2}{\rho^3\sqrt{\det(q)}}
\ea
which is a new term invloving the Lagrange multiplier $\mu$. 
Thus, we can solve the equation $\{H_{\rm primary},\tilde{c}(u)\}=0$ for $\mu$. 
\\
It follows that no new terms are produced not involving $\mu^j,\mu_j,\mu$ in
this second iteration step. Consequently,
the full set of (primary and secondary) 
constraints is given by 
$c^{\rm tot}, c_a^{\rm tot}, \tilde{c}, Z_j,Z^j, Z, z_a$ and $z$ and it remains
to
classify them into first and second class.
Obviously,
\ba
\label{PBZZ}
\{Z^j(x),Z_k(y)\} &=& P\;\delta^j_k \; \delta(x,y) \; ,
\nonumber\\
\{Z(x),\tilde{c}(y)\} &=& \frac{\frac{n P^2}{\rho^3}}{\sqrt{\det(q)}}
\;\delta(x,y) \; ,
\ea
does not vanish on the constraint surface defined by the final set of
constraints, hence they are second class constraints. 
Since $n$ appears only linearly in $\tilde{c}$ and $n^a$ does not appear at 
all, it follows that $z,z_a$ are first class.
\\
Let us consider the linear combination 
\ba
\tilde{c}^{\rm tot}_a &\equiv& I\; \rho_{,a}+I^j \; W_{j,a}+P\;T_{,a} +P_j
\;S^j_{,a}
+p\;n_{,a}+{\cal L}_{\vec{n}} \; p_a+c_a
\nonumber\\
&=& c^{\rm tot}_a + Z \; \rho_{,a}+Z^j \; W_{j,a}+Z_j\; S^j_{,a}+z\; n_{,a}
+{\cal L}_{\vec{n}} \; z_a
\ea
where
\be
c_a \equiv c^{\rm geo}_a+c^{\rm matter}_a
\ee
is the non-dust contribution to the spatial diffeomorphism constraint
$c^{\rm tot}_a$. Since all constraints are scalar or covector densities of
weight one and $\tilde{c}^{\rm tot}_a$ is the generator of spatial
diffeomorphisms, it follows that $\tilde{c}^{\rm tot}_a$ is first class.
Finally,  we consider as an Ansatz the linear combination
\be 
\tilde{c}^{\rm tot} \equiv c^{\rm tot}+\alpha^j \;Z_j+\alpha_j\; Z^j+\alpha \;
Z\; ,
\ee
and determine the phase space functions $\alpha^j,\alpha_j,\;\alpha$
such that $\tilde{c}^{\rm tot}$ has vanishing Poisson brackets with
$Z_j,Z^j,Z$ up to terms proportional to $Z_j, Z^j, Z$.
\\
We have
\ba
\{\tilde{c}^{\rm tot}(x),Z_j(y)\}&=&\{c^{\rm
tot}(x),Z_j(y)\}+\alpha_k(x)\{Z^k(x),Z_k(y)\}\nonumber\\
&=&\{c^{\rm tot}(x),Z_j(y)\}+\alpha_j(x)P(x)\delta^3(x,y).
\ea
where we used equation (\ref{PBZZ}) in the last line. Solving this equation for
$\alpha_j$ we end up with
\be
\alpha_j(x)=-\int d^3y \frac{1}{P(y)}\{c^{\rm tot}(x),Z_j(y)\}
=\left(\frac{1}{P}\sqrt{\det(q)} \, \rho\, q^{ab}U_bW_{j,a}\right)(x)
\ee
which is a sensible expression since $\{c^{\rm tot}(x),Z_k(y)\}\sim
\delta^3(x,y)$.
For the Poisson bracket involving $Z^j$ we get
\ba
\{\tilde{c}^{\rm tot}(x),Z^j(y)\}&=&\{c^{\rm
tot}(x),Z^j(y)\}+\alpha^k(x)\{Z_k(x),Z^j(y)\}\nonumber\\
&=&\{c^{\rm tot}(x),Z_j(y)\}-\alpha^j(x)P(x)\delta^3(x,y)
\ea
such that this equation can be solved for $\alpha^j$ explicitly given by
\be
\alpha^j(x)=\int d^3y \frac{1}{P(y)}\{c^{\rm tot}(x),Z^j(y)\}
=\left(\frac{1}{P}\rho\sqrt{\det(q)}q^{bc}U_bS^j_{,c}\right)(x).
\ee
Finally, for $Z$ we obtain
\ba
\label{ctotctilde}
\{\tilde{c}^{\rm tot}(x),Z(y)\}&=&\{c^{\rm tot}(x),Z(y)\}\sim \tilde{c}\approx 0
\ea
Hence this Poisson bracket vanishes weakly. Considering now the Poisson bracket
between $\tilde{c}$ and $\tilde{c}^{\rm tot}$ we get
\ba
\{\tilde{c}^{\rm tot}(x),\tilde{c}(y)\}&=&\{c^{\rm tot}(x),\tilde{c}(y)\}
+\alpha^j(x)\{Z_j(x),\tilde{c}(y)\}+\alpha_j(x)\{Z^j(x),\tilde{c}(y)\}
+\alpha(x)\{Z(x),\tilde{c}(y)\}\nonumber\\
\ea
The results of the last three individual Poisson brackets occurring above are
listed below
\ba
\{Z_j(x),\tilde{c}(y)\}&=&\big(n\sqrt{\det(q)}q^{ab}U_aW_jP\big)(y)\frac{
\partial}{\partial y^b}\delta^3(x,y)\\
\{Z^j(x),\tilde{c}(y)\}&=&-\big(n\sqrt{\det(q)}q^{ab}U_aPS^j_{,b}
\big)(y)\delta^3(x,y)\nonumber\\
\{Z(x),\tilde{c}(y)\}&=&\left(\frac{nP^2}{\rho^3\sqrt{\det(q)}}
\right)(y)\delta^3(x,y)\nonumber
\ea
Now, we can solve equation (\ref{ctotctilde}) for $\alpha$ which yields
\ba
\alpha(x)&=&
-\int d^3y\frac{\rho^3\sqrt{\det(q)}}{nP^2}(y)\{c^{\rm tot}(x),\tilde{c}(y)\}\\
&&-\left(\frac{\rho^3\sqrt{\det(q)}}{nP^2}\big[n\sqrt{\det(q)}q^{ab}
U_aW_jP\big)\big]_{,b}
\frac{1}{P}\rho\sqrt{\det(q)}q^{bc}U_bS^j_{,c}\right)(x)
\nonumber\\
&&-\left(\frac{\rho^3\sqrt{\det(q)}}{nP^2}\big(n\sqrt{\det(q)}q^{ab}U_aPS^j_{,b}
\big)
\frac{1}{P}\sqrt{\det(q)} \, \rho\, q^{ab}U_bW_{j,a}\right)(x)
\ea
Here we reinserted the expressions for $\alpha_j$ and $\alpha^j$ derived before.
The final step which includes the construction of the Dirac bracket can again be
found in the main text.

\section{Comparison with Symplectic reduction}
\label{s3.2}

The spatial diffeomorphism invariant quantities  
\be
\label{CoordPS}
\left(\tilde{\xi}(\sigma),\tilde{\pi}(\sigma)\right)\;, \; 
\left(\tilde{T}(\sigma),
\tilde{P}(\sigma)\right)\;, \; \left(\tilde{q}_{ij}(\sigma),
\tilde{p}^{ij}(\sigma)\right)\; 
\ee
shown in equation (\ref{3.27})
are also obtained in \cite{3} through symplectic reduction which is an
alternative method to show that the pairs 
in (\ref{3.27}) are conjugate. 
\\
To see how this works, we compute
\ba \label{3.28}
\frac{d}{dt}\widetilde{T}(\sigma)
&=& 
\frac{d}{dt} \int_{{\cal X}} \; d^3x\;  
\det(\partial S/\partial x)\; \delta(S(x),\sigma) T(x)
\nonumber\\
&=& 
\int_{{\cal X}} \; d^3x\; \det(\partial S/\partial x)\; \left( 
\delta(S(x),\sigma)\; \left[\frac{d}{dt} T(x)\right]
+S^a_j(x)\; \left[\frac{d}{dt} S^j_{,a}(x)\right]\;  
\delta(S(x),\sigma)\; T(x)\right.
\nonumber\\
&&\left. 
\hspace{4cm}
+\left[\frac{d}{dt} S^j(x)\right]\; 
\left[\frac{\partial\delta(\sigma',\sigma)}{\partial\sigma^{j\prime}}\right]_{
\sigma'=S(x)}
T(x)\right)
\nonumber\\
&=& 
\left[\frac{d}{dt} T(x)\right]_{S(x)=\sigma} +
\int_{{\cal X}} \; d^3x\; \det(\partial S/\partial x)\; 
\left[\frac{d}{dt} S^j(x)\right]\; 
\left(-S^a_j(x) \partial_a \left[\delta(S(x),\sigma)\; T(x)\right]\right.
\nonumber\\
&& \left.
\hspace{4cm}
+\left[\frac\partial{\delta(\sigma',\sigma)}{\partial\sigma^{j\prime}}\right]_{
\sigma'=S(x)}\; 
T(x)\right)
\nonumber\\
&=&
\left[\frac{d}{dt} T(x)\right]_{S(x)=\sigma} -
\int_{{\cal X}} \; d^3x\; \det(\partial S/\partial x)\; 
\delta(S(x),\sigma)\; \left[\frac{d}{dt} S^j(x)\right]\; S^a_j(x) T_{,a}(x)
\nonumber\\
&=&
\left[\frac{d}{dt} T(x)-(\frac{d}{dt} S^j(x))\; S^a_j(x)
T_{,a}(x)\right]_{S(x)=\sigma} 
\ea
where we have used $\partial_a [S^a_j \det(\partial S/\partial x)]=0$.
Exactly the same calculation reveals
\ba \label{3.29}
\frac{d}{dt}\tilde{\xi}(\sigma)
& = &\left[\frac{d}{dt} \xi(x)-\left(\frac{d}{dt} S^j(x)\right)\; S^a_j(x) 
\xi_{,a}(x)\right]_{S(x)=\sigma} 
\nonumber\\
\frac{d}{dt}\tilde{q}_{jk}(\sigma)
& = &\left[\frac{d}{dt} q_{jk}(x)-\left(\frac{d}{dt} S^l(x)\right)\; S^a_l(x) 
q_{jk,a}(x)\right]_{S(x)=\sigma} 
\ea
Using (\ref{3.28}) and (\ref{3.29}) we can now rewrite the symplectic 
potential in terms of the spatially 
diffeomorphism invariant variables as follows 
$\dot{(.)}:=\frac{d}{dt}(.)$ and $J=\det(\partial S/\partial x)$)
\ba \label{3.30}
\Theta &=& \int_{{\cal X}}\;  d^3x 
\; \left[\dot{\xi}\;\pi+\dot{T}\;P+\dot{S}^j\;P_j+
\dot{q}_{ab}\;p^{ab}\right]
\nonumber\\
&=& \int_{{\cal X}}\; d^3x \; 
\left[\dot{\xi}\;\pi+\dot{T}\;P+\dot{S}^j\;P_j+
\left(\frac{d}{dt}\big(q_{jk}\; S^j_{,a} \;S^k_{,b}\big)\right)\;p^{ab}\right]
\nonumber\\
&=& \int_{{\cal X}} \;d^3x \; 
\left[\dot{\xi}\;\pi+\dot{T}\;P+\dot{S}^j\;P_j+
\dot{q}_{jk} (S^j_{,a}\; S^k_{,b}\;p^{ab})
+2q_{jk} \dot{S}^j_{,a} S^k_{,b}\;\;p^{ab}\right]
\nonumber\\
&=& \int_{{\cal X}} \;d^3x \; 
\left[\dot{\xi}\;\pi+\dot{T}\;P+\dot{S}^j\;P_j+
\dot{q}_{jk} (S^j_{,a}\; S^k_{,b}\;p^{ab})
-2 \dot{S}^j \partial_a(q_{jk}  S^k_{,b}\;\;p^{ab})\right]
\nonumber\\
&=& \int_{{\cal X}} \; J\;d^3x \; \; 
\left[\dot{\xi}\;\frac{\pi}{J}+\dot{T}\;\frac{P}{J}
+\dot{q}_{jk} \frac{S^j_{,a}\; S^k_{,b}\;p^{ab}}{J}\right]
+\int_{{\cal X}} \;d^3x \;
\dot{S}^j\;\left[P_j  
-2 \partial_a(q_{bc}  S^c_j\;\;p^{ab})\right]
\nonumber\\
&=& \int_{{\cal S}} \;d^3\sigma \; \; 
\tilde{\pi}\left[\big[\dot{\xi}\big]_{S(x)=\sigma}+\tilde{P}\big[\dot{T}\big]_{
S(x)=\sigma}
+\tilde{p}^{jk}\big[\dot{q}_{jk}\big]_{S(x)=\sigma} \right]
+\int_{{\cal X}} \;d^3x \;
\dot{S}^j\;\left[P_j  
-2 \partial_a(q_{bc}  S^c_j\;\;p^{ab})\right]
\nonumber\\
&=& \int_{{\cal S}} \;d^3\sigma \; \; 
\left[\dot{\widetilde{\xi}}\;\tilde{\pi}+\dot{\widetilde{T}}\;\widetilde{P}
+\dot{\tilde{q}}_{jk}\tilde{p}^{jk}\right]
\nonumber\\
&& +\int_{{\cal X}} \;d^3x \;
\dot{S}^j\;\left[P_j+ S^a_j\left(\pi\; \xi_{,a}+P T_{,a}+ p^{bc} S^k_{,b} 
S^l_{,c}  S^a_j q_{kl,a}\right)   
-2 \partial_a(q_{bc}  S^c_j\;\;p^{ab})\right]
\nonumber\\
&=& \int_{{\cal S}} \;d^3\sigma \; \; 
\left[\dot{\tilde{\xi}}\;\tilde{\pi}+\dot{\widetilde{T}}\;\widetilde{P}
+\dot{\tilde{q}}_{jk} \tilde{p}^{jk}\right]
\nonumber\\
&& +\int_{{\cal X}} \;d^3x \;
\dot{S}^j\;\left[P_j+ S^a_j\left(\pi\; \xi_{,a}+P T_{,a}+ p^{bc} 
S^k_{,b} S^l_{,c} \partial_a \big(S^e_k S^f_l q_{ef}\big)\right)   
-2 \partial_a\big(q_{bc}  S^c_j\;\;p^{ab}\big)\right]
\nonumber\\
&=& \int_{{\cal S}} \;d^3\sigma \; \; 
\left[\dot{\tilde{\xi}}\;\tilde{\pi}+\dot{\widetilde{T}}\;\widetilde{P}
+\dot{\tilde{q}}_{jk} \tilde{p}^{jk}\right]
\nonumber\\
&& +\int_{{\cal X}} \;d^3x \;
\dot{S}^j\;\left[P_j+ S^a_j\left(\pi\; \xi_{,a}+P T_{,a}+ p^{bc} 
\big(q_{bc,a}+ 2 q_{ec} S^k_{,b} S^e_{k,a}\big)\right)   
-2 \partial_a\big(q_{bc}  S^c_j\;\;p^{ab}\big)\right]
\nonumber\\
&=& \int_{{\cal S}} \;d^3\sigma \; \; 
\left[\dot{\tilde{\xi}}\;\tilde{\pi}+\dot{\widetilde{T}}\;\widetilde{P}
+\dot{\tilde{q}}_{jk} \tilde{p}^{jk}\right]
\nonumber\\
&& +\int_{{\cal X}} \;d^3x \;
\dot{S}^j\left[P_j+ S^a_j\left(\pi\; \xi_{,a}+P T_{,a}+ 
-2\left[q_{ab} \partial_c
p^{bc}+\frac{1}{2}\left(2q_{a(b,c)}-q_{bc,a}\right)p^{bc}\right]\right)\right]
\nonumber\\
&=& \int_{{\cal S}} \;d^3\sigma \; \; 
\left[\dot{\tilde{\xi}}\;\tilde{\pi}+\dot{\widetilde{T}}\;\widetilde{P}
+\dot{\tilde{q}}_{jk} \tilde{p}^{jk}\right]
+\int_{{\cal X}} \;d^3x \;
\dot{S}^j\left[P_j+ S^a_j\;(\pi\; \xi_{,a}+P T_{,a}+ 
-2 q_{ab} D_c p^{bc})\right]
\nonumber\\
&=& \int_{{\cal S}} \;d^3\sigma \; \; 
\left[\dot{\tilde{\xi}}\;\tilde{\pi}+\dot{\widetilde{T}}\;\widetilde{P}
+\dot{\tilde{q}}_{jk} \tilde{p}^{jk}\right]
+\int_{{\cal X}} \;d^3x \;
\dot{S}^j\; S_a^j c^{\rm tot}_a 
\ea
where we used 
\be \label{3.31}
S^a_j S^k_{,b} S^e_{k,a}
=-S^a_j S^k_{,ba} S^e_k
=-S^a_j S^k_{,ab} S^e_k
=S^a_{j,b} S^k_{,a} S^e_k
=S^e_{j,b}
\ee
as well as the definition of the Christoffel symbol in the second to 
last step (notice that $p^{ab}$ is a tensor density so that 
$D_b p^{ab}=\partial_b p^{ab}+\Gamma^a_{bc} p^{bc}$). 

Formula (\ref{3.30}) means that on the full phase space we can switch to 
the new canonical pairs (\ref{CoordPS}) on $\cal X$ and the canonical pair
$(S^j,P'_j=S^a_j c^{\rm tot}_a=c^{\rm tot}_j)$. The pairs (\ref{CoordPS}) are
thus
$c^{\rm tot}_j$ invariant while $S^j$ is pure gauge. The symplectic 
reduction of the full phase space with respect to $c^{\rm tot}_j$ is 
therefore precisely coordinatised by (\ref{CoordPS}) which is identical to
equation (\ref{3.27}) in the main text.

\section{Effective Action and Fixed Point Equation}
\label{s6}

The aim of the present section is to derive, at least in implicit form,
the Lagrangian that corresponds canonically to the physical Hamiltonian.
This can be done by calculating the inverse Legendre transform\footnote{
This is possible because the Legendre transform is regular.
} and leads to a fixed point equation, which can be solved order by order,
in principle.
Interestingly, the Lagrangian turns out to be local in dust time,
but will be non--local in dust space. However, the Hamiltonian
description is completely local.

The inverse Legendre transform requires to solve for the momenta
$P^{jk}(\sigma),\;\Pi(\sigma)$ in terms of the corresponding
velocities $V_{jk}(\sigma),\;\Upsilon(\sigma)$, respectively, defined by
\footnote{ Note that $\dot{Q}_{jk}$ and $\dot{\Xi}$ must be treated
as independent variables in addition to $Q_{jk}$ and $\Xi$. } \ba
\label{6.1} V_{jk}(\sigma) &\equiv&
\dot{Q}_{jk}(\sigma)=\left\{\HF,Q_{jk}(\sigma)\right\} 
\nonumber\\
\Upsilon(\sigma) &\equiv& \dot{\Xi}(\sigma)=\left\{\HF,\Xi(\sigma)\right\}.
\ea
This can be done by using the first order equations of motion
for $Q_{jk}(\sigma),\;\Xi(\sigma)$, derived from the physical Hamiltonian
$\HF$.
From the physical Hamiltonian $\HF$ we obtain the Lagrangian
\be \label{6.2}
\LF[Q,V;\Xi,\Upsilon]=
\int_{{\cal S}} \; {\rm d}^3\sigma\; L(\sigma)
=\int_{{\cal S}} \; {\rm d}^3\sigma\;\left[\left(\frac{1}{\kappa} P^{jk}
V_{jk}+\frac{1}{\lambda}\Pi
\Upsilon\right)-\HF[Q,P;\Xi,\Pi]\right]_{{\rm (\ref{6.1})}}
\ee
where it it is understood that the solution of
(\ref{6.1}) for $P^{jk},\;\Pi$ has to be inserted.

With the dynamical lapse and shift given by
$N=C/H,\;N_j=-C_j/H$, respectively, we obtain for $P^{jk}$
\be
P^{jk}=\frac{\sqrt{\det(Q)}}{2N}[G^{-1}]^{jkmn}\big(\dot{Q}_{mn}-({\cal
L}_{\vec{N}}Q)_{mn}\big)
=\sqrt{\det(Q)}[G^{-1}]^{jkmn}K_{mn} \; ,
\ee
with $K_{mn}$ denoting the extrinsic curvature.
This leads to the following expression for the velocities $V_{jk}$ and
$\Upsilon$:
\ba \label{6.3}
V_{jk} &=& 2\left[N\; K_{jk}+D_{(j} N_{k)}\right]\quad\mathrm{and}\quad \Upsilon
=\frac{N}{\sqrt{\det(Q)}}\Pi+Q^{jk} N_j\; D_k \Xi\; .
\ea
We conclude
\ba \label{6.4}
\LF &=& \int_{{\cal S}} \; {\rm d}^3\sigma\; \left[\frac{1}{\kappa}V_{jk} P^{jk} +
\frac{1}{\lambda}\Upsilon\Pi-H\right]
\\
&=& \int_{{\cal S}} \; {\rm d}^3\sigma\; \left[\frac{2}{\kappa}\left(N K_{jk}+D_{(j}
N_{k)}\right)
P^{jk} +
\left(N\frac{\Pi}{\lambda\sqrt{\det(Q)}}+N^j D_j\frac{\Xi}{\lambda}\right)\Pi-H\right]
\nonumber\\
&=& \int_{{\cal S}} \; {\rm d}^3\sigma\;
\left[N\Big(\frac{2}{\kappa}K_{jk} P^{jk}+\frac{\Pi^2}{\lambda\sqrt{\det(Q)}}\Big)
+N^j C_j-H\right]
\nonumber\\
&=& \int_{{\cal S}} \; {\rm d}^3\sigma\;\frac{1}{H}
\left[C\left(\frac{2}{\kappa}K_{jk} P^{jk}+\frac{\Pi^2}{\lambda\sqrt{\det(Q)}}\right)
-Q^{jk} C_j
C_k-H^2\right]
\nonumber\\
&=& \int_{{\cal S}} \; {\rm d}^3\sigma\;\frac{C}{H}
\left[\frac{2}{\kappa}K_{jk} P^{jk}+\frac{\Pi^2}{\lambda\sqrt{\det(Q)}}-C\right]
\nonumber\\
&=& \int_{{\cal S}} \; {\rm d}^3\sigma\;N
\left[\frac{2}{\kappa}\sqrt{\det(Q)}(K_{jk}
K^{jk}-(K_j^j)^2)+\frac{\Pi^2}{\lambda \sqrt{\det(Q)}}-C\right]
\nonumber\\
&=& \int_{{\cal S}} \; {\rm d}^3\sigma\;N \sqrt{\det(Q)}
\left[\frac{1}{\kappa}\left(K_{jk}
K^{jk}-(K_j^j)^2+R^{(3)}[Q]-2\Lambda\right)
+\frac{1}{2\lambda}\left(\frac{\Pi^2}{\det(Q)}-[Q^{jk}
\Xi_{,j} \Xi_{,k}+v(\Xi)]\right)\right]
\nonumber\\
&=& \int_{{\cal S}} \; {\rm d}^3\sigma\;N\; \sqrt{\det(Q)}\;
\left[\frac{1}{\kappa}\left(K_{jk}
K^{jk}-(K_j^j)^2+R^{(3)}[Q]-2\Lambda\right)
+\frac{1}{2\lambda}\left((\nabla_u\Xi)^2-(Q^{jk}
\Xi_{,j} \Xi_{,k}+v(\Xi))\right)\right] \nonumber.
\ea
In the third step we performed an integration by
parts and in the fourth step we substituted the expressions for
dynamical lapse and shift, in the sixth
we rewrote $P^{jk}$ in terms of $K_{jk}$, in the seventh we substituted
for $C$ and in the last we introduced the vector field
$u=\frac{1}{N}(\partial_\tau-N^j\partial_{\sigma^j})$.

If lapse and shift would be independent variables, the final
expression in (\ref{6.5}) would coincide with the 3+1--
decomposition of the Einstein -- Hilbert term minimally coupled to
a Klein -- Gordon field with potential $v$!
Since $N_j$ is a constant of
the physical motion and $N=\sqrt{1+Q^{jk} N_j N_k}$, we could, in particular,
consider the case $N_j=0$, whence $N=1$. In that case (\ref{6.4})
would agree with the usual Lagrangian description on dust space--time for a
static
foliation. However, fundamentally lapse and shift are not independent
variables, and we must use this fact in (\ref{6.1}) in order to
solve for $P^{jk},\Pi$. We now turn to this task.

By definition
\ba \label{6.5}
N_j &=& -\frac{C_j}{H}=-\frac{C_j}{C}\;\frac{C}{H}=
-\frac{C_j/\sqrt{\det(Q)}}{C/\sqrt{\det(Q)}}
\sqrt{1+Q^{jk} N_j N_k} \; ,
\nonumber\\
\frac{C_j}{\sqrt{\det(Q)}} &=& -\frac{2}{\kappa}\left(D_k K^k_j-D_j
K^k_k\right)+\frac{1}{\lambda} \left(\nabla_u \Xi\right) D_j \Xi \; ,
\nonumber\\
\frac{C}{\sqrt{\det(Q)}} &=& \frac{1}{\kappa}\left(K_{jk}
K^{jk}-\left[K_j^j\right]^2-R^{(3)}\left[Q\right]\right)+\frac{1}{2\lambda}
\left(\left(\nabla_u \Xi\right)^2+Q^{jk} \Xi_{,j}
\Xi_{,k}+v(\Xi)\right) \; ,
\nonumber\\
K_{jk} &=& \frac{1}{2 \sqrt{1+Q^{jk} N_j N_k}}\left[V_{jk}-2D_{(j} N_{k)}\right]
\; ,
\nonumber\\
\nabla_u \Xi &=&\frac{1}{\sqrt{1+Q^{jk} N_j N_k}}\left[\Upsilon-Q^{jk} N_j
D_k \Xi\right]\; .
\ea
The set of equations (\ref{6.5}), when inserted into each other,
yields an equation of the form
\be \label{6.6}
N_j=G_j[N_k;Q_{kl},V_{kl},\Xi,\Upsilon] \; ,
\ee
where $G_j$ is a local function of its arguments and their
spatial derivatives up to second order (in particular, second spatial
derivatives of $N_j$). Since $P^{jk},\Pi$ are known in terms of
$Q_{jk},V_{jk},\Xi,\Upsilon$, once $N_j$ (and thus $N$) is known as a
function of these arguments, we have reduced the task of performing the
inverse Legendre transform to solving the {\sl fixed point equation}
(\ref{6.6}).

Unfortunately, (\ref{6.6}) is not algebraic in $N_j$, so
that a solution just by quadratures is impossible. Also, it represents a
highly nonlinear system of partial differential equations of
degree two, so that linear solution methods fail, as well. We leave
the full investigation of this system for future research. However, the
fact
that it is a system of fixed point equations suggests to look for a
solution by perturbative or fixed point methods: \\
1. If we make the Ansatz that $N_j$ is small, in an appropriate sense, then
we may expand $G_j[N_k]$ around
$N_k=0$ to linear order and solve the
resulting linear system of PDE's.\\
2. The fixed point method suggests to write the solution in the form
\be \label{6.7}
N_j=G_j(G_{k_1}(G_{k_2}(..(G_{k_n}(..))..))) \; .
\ee
If convergence is under control, then an $n-th$ order approximation may be
given in the form
\be \label{6.8}
N^{(n)}_j=G_j(G_{k_1}(G_{k_2}(..(G_{k_n}(0))..))) \; ,
\ee
which consists in setting the starting point of the iteration at $N_j=0$
(which is a reasonable guess if the exact solution is indeed small in an
appropriate sense, having a {\sl test clock} in mind)
and to iterate $n$ times. The expression (\ref{6.8}) contains spatial
derivatives of the metric of order up to $2(n+1)$, but is only
of first order in $\tau$--derivatives, thus establishing that
the final Lagrangian is spatially non--local in dust space but temporally
local in dust time.

\section{Two Routes to Second Time Derivatives of Linear Perturbations}
\label{sb}

In this appendix we consider a general Hamiltonian system with canonical
coordinates $(q,p)$ and standard Poisson brackets $\{p,q\}=1$ and a
Hamiltonian
function $H(q,p)$. We will consider only one degree of freedom but
everything generalises to an arbitrary number of degrees of freedom.
\begin{Lemma} \label{sb.1}
Let $(q_0(\tau),p_0(\tau))$ be an exact solution of the Hamiltonian
equations of motion
\ba \label{b.1}
\dot{q}_0(\tau) &=& \Big[\{H,q\}(q,p)\Big]_{q=q_0(\tau)\atop p=p_0(\tau)}
\nonumber\\
\dot{p}_0(\tau) &=& \Big[\{H,p\}(q,p)\Big]_{q=q_0(\tau)\atop p=p_0(\tau)}
\ea
Define the perturbations $\delta q:=q-q_0(\tau),\;\delta
p:=p-p_0(\tau)$.
Let $H(q,p)=\sum_{n=0}^\infty H^{(n)}$ be the expansion of $H(q,p)$
around $q_0(\tau),p_0(\tau)$ in terms of the perturbations where
$H^{(n)}$ is the n-th order term in terms of the perturbations. Then\\
(i.) Expanding the full Hamiltonian equations of motion for
$\dot{q},\;\dot{p}$ to linear order
is equivalent to using the function $H^{(2)}$ as a Hamiltonian for the
perturbations.\\
(ii.) Expanding the equation for $\ddot{q}$ to linear order is equivalent
to the equations for $\dot{\delta q}=\delta\dot{q},\dot{\delta
p}=\delta\dot{p}$ to linear order.
\end{Lemma}
\begin{proof}
Notice that $q_0(\tau),p_0(\tau)$ do not carry any phase space
dependence in contrast to $\delta q,\delta p$. Therefore
$\{\delta q,\delta q\}=\{\delta p,\delta p\}=0$ and
$\{\delta p,\delta q\}=1$. \\
(i.)\\
Consider the full Hamiltonian equations of motion for a general solution
$(q(\tau),p(\tau))$
\be \label{sb.2}
\dot{q}(\tau)=H_{,p}(q(\tau),p(\tau))\quad\mathrm{and}\quad
\dot{p}(\tau)=-H_{,p}(q(\tau),p(\tau))
\ee
where $H_{,q}=\partial H/\partial q,\;H_{,p}=\partial H/\partial p$.
We set $\delta q(\tau)=q(\tau)-q_0(\tau)$ and $\delta
p(\tau)=p(\tau)-p_0(\tau)$. Subtracting from (\ref{sb.2}) the equations
for $(q_0(\tau),p_0(\tau))$ we obtain
\ba \label{sb.3}
\dot{\delta q}(\tau)
&=& H_{,p}(q(\tau),p(\tau))-H_{,p}(q_0(\tau),p_0(\tau))
\nonumber\\
\delta
\dot{p}(\tau) &=& -H_{,q}(q(\tau),p(\tau))+H_{,q}(q_0(\tau),p_0(\tau))
\ea
which is still exact.
Expanding the right hand side of (\ref{sb.3}) to first order in $\delta
q(\tau),\;\delta p(\tau)$ we find
\ba \label{sb.4}
\delta \dot{q}(\tau)
&=&
H_{,pq}(q_0(\tau),p_0(\tau)) \delta q(\tau)
+H_{,pp}(q_0(\tau),p_0(\tau)) \delta p(\tau)
\nonumber\\
\delta
\dot{p}(\tau) &=& -H_{,qq}(q_0(\tau),p_0(\tau)) \delta q(\tau)-
H_{,qp}(q_0(\tau),p_0(\tau)) \delta p(\tau)
\ea
On the other hand we have
\be \label{sb.5}
H^{(2)}=
\frac{1}{2} H_{,qq}(q_0(\tau),p_0(\tau)) \big[\delta q\big]^2
+\frac{1}{2} H_{,pp}(q_0(\tau),p_0(\tau)) \big[\delta p\big]^2
+H_{,qp}(q_0(\tau),p_0(\tau)) \big[\delta q\big]\;\big[\delta p\big]
\ee
Then it is trivial to check that (\ref{sb.4}) is reproduced by
\be \label{sb.6}
\delta \dot{q}(\tau)=\{H^{(2)},\delta q\}_{
\delta q=\delta q(\tau)\atop\delta p=\delta p(\tau)}\quad\mathrm{and}\quad
\delta \dot{p}(\tau)=\{H^{(2)},\delta p\}_{
\delta q=\delta q(\tau)\atop\delta p=\delta p(\tau)}
\ee
(ii.)\\
Let $p(\tau)=F(q(\tau),\dot{q}(\tau))$ be the solution of solving
$\dot{q}(\tau)=H_{,p}(q(\tau),p(\tau))$ for $p(\tau)$. Then
\ba \label{sb.7}
\ddot{q}(\tau) &=&
H_{,pq}\big(q(\tau),p(\tau)\big) \dot{q}(\tau)
+H_{,pp}\big(q(\tau),p(\tau)\big) \dot{p}(\tau)
\nonumber\\
&=& H_{,pq}\big(q(\tau),p(\tau)\big) \dot{q}(\tau)
-H_{,pp}\big(q(\tau),p(\tau)\big) H_{,q}\big(q(\tau),p(\tau)\big)
\nonumber\\
&=& H_{,pq}\big(q(\tau),F(q(\tau),\dot{q}(\tau))\big) \dot{q}(\tau)
-H_{,pp}\big(q(\tau),F(q(\tau),\dot{q}(\tau))\big)
H_{,q}\big(q(\tau),F(q(\tau),\dot{q}(\tau))\big)
\nonumber\\
&=:& G\big(q(\tau),\dot{q}(\tau)\big)
\ea
Equation (\ref{sb.7}) is what we mean by the $\ddot{q}(\tau)$ form of
the
quations of motion, i.e an equation only involving $q,\dot{q},\ddot{q}$
but no longer the momentum.
Subtracting from (\ref{sb.7}) the corresponding
equation for $\ddot{q}_0(\tau)$ and expanding the right hand side to
first order we obtain with $G=G(q,v)$
\be \label{sb.8}
\ddot{\delta q}(\tau) =
G_{,q}(q_0(\tau),\dot{q}_0(\tau))\delta q(\tau)
+G_{,v}(q_0(\tau),\dot{q}_0(\tau))\delta \dot{q}(\tau)
\ee
Now
\ba \label{sb.9}
G_{,q}(q,v) &=& \big[H_{,pqq}(q,p) v-H_{,ppq}(q,p) H_{,q}(q,p)
-\big[H_{,pp}(q,p) H_{,qq}(q,p)\big]_{p=F(q,v)}
\nonumber\\
&& +\big[H_{,ppq}(q,p) v-H_{,ppp}(q,p) H_{,q}(q,p)
-H_{,pp}(q,p) H_{,pq}(q,p)\big]_{p=F(q,v)} \; F_{,q}(q,v)
\\
G_{,v}(q,v) &=& H_{,pq}(q,p)
+\big[H_{,ppq}(q,p) v
-H_{,ppp}(q,p) H_{,q}(q,p)
-H_{,pp}(q,p) H_{,pq}(q,p)\big]_{p=F(q,v)} \; F_{,v}(q,v)
\nonumber
\ea
Since
$v=H_{,p}(q,F(q,v))$ is an identity we obtain
\be \label{sb.10}
1=H_{,pp}(q,F(q,v)) F_{,v}(q,v)\quad\mathrm{and}\quad
0=H_{,pq}(q,F(q,v))+H_{,pp}(q,F(q,v)) F_{,q}(q,v)
\ee
by taking the derivative with respect to the independent variables $v,q$
respectively. This way we can eliminate the derivatives of $F$
\be \label{sb.11}
F_{,v}(q,v)=\frac{1}{H_{,pp}(q,F(q,v))}\quad\mathrm{and}\quad
F_{,q}(q,v)=-\frac{H_{,pq}(q,F(q,v))}{H_{,pp}(q,F(q,v))}
\ee
Substituting (\ref{sb.11}) into (\ref{sb.9}) we obtain the simplified
expressions
\ba \label{sb.12}
G_{,v}(q,v) &=&
\left[\left(\frac{H_{,ppq} v - H_{,ppp}
H_{,q}}{H_{,pp}}\right)(q,p)\right]_{p=F(q,v)}
\nonumber\\
G_{,q}(q,v) &=& \big[H_{,pqq}(q,p) v-H_{,ppq}(q,p) H_{,q}(q,p)
-H_{,pp}(q,p) H_{,qq}(q,p)\big]_{p=F(q,v)}
\nonumber\\
&& -\big[H_{,ppq}(q,p) v-H_{,ppp}(q,p) H_{,q}(q,p)
-H_{,pp}(q,p) H_{,pq}(q,p)\big]_{p=F(q,v)} \; \frac{H_{,pq}}{H_{,pp}}
\ea
Now we invert the second equation in (\ref{sb.4}) for $\delta{p}$ and
obtain
\be \label{sb.13}
\delta
p(\tau)=\frac{\delta\dot{q}(\tau)-H_{,pq}(q_0(\tau),p_0(\tau))\delta
q(\tau)}{H_{,pp}(q_0(\tau),p_0(\tau)}
\ee
Taking the time derivative of the first equation in (\ref{sb.4}) and
using (\ref{sb.13}) yields after some algebra 
\ba \label{sb.14}
\delta \ddot{q}(\tau) &=&
\left[\dot{H}_{,pq}-H_{,pp} H_{,qq}-\frac{H_{,pq}(\dot{H}_{,pp} -H_{,pp}
H_{,pq})}{H_{,pp}}\right](q_0(\tau),p_0(\tau)) \; \delta q(\tau)
\nonumber\\
&& +\left[\frac{\dot{H}_{,pp}}{H_{,pp}}\right](q_0(\tau),p_0(\tau)) \; \delta
\dot{q}(\tau) \ea where e.g. $\dot{H}_{,pp}(q_0(\tau),p_0(\tau):=
\frac{d}{d\tau} H_{,pp}(q_0(\tau),p_0(\tau)$. Carrying out the
remaining time derivatives in (\ref{sb.14}) and comparing with
(\ref{sb.8}) evaluated with the help of (\ref{sb.12}) at
$q=q_0(\tau),\;v=\dot{q}_0(\tau)=H_{,p}(q_0(\tau),p_0(\tau)$ we see
that the expressions coincide.
\end{proof}

\section{Constants of the Motion of n'th Order Perturbation Theory}
\label{sc}
In this section we will show that for any fully conserved quantity 
$F$ of a Hamiltonian sytem with Hamiltonian $H$, when expanding both the 
equations of motion and $F$ to order $n$ then $F$ is still a constant of 
motion up to terms of order $n+1$.
\\
\\
Let $m_0(\tau)=(q_0(\tau),p_0(\tau))$ be an exact solution of a
Hamiltonian system with canonical coordinates $m=(q,p)$, non
vanishing Poisson brackets $\{p,q\}=1$ and Hamiltonian
$H=H(m)=H(q,p)$. Define $\delta m=m-m_0(\tau)$. Since $m_0(\tau)$ is
just a number (for fixed $\tau$) we immediately have the non
vanishing Poisson brackets $\{\delta p,\delta q\}=1$. For any
function $F$ on phase space we consider its Taylor expansion around
$m_0(\tau)$ given by \be \label{sc.1} F(m)=\sum_{n=0}^\infty\;
F^{(n)}(m_0(\tau);\delta m) \ee where $F^{(n)}(m_0(\tau);\delta m)$
is a homogeneous polynomial of degree $n$ in $\delta m$ whose
coefficients depend explicitly on the background solution
$m_0(\tau)$, that is \be \label{sc.2} F^{(n)}(m_0(\tau);\delta
m)=\sum_{k=0}^n \frac{1}{k!\; (n-k)!} \; \left[\frac{\partial^b
F}{\left[\partial q\right]^k\;\left[\partial p\right]^{n-k}}\right](m_0(\tau))\;
\left[\delta q\right]^k
\left[\delta p\right]^{n-k} \ee
\begin{Lemma} \label{lasc.1}
The Poisson bracket $\{F,G\}$ can be computed either by first
expanding $F,G$ as in (\ref{sc.1}) and then using the Poisson bracket
for $\delta m$ or by using the Poisson bracket for $m$ and then
expanding the result as in (\ref{sc.1}).
\end{Lemma}
\begin{proof}
The proof is elementary: Since
\be \label{sc.3}
F=\sum_{k,l=0}^\infty \; \frac{1}{k!\; l!}\;
\left[\frac{\partial^{k+l} F}{\left[\partial q\right]^k\;\left[\partial
p\right]^l}\right](m_0(\tau))
\left[\delta q\right]^k\;\left[\delta p\right]^l
\ee
we have with the substitution of $F$ by $F_{,q}$
\be \label{sc.4}
F_{,q}=\sum_{k,l=0}^\infty \; \frac{1}{k!\; l!}\;
\left[\frac{\partial^{k+l+1} F}{\left[\partial q\right]^{k+1}\;\left[\partial
p\right]^l}\right](m_0(\tau))
\left[\delta q\right]^k\;\left[\delta p\right]^l
=F_{,\delta q}
\ee
and similarly $F_{,p}=F_{,\delta p}$. Since one computes Poisson
brackets the first way by first expanding and then taking derivatives
with respect to $\delta q,\delta p$ while the second way we compute
Poisson brackets with respect to $q,p$ and then expand, the assertion
follows.
\end{proof}
\begin{Lemma} \label{lasc.2}
Suppose we expand the Hamiltonian to $n-$th order in $\delta m$ with
$n\ge 1$. Suppose
also that $F$ is an exact constant of the motion with respect to the
Hamiltonian $H$. Then:\\
(i.) The equations of motion up to order $n$ for $\delta m$ are generated
by the Hamiltonian
\be \label{sc.5}
H_n=\sum_{k=2}^{n+1} \; H^{(k)}
\ee
(ii.) The perturbation up to order $n$ of $F$ given by
\be \label{sc.6}
F_n:=\sum_{k=1}^n F^{(n)}
\ee
is a constant of motion with respect to $H_n$ up to terms of order at
least $n+1$.
\end{Lemma}
Notice that the Hamiltonian starts at order two and ends at order $n+1$.
\begin{proof}
~\\
(i.)\\
Let $m(t)$ be any solution of the exact equation of motion. We have
for example
\be \label{sc.7}
\dot{q}(\tau)=[\{H,q\}]_{m=m(\tau)}
\ee
Subtracting the same equation for $m_0(\tau)$ and setting $\delta
q(\tau)=q(\tau)-q_0(\tau)$ we find
\be \label{sc.8}
\delta \dot{q}(\tau)=[\{H,q\}]_{m=m(\tau)}-[\{H,q\}]_{m=m_0(\tau)}
=\sum_{k=2}^\infty H^{(k)}_{,\delta p}
\ee
from which the assertion follows immediately (the proof for $\delta p$
is identical).\\
(ii.)\\
Using the explicit background dependence of
$F^{(k)}=F^{(k)}(m_0(\tau);\delta m(\tau))$
we have
\ba \label{sc.9}
\frac{d}{d\tau} F^{(k)} &=&
\frac{\partial F^{(k)}}{\partial q_0} \dot{q}_0(\tau)
+\frac{\partial F^{(k)}}{\partial p_0} \dot{p}_0(\tau)
+\frac{\partial F^{(k)}}{\partial \delta q} \delta\dot{q}(\tau)
+\frac{\partial F^{(k)}}{\partial \delta p} \delta\dot{p}(\tau)
\nonumber\\
&=& \frac{\partial F^{(k)}}{\partial q_0}
\frac{\partial H^{(1)}}{\partial \delta p}
-\frac{\partial F^{(k)}}{\partial p_0}
\frac{\partial H^{(1)}}{\partial \delta q}
+\frac{\partial F^{(k)}}{\partial \delta q}
\frac{\partial H_n}{\partial \delta p}
-\frac{\partial F^{(k)}}{\partial \delta p}
\frac{\partial H_n}{\partial \delta q}
\ea
where we used the first part of the lemma as well as the fact that
$H_{,q}(m_0)=H^{(1)}_{,\delta q}, \;H_{,p}(m_0)=H^{(1)}_{,\delta p}$.
All Poisson brackets are with respect to the coordinates $\delta q,
\delta p$.

Now observe the important fact
\ba \label{sc.10}
F^{(k)}_{,q_0} &=&
\sum_{l=0}^k \frac{1}{l! (k-l)!} \frac{\partial^{k+1}
F}{\left[\partial q\right]^{l+1} \left[\partial p\right]^{(k+1)-(l+1)}}
\left[\delta q\right]^l
\left[\delta p\right]^{(k+1)-(l+1)}
\nonumber\\
&=& \frac{\partial}{\partial \delta q}
\sum_{l=0}^k \frac{1}{(l+1)! ((k+1)-(l+1))!} \frac{\partial^{k+1}
F}{\left[\partial q\right]^{l+1} \left[\partial p\right]^{(k+1)-(l+1)}}
\left[\delta q\right]^{l+1}
\left[\delta p\right]^{(k+1)-(l+1)}
\nonumber\\
&=& \frac{\partial}{\partial \delta q}
\sum_{l=1}^{k+1} \frac{1}{l! (k+1-l)!} \frac{\partial^{k+1}
F}{\left[\partial q\right]^l \left[\partial p\right]^{k+1-l}} \left[\delta
q\right]^l
\left[\delta p\right]^{k+1-l}
\nonumber\\
&=& \frac{\partial}{\partial \delta q}
\left[F^{(k+1)}-
\frac{1}{(k+1)!} \frac{\partial^{k+1}
F}{\left[\partial p\right]^{k+1}}
\left[\delta p\right]^{k+1}\right]
\nonumber\\
&=& F^{(k+1)}_{,\delta q}
\ea
and similarly
$F^{(k)}_{,p_0}=F^{(k+1)}_{,\delta p}$.

Combining (\ref{sc.9}) and (\ref{sc.10}) we see that
\be \label{sc.11}
\frac{d F^{(k)}}{d\tau}=\{H^{(1)},F^{(k+1)}\}+\{H_n,F^{(k)}\}
\ee
Hence
\be \label{sc.12}
\frac{d F_n}{d\tau}
=\sum_{k=1}^n \left[\{H^{(1)},F^{(k+1)}\}+\sum_{l=2}^{n+1}
\{H^{(l)},F^{(k)}\}\right]
\ee
We would like to show that the terms up to order $n$ in (\ref{sc.12})
vanish identically. Since $\{H^{(l)},F^{(k)}\}$ is of order $k+l-2$, for
given $k$ we can restrict the sum over $l$ from $l=2$ until $n+2-k$
up to terms of order $O(\delta^{n+1})$. Notice that $n+2-k$ is at
least $2$ (for $k=n$) and at most $n+1$ (for $k=1$) for all values of
$k$ which is the allowed range of $l$. It follows
\ba \label{sc.13}
\frac{d F_n}{d\tau}+O(\delta^{n+1})
&=& \sum_{k=1}^n \left[\{H^{(1)},F^{(k+1)}\}+\sum_{l=2}^{n+2-k}
\{H^{(l)},F^{(k)}\}\right]
\nonumber\\
&=& \sum_{k=1}^n \left[\{H^{(1)},F^{(k+1)}\}+\sum_{r=k}^n
\{H^{(r-k+2)},F^{(k)}\}\right]
\nonumber\\
&=& \sum_{r=1}^n \left[\{H^{(1)},F^{(r+1)}\}+\sum_{k=1}^r
\{H^{(r-k+2)},F^{(k)}\}\right]
\nonumber\\
&=& \sum_{r=1}^n \; \sum_{k=1}^{r+1}
\{H^{(r-k+2)},F^{(k)}\}
\ea
where in the second step we have introduced the summation variable
$r=k+l-2$ which for given $k$ takes range in $k,.., n$ (lowest value
for $l=2$ and highest value for $l=n+2-k$) whence $l=r-k+2$, in the
third step we have changed the order of the $k$ and $r$ summation in the
second term (keeping in mind the constraint $1\le k\le r\le n$) while
the summation variable $k$ was renamed by $r$ in the first term and in
the fourth step we noticed that the first and second term can be combined by
having the $k$ summation extend to $r+1$.

Now we exploit the fact that $F$ is an exact invariant, that is
\be \label{sc.14}
0=\{H,F\}=\sum_{k,l=1}^\infty \{H^{(l)},F^{(k)}\}=\sum_{r=0}^\infty\;
\left[\sum_{k=1}^{r+1} \{H^{(r-k+2)},F^{(k)}\}\right]
\ee
where in the second step we collected all terms of order $r$ (notice
that $l=r-k+2\ge 1$ as required). Since (\ref{sc.13}) is an identity on
the entire phase space, the Taylor coefficients of $[\delta q]^k \;
[\delta p]^l$ have to vanish separately for all $k,l\ge 0$. The term
corresponding to order $r$ in (\ref{sc.13}) contains {\it all} terms of
the form $[\delta q]^s [\delta p]^{r-s},\;s=0,..,r$. Therefore we
conclude
\be \label{sc.15}
\sum_{k=1}^{r+1} \{H^{(r-k+2)},F^{(k)}\}=0
\ee
identically for all $r$. In particular, (\ref{sc.13}) implies
\be \label{sc.16}
\frac{d F_n}{d\tau}=O(\delta^{n+1})
\ee
\end{proof}
The only $n$ for which the term $O(\delta^{n+1})$ vanishes is for $n=1$
as one can see from (\ref{sc.12}) since then $k=1,l=2$ can only take one
value which already contributes to order $r=1$. Thus for $n=1$ we even
have
\be \label{sc.17}
\frac{d F_1}{d\tau}=0
\ee
It is instructive to see how the background equations find their way
into demonstrating the important result (\ref{sc.16}) which ensures that
an exact invariant expanded up to order $n$ remains an invariant up
to higher orders for the equations of motion expanded up to order $n$,
thus simplifying the task to integrate those equations of motion.


\section{Generalisation to Other Deparametrising Matter}
\label{sd} In this work dust was used as a reference frame in order
to define a physical time evolution. We chose dust because then, for
the case of no perturbations, the induced physical Hamiltonian
yields the exact FRW equations used in standard cosmology. However,
when perturbations of the metric and the scalar field are
considered, deviations from the standard FRW framework occur, which
are, however, still in agreement with observational data. Since
General Relativity does not tell us which is the right clock to use
for cosmology, we chose a clock such that the resulting physical
Hamiltonian is as close as possible to the FRW-Hamiltonian in
standard cosmology, where one uses the Hamiltonian constraint $c$ as
a true Hamiltonian. The physical Hamiltonian used here,
$H_{\mathrm{dust}}=\sqrt{C^2-Q^{ij}C_iC_j}$, reduces to
$H_{\mathrm{dust}}^{\scriptscriptstyle\mathrm{FRW}}=C$ for an FRW
universe, namely to the gauge invariant version of the Hamiltonian
constraint. However, the question arises how  generic are the
results obtained from a dust clock and what changes do we expect
when choosing other matter than dust to reparametrise or even
deparametrise the constraints of General Relativity. To illustrate
this issue, let us discuss the phantom clock introduced in \cite{8}
which leads to a physical Hamiltonian of the form \be
\HF_{\mathrm{phan}}=\int\limits_{\chi} d^3\sigma
H_{\mathrm{phan}}(\sigma), \ee with the Hamiltonian density
$H_{\mathrm{phan}}$ defined as \ba
H_{\mathrm{phan}}(\sigma)&=&\sqrt{\frac{1}{2}\left(F(C,C_i,Q_{ij})\right)
                        +\sqrt{\frac{1}{4}\left(F(C,C_i,Q_{ij})\right)^2
-\alpha^2Q^{ij}(\sigma,\tau)C_iC_j(\sigma)Q(\sigma,\tau)}} \ea where
\be F(C,C_i,Q_{ij}):=C^2(\sigma,\tau) -
Q^{ij}(\sigma,\tau)C_iC_j(\sigma) -\alpha^2Q(\sigma,\tau).
 \ee Above
we introduced the abbreviation $Q=\det(Q_{ij})$, and $\alpha>0$ is
for the moment an arbitrary constant of dimension $cm^{-2}$ that
enters the phantom field action as a free parameter. Recall that the
expressions for $C$ and $C_j$ were given by the geometry and matter
part of the total Hamiltonian and diffeomorphism constraint,
respectively, that is \be
C(\sigma,\tau)=C^{\mathrm{\rm geo}}(\sigma,\tau)+C^{\mathrm{matter}}(\sigma,
\tau)\quad\mathrm{and}\quad
C_j(\sigma)=C_j^{\mathrm{\rm geo}}(\sigma)+C_j^{\mathrm{matter}}(\sigma).
\ee From now on we will drop the $\tau$ dependence of
$C(\sigma,\tau)$ and $Q_{ij}(\sigma,\tau)$ in the expression for
$H_{\mathrm{phan}}$ and write them explicitly only when confusion
could arise otherwise.

For the dust Hamiltonian we saw that the first order equations of
motion obtained for $\Xi, \Pi$ and $Q_{ij},P^{ij}$ look similar to
the standard cosmological equations apart from the fact that in the
case of a dust clock we obtain a dynamical, that is phase space
dependent, lapse function $N_{\mathrm{dust}}=C/H_{\mathrm{dust}}$
and shift vector $N^i_{\mathrm{dust}}=N^i/H_{\mathrm{dust}}$.
Moreover, in the general equations
there occurs a term proportional to $N^i_{\mathrm{dust}}N^j_{\mathrm{dust}}$. \\
\\
\\
With respect to $\HF_{\mathrm{phan}}$, we want to analyse now
whether it produces a similar effect and leads to possibly more
serious deviations from the standard equations. Starting with the
first order equation for $\Xi$ we obtain \ba
\dot{\Xi}(\sigma,\tau)&=&\{\HF_{\mathrm{phan}},\Xi(\sigma,\tau)\}\\
&=&\int\limits_{\chi}d^3\sigma'\Big[\frac{C}{H_{\mathrm{phan}}}
(\sigma')\left(\frac{1}{2}
+\frac{(C^2-Q^{ij}C_iC_j-\alpha^2Q)(\sigma')}{4\sqrt{(\frac{1}{4}(C^2-Q^{ij}
C_iC_j-\alpha^2Q)^2-\alpha^2Q^{ij}C_iC_jQ)(\sigma')}}\right)\{C^{\mathrm{matter}
}(\sigma'),\Xi(\sigma,\tau)\}\nonumber \\
&&-\frac{Q^{ij}C_j}{H_{\mathrm{phan}}}(\sigma')\left(\frac{1}{2}
+\frac{(C^2-Q^{ij}C_iC_j+\alpha^2Q)(\sigma')}{4\sqrt{(\frac{1}{4}(C^2-Q^{ij}
C_iC_j-\alpha^2Q)^2-\alpha^2Q^{ij}C_iC_jQ)(\sigma')}}\right)\{C^{\mathrm{matter}
}(\sigma'),\Xi(\sigma,\tau)_i\}\Big]. \nonumber
\ea
Introducing the
dynamical lapse function $N_{\mathrm{phan}}$ and the dynamical shift
covector $N^i_{\mathrm{phan}}$ as \ba \label{NNaphan}
N_{\mathrm{phan}}(\sigma)&:=&\frac{C}{H_{\mathrm{phan}}}(\sigma)\left(\frac{1}{2
}
+\frac{(C^2-Q^{ij}C_iC_j-\alpha^2Q)(\sigma)}{4\sqrt{(\frac{1}{4}(C^2-Q^{ij}
C_iC_j-\alpha^2Q)^2-\alpha^2Q^{ij}C_iC_jQ)(\sigma)}}\right)\nonumber\\
N^i_{\mathrm{phan}}(\sigma)&:=&-\frac{Q^{ij}C_j}{H_{\mathrm{phan}}}
(\sigma)\left(\frac{1}{2}
+\frac{(C^2-Q^{ij}C_iC_j+\alpha^2Q)(\sigma)}{4\sqrt{(\frac{1}{4}(C^2-Q^{ij}
C_iC_j-\alpha^2Q)^2-\alpha^2Q^{ij}C_iC_jQ)(\sigma)}}\right),
 \ea
 we 
can rewrite the first order equation of motion for $\Xi$ as \be
\dot{\Xi}(\sigma,\tau)=\{\HF_{\mathrm{phan}},\Xi(\sigma,\tau)\}=\int\limits_{
\chi}d^3\sigma'\left(
N_{\mathrm{phan}}(\sigma')\{C^{\mathrm{matter}}(\sigma'),\Xi(\sigma,\tau)\}
+
N^i_{\mathrm{phan}}\{C^{\mathrm{matter}}_i(\sigma'),\Xi(\sigma,\tau)\}\right).
\ee Hence, we realise that, similarly to the case of the dust clock,
the effect of the phantom clock results in the appearance of a
dynamical lapse function and a dynamical shift vector. However, due
to the more complicated structure of $\HF_{\mathrm{phan}}$ as compared
to $\HF_{\mathrm{dust}}$,  $N_{\mathrm{phan}}$ and
$N^i_{\mathrm{phan}}$ are not simply given in terms of $C/H_{\rm phan}$ and $-C_j/H_{\rm phan}$ respectively. 
Now also terms occur which include higher than
linear powers of the constraints in the nominator and denominator.
Since the quantities $\Pi$ and obviously $Q_{ij}$ also Poisson
commute with any function that does depend on $Q_{ij}$ only, the
calculation works analogously in these cases and we obtain \be
\dot{\Pi}(\sigma,\tau)=\{\HF_{\mathrm{phan}},\Pi(\sigma,\tau)\}=\int\limits_{
\chi}d^3\sigma'\left(
N_{\mathrm{phan}}(\sigma')\{C^{\mathrm{matter}}(\sigma'),\Xi(\sigma)\}
+
N^i_{\mathrm{phan}}\{C^{\mathrm{matter}}_i(\sigma'),\Pi(\sigma,\tau)\}\right)
\ee and \be
\dot{Q}_{ij}(\sigma,\tau)=\{H_{\mathrm{phan}},Q_{ij}(\sigma)\}=\int\limits_{\chi
}d^3\sigma'\left(
N_{\mathrm{phan}}(\sigma')\{C^{\mathrm{geo}}(\sigma'),Q_{ij}(\sigma,\tau)\}
+
N^i_{\mathrm{phan}}\{C^{\mathrm{geo}}_i(\sigma'),Q_{ij}(\sigma,\tau)\}\right).
\ee For the dynamical variable $P^{ij}$ things look slightly
different, because $P^{ij}$ does not Poisson commute with functions
depending on $Q_{ij}$. Therefore we get, as  in the dust case,
an additional contribution proportional to the Poisson bracket
$\{Q^{kl}(\sigma'),P^{ij}(\sigma,\tau)\}$. Furthermore, since
$\HF_{\mathrm{phan}}$ includes a term of the form $\alpha^2Q$, we
also obtain a term proportional to the Poisson bracket
$\{Q(\sigma'),P^{ij}(\sigma,\tau)\}$. The explicit results for these
Poisson brackets are \ba
\{Q^{kl}(\sigma'),P^{ij}(\sigma,\tau)\}&=&\frac{\kappa}{2}(Q^{ik}Q^{lj}+Q^{jk}Q^{il}
)(\sigma',\tau)\delta^3(\sigma',\sigma)\nonumber\\
\{Q(\sigma'),P^{ij}(\sigma,\tau)\}&=&-\kappa(Q
Q^{ij})(\sigma',\tau)\delta^3(\sigma',\sigma).
 \ea Inserting this back
into the eqn. for $\dot{P}^{ij}$, we end up with \ba
\dot{P}^{ij}(\sigma,\tau)&=&\{\HF_{\mathrm{phan}},P^{ij}(\sigma,\tau)\}
\nonumber\\
&=&\int\limits_{\chi}d^3\sigma'\left(
N_{\mathrm{phan}}(\sigma')\{C^{\mathrm{geo}}(\sigma'),P^{ij}(\sigma,\tau)\} +
N^i_{\mathrm{phan}}\{C^{\mathrm{geo}}_i(\sigma'),P^{ij}(\sigma,\tau)\}
\right)\nonumber\\
&&-\frac{\kappa}{2}\left(H_{\mathrm{phan}}N^i_{\mathrm{phan}}N^j_{\mathrm{phan}}\left
[\frac{1}{2}
+\frac{(C^2-Q^{ij}C_iC_j+\alpha^2Q)}{4\sqrt{(\frac{1}{4}(C^2-Q^{ij}
C_iC_j-\alpha^2Q)^2-\alpha^2Q^{ij}C_iC_jQ)}}\right]^{-1}\right)(\sigma,
\tau)\nonumber\\
&&+\frac{\kappa}{2}\left(\frac{\alpha^2Q^{ij}Q}{H_{\mathrm{phan}}}\left[\frac{1}{2}
+\frac{C^2 + Q^{ij}C_iC_j-\alpha^2Q}{4\sqrt{\frac{1}{4}(C^2-Q^{ij}
C_iC_j-\alpha^2Q)^2-\alpha^2QQ^{ij}C_iC_j}}\right]\right)(\sigma,\tau).
\ea The remaining Poisson brackets in the first order equations for
the dynamical variables are the same Poisson brackets that occur
when $\HF_{\mathrm{dust}}$ is used as a Hamiltonian. Inserting the
results obtained there into the corresponding equations for the case
of $\HF_{\mathrm{phan}}$, we obtain the following final form of the
first order equations: \ba \label{FOeqn1}
\dot{\Xi}(\sigma,\tau)&=&\frac{N_{\mathrm{phan}}\Pi}{Q}(\sigma,\tau)
+
\left({\cal L}_{\vec{N}^{\mathrm{phan}}}\Xi\right)(\sigma,\tau)\nonumber\\
\dot{\Pi}(\sigma,\tau)&=&\left[N_{\mathrm{phan}}QQ^{ij}\Xi,i\right]_{,j}(\sigma,
\tau)-\frac{1}{2}(N_{\mathrm{phan}}QV'(\Xi))(\sigma,\tau) + \left({\cal
L}_{\vec{N}_{\mathrm{phan}}}\Pi\right)(\sigma,\tau)\nonumber\\
\dot{Q}_{ij}(\sigma,\tau)&=&\frac{2N_{\mathrm{phan}}}{Q}(\sigma,\tau)\left(G_{
ijmn}P^{mn}\right)(\sigma,\tau) + \left({\cal
L}_{\vec{N}_{\mathrm{\mathrm{phan}}}}Q\right)_{ij}(\sigma,\tau),
 \ea
where \be G_{ijmn}:=\frac{1}{2}\left(Q_{im}Q_{jn}+Q_{in}Q_{jm} -
Q_{ij}Q_{mn}\right), \ee with its inverse given by \be
[G^{-1}]^{ijmn}:=\frac{1}{2}\left(Q^{im}Q^{jn}+Q^{in}Q^{jm} -
2Q^{ij}Q^{mn}\right). \ee For the gravitational momentum we have \ba
\label{DotP}
\dot{P}^{ij}(\sigma,\tau)&=&\Big(N_{\mathrm{phan}}\Big[-\frac{Q_{mn}}{\sqrt{\det
{Q}}}\Big(2 P^{jm} P^{kn}-P^{jk} P^{mn}\Big)
+\frac{\kappa}{2}Q^{jk}\;C - \Q\;Q^{jk}\big(2\Lambda
+\frac{\kappa}{2\lambda}\big(\Xi^{,m}\Xi_{,m}+v(\Xi)\big)\big)\Big]
\nonumber\\
&& +\Q [G^{-1}]^{jkmn}\Big((D_mD_n
N_{\mathrm{phan}})-N_{\mathrm{phan}}R_{mn}[Q]\Big)
+\frac{\kappa}{2\lambda}N_{\mathrm{phan}}\Q\;\Xi^{,j}\Xi^{,k}
 +({\cal L}_{\vec{N}_{\mathrm{phan}}} P)^{jk}\Big)(\sigma,\tau)\nonumber\\
&&-\frac{\kappa}{2}\left(H_{\mathrm{phan}}N^i_{\mathrm{phan}}N^j_{\mathrm{phan}}\left
[\frac{1}{2}
+\frac{(C^2-Q^{ij}C_iC_j+\alpha^2Q)}{4\sqrt{(\frac{1}{4}(C^2-Q^{ij}
C_iC_j-\alpha^2Q)^2-\alpha^2Q^{ij}C_iC_jQ)}}\right]^{-1}\right)(\sigma,
\tau)\nonumber\\
&&+\frac{\kappa}{2}\left(\frac{\alpha^2Q^{ij}Q}{H_{\mathrm{phan}}}\left[\frac{1}{2}
+\frac{C^2 + Q^{ij}C_iC_j-\alpha^2Q}{4\sqrt{\frac{1}{4}(C^2-Q^{ij}
C_iC_j-\alpha^2Q)^2-\alpha^2QQ^{ij}C_iC_j}}\right]\right)(\sigma,\tau).
\ea One can see that the first order equations for $\Xi,\Pi$ and
$Q_{ij}$ in equation (\ref{FOeqn1}) are identical to those for the
dust Hamiltonian $\HF_{\mathrm{dust}}$, apart from the different
definitions of the dynamical lapse function $N_{\mathrm{phan}}$ and
shift vector $N^i_{\mathrm{phan}}$ in equation (\ref{NNaphan}).
However, the equation for $\dot{P}^{ij}$ differs from the
corresponding dust Hamiltonian equation. The term in the second last
line in equation (\ref{DotP}) corresponds to the term
$-\frac{\kappa}{2}(H_{\mathrm{dust}}N^i_{\mathrm{dust}}N^j_{\mathrm{dust}})(\sigma)$
 in the equation for $\dot{P}^{ij}$ derived from
$\HF_{\mathrm{dust}}$. In the present case this term looks a bit
more complicated, since we have to divide the whole expression by
the term in the square brackets which is identical to one in case of
$\HF_{\mathrm{dust}}$. The additional term in the last
line of equation (\ref{DotP}) comes from the terms $\alpha^2Q$ and
$\alpha^2Q^{ij}C_iC_j$ in $\HF_{\mathrm{phan}}$, which are absent in
$\HF_{\mathrm{dust}}$. Since in this term we cannot factor out $C$
or $C_i$ but only $\alpha^2$, we are also not able to reexpress
this term by means of the dynamical lapse function
$N_{\mathrm{phan}}$ and the dynamical shift vector
$N^i_{\mathrm{phan}}$, respectively, as it was possible for the term
in the second last line.
\\
\\
In summary, when using a phantom scalar field as a clock we also
obtain deviations from the standard treatment in which the
Hamiltonian constraint is used as a true Hamiltonian.  These
deviations manifest themselves in the appearance of a dynamical
lapse function $N_{\mathrm{phan}}$ and a dynamical shift vector
$N^i_{\mathrm{phan}}$, analogous to the case where dust is used as a
clock. However, the explicit dependence of $N_{\mathrm{phan}}$ and
$N^i_{\mathrm{phan}}$ on the dynamical variables is more complicated
than for the corresponding quantities $N_{dust}$ and $N^i_{dust}$.
Another modification occurs in the equation for the gravitational
momentum $P^{ij}$. While it contains a term that is second order in
$N^i_{\mathrm{phan}}$, in complete analogy with the case of
$\HF_{\mathrm{dust}}$, it also features an additional term
proportional to $\alpha^2$.
\subsection*{The Special Case of an FRW Universe}
It is interesting to study the special case of FRW also for
$\HF_{\mathrm{phan}}$. Recall from \cite{8a} that by assuming
homogeneity and isotropy, $\HF_{\mathrm{dust}}, N_{\mathrm{dust}}$
and $N^i_{\mathrm{dust}}$ reduce to the following quantities \be
N_{\mathrm{dust}}(\sigma)\stackrel{\mathrm{\scriptscriptstyle
FRW}}\longrightarrow 1,\quad
N^i_{\mathrm{dust}}(\sigma)\stackrel{\mathrm{\scriptscriptstyle
FRW}}\longrightarrow 0,\quad
H_{\mathrm{dust}}(\sigma)\stackrel{\mathrm{\scriptscriptstyle
FRW}}\longrightarrow C^{\scriptscriptstyle\mathrm{FRW}}(\sigma), \ee
where \be
C^{\scriptscriptstyle\mathrm{FRW}}_{\mathrm{dust}}(\tau)=A^3
(\tau)\Big(\frac{1}{\kappa}\big(-6\Big(\frac{\dot{A}}{A}\Big)^2 +
2\Lambda\big) + \frac{1}{2\lambda} \left(\dot{\Xi}^2 +
V(\Xi)\right)\Big)(\tau). \ee Here $A=A(\tau)$ is a function of dust
time $\tau$ and the dot refers to a derivative with respect to dust
time. In particular $A(\tau)$ can be understood as the gauge
invariant extension of the ordinary scale factor $a(t)$ used in
standard cosmology. The difference between the two is that $A$ is
gauge invariant and thus a physical observable whereas $a$ is not,
since it does not commute with the Hamiltonian  constraint of FRW.
\\
The gravitational canonical variables are given by
$Q_{ij}=A^2(\tau)\delta_{ij}$ and
$P^{ij}=-2\dot{A}(\tau)\delta^{ij}$. We mentioned previously that a
consequence of this behaviour is that the unperturbed equations of
motion for $(\Xi,\Pi)$ and $(Q_{ij},P^{ij})$ agree with the FRW
equations used in standard cosmology. In particular, the deviation
from the general standard equation of motion for $P^{ij}$ vanishes
in the case of FRW, because it is quadratic in
$N^i_{\mathrm{dust}}$.
\\
For the phantom Hamiltonian $\HF_{\mathrm{phan}}$ things look
slightly different. Here we have the following behaviour of
$H_{\mathrm{phan}}, N_{\mathrm{phan}}$ and $N^i_{\mathrm{phan}}$,
when a homogenous and isotropic universe is considered: \be
\label{Limit}
N_{\mathrm{phan}}(\sigma)\stackrel{\mathrm{\scriptscriptstyle
FRW}}\longrightarrow
N^{\scriptscriptstyle\mathrm{FRW}}_{\mathrm{phan}}:=
\left(\frac{C^{\scriptscriptstyle\mathrm{FRW}}_{\mathrm{phan}}}{H^{
\scriptscriptstyle\mathrm{FRW}}_{\mathrm{phan}}}\right)(\sigma),\quad
N^i_{\mathrm{phan}}(\sigma)\stackrel{\mathrm{\scriptscriptstyle
FRW}}\longrightarrow 0,\quad
H_{\mathrm{phan}}(\sigma)\stackrel{\mathrm{\scriptscriptstyle
FRW}}\longrightarrow
H^{\scriptscriptstyle\mathrm{FRW}}_{\mathrm{phan}}:=
\sqrt{(C^{\scriptscriptstyle\mathrm{FRW}}_{\mathrm{phan}})^2-\alpha^2A^6}
(\sigma) \ee with \be
C^{\scriptscriptstyle\mathrm{FRW}}_{\mathrm{phan}}(\tau)=\frac{A^3}{(N^{
\scriptscriptstyle\mathrm{FRW}}_{\mathrm{phan}})^2}(\tau)\Big(\frac{1}{\kappa}
\big(-6\Big(\frac{\dot{A}}{A}\Big)^2 +
2\Lambda(N^{\scriptscriptstyle\mathrm{FRW}}_{\mathrm{phan}})^2\big)
+ \frac{1}{2\lambda} \left(\dot{\Xi}^2 +
V(\Xi)(N^{\scriptscriptstyle\mathrm{FRW}}_{\mathrm{phan}})^2\right)\Big)(\tau).
\ee Hence, the dust clock and the phantom scalar field clock agree
only if the parameter $\alpha$ is chosen to be tiny compared to the
Hamiltonian constraint
$C^{\scriptscriptstyle\mathrm{FRW}}_{\mathrm{phan}}$, e.g. $\alpha
A^3\ll C^{\scriptscriptstyle\mathrm{FRW}}_{\mathrm{phan}}$.
Consequently, the equations of motion generated by
$\HF_{\mathrm{phan}}$ also deviate from the standard FRW equations.
The significance of this deviation again depends on the specific
value of the parameter $\alpha$, as was discussed in detail in
\cite{8}. For completeness we also list the first order equations of
motion here \ba
\dot{\Xi}(\tau)&=&\big(N^{\scriptscriptstyle\mathrm{FRW}}_{\mathrm{phan}}\frac{
\Pi}{\sqrt{Q}}\big)(\tau)\\
\dot{\Pi}(\tau)&=&-\big(N^{\scriptscriptstyle\mathrm{FRW}}_{\mathrm{phan}}\frac{
\sqrt{Q}}{2}V'(\Xi)\big)(\tau)\nonumber\\
\dot{Q}_{ij}&=&\big(N^{\scriptscriptstyle\mathrm{FRW}}_{\mathrm{phan}}\frac{2}{
\sqrt{Q}}G_{ijmn}P^{mn}\big)(\tau)\nonumber\\
\dot{P}^{ij}&=&N^{\scriptscriptstyle\mathrm{FRW}}_{\mathrm{phan}}(\tau)
\left(-\frac{Q_{mn}}{\sqrt{Q}}\left(2P^{im}P^{jn}-P^{ij}P^{mn}\right)
+\frac{\kappa}{2}Q^{ij}C^{\scriptscriptstyle\mathrm{FRW}}_{\mathrm{phan}}
-
Q^{ij}\sqrt{Q}(2\Lambda+\frac{\kappa}{2\lambda}V(\Xi))
+\frac{\alpha^2\kappa}{2}\frac{Q}{C^{\rm\scriptscriptstyle FRW}}Q^{ij}
\right)(\tau).\nonumber
\ea Taking into account that $Q_{ij}=A^2\delta_{ij}$ and using the
first order equation for $Q_{ij}$, we solve for the momenta
$P^{ij}=-2\dot{A}/N^{\scriptscriptstyle\mathrm{FRW}}_{\mathrm{phan}}\delta^{ij}$
in terms of $\dot{Q}_{jk}=2A\dot{A}$. In order to derive the
corresponding FRW equation with respect to the time generated by
$\HF^{\scriptscriptstyle\mathrm{FRW}}_{\mathrm{phan}}$, we take the
time derivative of the equation for $\dot{Q}_{ij}$ and insert into
the resulting equation for $\ddot{Q}_{ij}$ the expression for
$P^{ij}$ and $\dot{P}^{ij}$ given above. This yields \be
\label{ddotA1}
\Big(\frac{\ddot{A}}{A}\Big)=-\frac{1}{2}\Big(\frac{\dot{A}}{A}\Big)^2+\frac{1}{
2}(N^{\scriptscriptstyle\mathrm{FRW}}_{\mathrm{phan}})^2\Lambda-\frac{\kappa}{
4\lambda}\big(\frac{1}{2}\dot{\Xi}^2-\frac{1}{2}(N^{\scriptscriptstyle\mathrm{
FRW}}_{\mathrm{phan}})^2V(\Xi)\big)
+\frac{\dot{N}^{\scriptscriptstyle\mathrm{FRW}}_{\mathrm{phan}}}{N^{
\scriptscriptstyle\mathrm{FRW}}_{\mathrm{phan}}}\Big(\frac{\dot{A}}{A}\Big)-\frac{\kappa}{4}(N^{\scriptscriptstyle\mathrm{
FRW}}_{\mathrm{phan}})^2\Big(\frac{\alpha^2 A^3}{C^{\scriptscriptstyle\mathrm{FRW}}_{\mathrm{phan}}}\Big).
\ee 
Apart from the lapse functions in the equation above which are
not present in the standard FRW equations, we get an additional term
including the time derivative of the lapse function. Using the
explicit definition of the lapse function, we can perform this time
derivative, leading to \be
\frac{\dot{N}^{\scriptscriptstyle\mathrm{FRW}}_{\mathrm{phan}}}{N^{
\scriptscriptstyle\mathrm{FRW}}_{\mathrm{phan}}}\Big(\frac{\dot{A}}{A}
\Big)=\frac{3\big((N^{\scriptscriptstyle\mathrm{FRW}}_{\mathrm{phan}})^2-1\big)}
{(N^{\scriptscriptstyle\mathrm{FRW}}_{\mathrm{phan}})^2}\Big(\frac{\dot{A}}{A}
\Big)^2. \ee Consequently, equation (\ref{ddotA1}) can be rewritten
as \be \label{ddotA2}
\Big(\frac{\ddot{A}}{A}\Big)=-\Big(\frac{\dot{A}}{A}\Big)^2\Big(\frac{1}{2}-
\frac{3\big((N^{\scriptscriptstyle\mathrm{FRW}}_{\mathrm{phan}})^2-1\big)}{(N^{
\scriptscriptstyle\mathrm{FRW}}_{\mathrm{phan}})^2}\Big)
+\frac{1}{2}(N^{\scriptscriptstyle\mathrm{FRW}}_{\mathrm{phan}})^2\Lambda-\frac{
\kappa}{4\lambda}\big(\frac{1}{2}\dot{\Xi}^2-\frac{1}{2}(N^{
\scriptscriptstyle\mathrm{FRW}}_{\mathrm{phan}})^2V(\Xi)\big)
-\frac{\kappa}{4}(N^{\scriptscriptstyle\mathrm{
FRW}}_{\mathrm{phan}})^2\Big(\frac{\alpha^2 A^3}{C^{\scriptscriptstyle\mathrm{FRW}}_{\mathrm{phan}}}\Big).
\ee
 The term $(\dot{A}/A)^2$ can be replaced by considering the energy conservation
law  $\dot{H^{\scriptscriptstyle\mathrm{FRW}}_{\mathrm{phan}}}=0$ ,
that is
$H^{\scriptscriptstyle\mathrm{FRW}}_{\mathrm{phan}}=\epsilon_0$,
from which we get
$C^{\scriptscriptstyle\mathrm{FRW}}_{\mathrm{phan}}=\sqrt{
\epsilon_0^2+\alpha^2A^6}$. Solving this equation for
$(\dot{A}/A)^2$ yields \ba \label{ConsLaw}
3\Big(\frac{\dot{A}}{A}\Big)^2&=&(N^{\scriptscriptstyle\mathrm{FRW}}_{\mathrm{
phan}})^2\Lambda+\frac{\kappa}{2\lambda}\Big(\frac{1}{2}\dot{\Xi}^2+\frac{1}{2}
V(\Xi)(N^{\scriptscriptstyle\mathrm{FRW}}_{\mathrm{phan}})^2\Big)-\frac{\kappa}{
2}
(N^{\scriptscriptstyle\mathrm{FRW}}_{\mathrm{phan}})^2\alpha\sqrt{1+\frac{
\epsilon_0}{\alpha^2A^6}}\nonumber\\
&=&
(N^{\scriptscriptstyle\mathrm{FRW}}_{\mathrm{phan}})^2\Lambda+\frac{\kappa}{
2\lambda}(N^{\scriptscriptstyle\mathrm{FRW}}_{\mathrm{phan}})^2\Big(\rho_{
\mathrm{matter}}+\rho_{\mathrm{phan}}\Big),
 \ea where we used in the
last line \be
\rho_{\mathrm{matter}}=\frac{1}{2}\frac{1}{(N^{\scriptscriptstyle\mathrm{FRW}}_{
\mathrm{phan}})^2}\dot{\Xi}^2+\frac{1}{2}V(\Xi)
\quad\mathrm{and}\quad
\rho_{\mathrm{phan}}=-\alpha\sqrt{1+\frac{\epsilon_0}{\alpha^2A^6}}.
\ee Reinserting equation (\ref{ConsLaw}) into equation
(\ref{ddotA2}), we obtain the phantom FRW equation given by \be
3\Big(\frac{\ddot{A}}{A}\Big)=\Lambda\Big(1+4\big((N^{\scriptscriptstyle\mathrm{
FRW}}_{\mathrm{phan}})^2-1\big)\Big)
-\frac{\kappa}{4}\Big[\Big(\frac{1}{\lambda}\rho_{\mathrm{matter}}+\rho_{\mathrm
{phan}}\Big)\Big(1-5\big((N^{\scriptscriptstyle\mathrm{FRW}}_{\mathrm{phan}}
)^2-1\big)\Big)+3(N^{
\scriptscriptstyle\mathrm{FRW}}_{\mathrm{phan}})^2\Big(\frac{1}{\lambda}p_{\mathrm{matter}}
+p_{\mathrm{phan}}\Big)
\Big], \ee whereby
we introduced \ba
p_{\mathrm{matter}}&=&\frac{1}{2}\frac{1}{(N^{\scriptscriptstyle\mathrm{FRW}}_{
\mathrm{phan}})^2}\dot{\Xi}^2-\frac{1}{2}V(\Xi)\nonumber\\
p_{\mathrm{matter}}&=&-\frac{1}{3A^2}\frac{d}{dA}\Big(A^3\rho_{\rm phan}\Big)=\frac{\alpha^2 A^3}{C^{
\scriptscriptstyle\mathrm{FRW}}_{\mathrm{phan}}}
. \ea That this
equation agrees with the one derived in \cite{8} can be seen when
expressing $(N^{\scriptscriptstyle\mathrm{FRW}}_{\mathrm{phan}})^2$
in terms of the deviation parameter $x:=\epsilon_0/\alpha^2A^6$ used
there, resulting in
$(N^{\scriptscriptstyle\mathrm{FRW}}_{\mathrm{phan}})^2=1+1/x$. 
\\
\\
In general, choosing one clock or the other might have significant
effects on the equation of motion. General Relativity does not tell
us which clock is convenient to work with, hence additional physical
input is needed. The results of the application of this framework
for FRW in \cite{8a} show that choosing dust as clock reproduces the
standard FRW equations. Thus we could call the dust clock the
FRW-clock. Since so far an (approximate) FRW universe is in
agreement with observational data, dust seems to be a good choice.
However, the $\alpha$ parameter in
$\HF^{\scriptscriptstyle\mathrm{FRW}}_{\mathrm{phan}}$ can be chosen
such that the resulting equation of motion also do not contradict
present experiments. Therefore, based on experimental constraints,
none of the two clocks is excluded, nor is one of them preferred.
From a theoretical point of view, the choice of a clock is mainly
guided by the requirement that the constraints can be
deparametrised, that is, they can be written in the form
$C=p^{\mathrm{clock}}+H^{\mathrm{clock}}$. Here the Hamiltonian
density $H^{\mathrm{clock}}$ must not depend on the clock variables
anymore, and furthermore it should be positive definite.
Additionally the structure of $H^{\mathrm{clock}}$ should not be too
complicated such that calculations of, for
instance, the equation of motions are still possible.\\
However, in principle, we have a large amount of freedom to choose a
 clock, as long as the induced equations of motion do not contradict
experiments.

\section{Linear Perturbation  Theory: Some Calculations in More Detail}
\label{LPA}

In section \ref{DerDQ} we derived the second order
equation of motion for the linear perturbation
 of the (manifestly) gauge invariant three metric $\delta Q_{jk}$. For that we
needed the
  perturbation of the geometry and matter part of the gauge invariant
Hamiltonian
  constraint, denoted by $C^{\mathrm{geo}}$ and $C^{\mathrm{matter}}$,
respectively, as these terms occur in
  the third term on the right-hand side of the unperturbed equation of motion
for $Q_{jk}$,
  equation (\ref{ResDDQ}). We omitted the details in the main text due to their
length, and also
  because it turns out that several terms cancel when inserted back into the
  expression of the perturbation of the third term, shown in equation
(\ref{Per31}).
   For the interested reader, however, the detailed perturbations of the
constraints are given below.
\\
The perturbed geometry constraint $\delta C^{\mathrm{geo}}$ is given by
\ba
\delta C^{\mathrm{geo}}
&=&
\Big[\QN{j}{k}\overline{C}^{\mathrm{geo}}
+\frac{1}{2}\frac{\ov{N}^j\ov{N}^k}{\ov{N}^2}\Big(\ov{C}^{\mathrm{geo}}-2\frac{
\Qb}{\kappa}\big(2\Lambda-\ov{R}\big)\Big)
+\frac{\Qb}{\kappa}\Big(\ov{R}^{jk}
-[\ov{G}^{-1}]^{jkmn}\ov{D}_m\ov{D}_n\Big)\nonumber\\
&&
\quad
+\frac{\Qb}{2\ov{N}^2\kappa}\Big(\dot{\ov{Q}}_{mn}-\big({\cal
L}_{\Nb}\ov{Q}\big)_{mn}\Big)
\Big[[\ov{G}^{-1}]^{jkmn}\Big(\frac{\partial}{\partial\tau}-{\cal L}_{\Nb}\Big)
-\Big(\dot{\ov{Q}}_{rs}-\big({\cal
L}_{\Nb}Q\big)_{rs}\Big)\ov{Q}^{nk}[\ov{G}^{-1}]^{rsmj}\Big]\nonumber\\
&&
\quad
+\frac{\Qb}{2\ov{N}^2\kappa}\Big(\dot{\ov{Q}}_{rs}-\big({\cal
L}_{\Nb}Q\big)_{rs}\Big)[\overline{G}^{-1}]^{rsmn}
\Big[\ov{Q}^{jt}\ov{N}^k[\ov{Q}_{mn}]_{,t}
+2\overline{Q}_{tn}\frac{\partial}{\partial x^m}\Big(\ov{Q}^{jt}\ov{N}^k\Big)\Big]
\Big]\delta Q_{jk}\nonumber\\
&&
+\Big[
-2\Nv{j}\Big(\ov{C}^{\mathrm{geo}}-\frac{\Qb}{\kappa}\big(2\Lambda-\ov{R}
\big)\Big)\nonumber\\
&& \quad
-\frac{\Qb}{2\ov{N}^2\kappa}\Big(\dot{\ov{Q}}_{rs}-\big({\cal
L}_{\Nb}Q\big)_{rs}\Big)[\ov{G}^{-1}]^{rsmn}
\Big[\ov{Q}^{jt}[\ov{Q}_{rs}]_{,t}
+2\ov{Q}_{ts}\frac{\partial}{\partial x^k}\Big(\ov{Q}^{jt}\Big)\Big]\Big]\delta
N_j.
 \ea
 Here we used that the perturbation of the Ricci scalar can
be written as \be \delta
R=\Big[[\ov{G}^{-1}]^{jkmn}\ov{D}_m\ov{D}_n-\ov{R}^{jk}\Big]\delta
Q_{jk}. \ee For the perturbed matter part of the constraint $\delta
C^{\mathrm{matter}}$ we obtain \ba \delta C^{\mathrm{matter}} &=&
\Big[\QN{j}{k}\overline{C}^{\mathrm{matter}}
+\frac{1}{2}\frac{\ov{N}^j\ov{N}^k}{\ov{N}^2}\Big(\ov{C}^{\mathrm{matter}}-\frac
{\Qb}{\lambda}\Big(\ov{Q}^{jk}\ov{\Xi}_{,j}\ov{\Xi}_{,k}+v(\Xi)\Big)\Big)\\
&&
\quad
+\frac{\Qb}{\ov{N}^2\lambda}\ov{Q}^{jm}\ov{Q}^{kn}\big(\dot{\Xi}-\big({\cal
L}_{\Nb}\ov{\Xi}\big)\big)\ov{N}_m\ov{\Xi}_{,n}
-\frac{\Qb}{2\lambda}\ov{Q}^{jm}\ov{Q}^{jk}\ov{\Xi}_{,m}\ov{\Xi}_{,n}
\Big]\delta Q_{jk}\nonumber\\
&& +\Big[ -\frac{\Qb}{\ov{N}^2\lambda}\big(\dot{\Xi}-\big({\cal
L}_{\Nb}\ov{\Xi}\big)\big)\ov{Q}^{jk}\ov{\Xi}_{,k}
-2\Nv{j}\Big(\ov{C}^{\mathrm{matter}}-\frac{\Qb}{2\lambda}\Big(\ov{Q}^{jk}\ov{
\Xi}_{,j}\ov{\Xi}_{,k}+v(\ov{\Xi})\Big)\Big)\Big]
\delta N_j\nonumber\\
&& +\Big[
\frac{\Qb}{\lambda}\Big(\frac{1}{\ov{N}^2}\Big(\dot{\ov{\Xi}}-\big({\cal
L}_{\Nb}\ov{\Xi}\big)\Big)\Big(\frac{\partial}{\partial\tau}-{\cal
L}_{\Nb}\Big) +\ov{Q}^{jk}\ov{\Xi}_{,k}\frac{\partial}{\partial x^j}
+\frac{1}{2}v^{\prime}(\ov{\Xi})\Big)\Big]\delta\Xi .\nonumber
 \ea
Since the perturbation of $\delta C=\delta C^{\mathrm{geo}}+\delta
C^{\mathrm{matter}}$ occurs in equation (\ref{Per31}) multiplied by
a factor of $\frac{\kappa\ov{N}^2}{2\Qb}$, we will present it here
already with this factor in front \ba
\lefteqn{\frac{\kappa\ov{N}^2}{2\Qb}\delta C}
\\
&=&
\Big[\QN{m}{n}\frac{\kappa\ov{N}^2}{2\Qb}\overline{C}
\nonumber\\
\quad
&&+\frac{1}{2}\frac{\ov{N}^m\ov{N}^n}{\ov{N}^2}\frac{\kappa\ov{N}^2}{2\Qb}
\Big(\ov{C}-2\frac{\Qb}{\kappa}\big(2\Lambda-\ov{R}\big)
-\frac{\Qb}{\lambda}\big(\ov{Q}^{rs}\ov{\Xi}_{,r}\ov{\Xi}_{,s}+v(\Xi)\big)
\Big)\nonumber\\
&&
\quad
+\frac{\ov{N}^2}{2}\Big(\ov{R}^{mn}
-[\ov{G}^{-1}]^{mnrs}\ov{D}_r\ov{D}_s\Big)+\frac{\kappa}{2\lambda}\ov{Q}^{mr}\ov
{Q}^{ns}\big(\dot{\Xi}-\big({\cal
L}_{\Nb}\ov{\Xi}\big)\big)\ov{N}_r\ov{\Xi}_{,s}
-\frac{\kappa\ov{N}^2}{2\lambda}\ov{Q}^{jm}\ov{Q}^{kn}\ov{\Xi}_{,j}\ov{\Xi}_{,k}
\nonumber\\
&&
\quad
+\frac{1}{4}\Big(\dot{\ov{Q}}_{rs}-\big({\cal L}_{\Nb}\ov{Q}\big)_{rs}\Big)
\Big[[\ov{G}^{-1}]^{mnrs}\Big(\frac{\partial}{\partial\tau}-{\cal L}_{\Nb}\Big)
-\Big(\dot{\ov{Q}}_{tu}-\big({\cal
L}_{\Nb}Q\big)_{tu}\Big)\ov{Q}^{sn}[\ov{G}^{-1}]^{turm}\Big]\nonumber\\
&&
\quad
+\frac{1}{4}\Big(\dot{\ov{Q}}_{rs}-\big({\cal L}_{\Nb}Q\big)_{rs}\Big)
[\overline{G}^{-1}]^{rsjk}
\Big[\ov{Q}^{mt}\ov{N}^n[\ov{Q}_{jk}]_{,t}
+2\overline{Q}_{tk}\frac{\partial}{\partial x^j}\Big(\ov{Q}^{mt}\ov{N}^n\Big)\Big]
\Big]\delta Q_{mn}\nonumber\\
&&
+\Big[
-2\frac{\kappa\ov{N}^2}{2\Qb}\Nv{m}\Big(\ov{C}-\frac{\Qb}{\kappa}
\big(2\Lambda-\ov{R}\big)-\frac{\Qb}{2\lambda}\big(\ov{Q}^{rs}\ov{\Xi}_{,r}\ov{
\Xi}_{,s}+v(\Xi)\big)\Big)
-\frac{\kappa}{2\lambda}\big(\dot{\Xi}-\big({\cal
L}_{\Nb}\ov{\Xi}\big)\big)\ov{Q}^{mn}\ov{\Xi}_{,n}
\nonumber\\
&&
\quad
-\frac{1}{4}\Big(\dot{\ov{Q}}_{rs}-\big({\cal L}_{\Nb}Q\big)_{rs}\Big)
[\ov{G}^{-1}]^{rsjk}
\Big[\ov{Q}^{mt}[\ov{Q}_{jk}]_{,t}
+2\ov{Q}_{tk}\frac{\partial}{\partial x^j}\Big(\ov{Q}^{mt}\Big)\Big]
\Big]\delta
N_m\nonumber\\
&& +\Big[
\frac{\ov{N}^2\kappa}{2\lambda}\Big(\frac{1}{\ov{N}^2}\Big(\dot{\ov{\Xi}}-\big({
\cal
L}_{\Nb}\ov{\Xi}\big)\Big)\Big(\frac{\partial}{\partial\tau}-{\cal
L}_{\Nb}\Big) +\ov{Q}^{mn}\ov{\Xi}_{,n}\frac{\partial}{\partial x^m}
+\frac{1}{2}v^{\prime}(\ov{\Xi})\Big)\Big]\delta\Xi .\nonumber \ea
Going back to the second order equation of motion for $Q_{jk}$ shown
in equation (\ref{ResDDQ}), we remind the reader that the
perturbation of the first term on the right-hand side involves a
term that we had already calculated for the equation of motion of
$\delta\Xi$. For this reason, we presented in the main text only the
perturbation of the remaining term $(\dot{Q}_{jk} - ({\cal
L}_{\N}Q)_{jk})$, not the final result for the full first term. For
those interested in more detail, we display it here: \ba
 \lefteqn{\delta\Big(\Big[
\frac{\dot{N}}{N}-\frac{(\Q)^{\bf\dot{}}}{\Q}+\frac{N}{\Q}\Big({\cal
L}_{\vec{N}}\frac{\Q}{N}\Big)\Big]\big(\dot{Q}_{jk}-\big({\cal
L}_{\Nb}Q\big)_{jk}\big)\Big)}\\
&=&
\Big[\Big[
\frac{\dot{\ov{N}}}{\ov{N}}-\frac{(\Qb)^{\bf\dot{}}}{\Qb}+\frac{\ov{N}}{\Qb}
\Big({\cal
L}_{\Nb}\frac{\Qb}{\ov{N}}\Big)\Big]\Big(\ov{Q}^{rn}\ov{N}^n[\ov{Q}_{jk}]_{,r}
+\ov{Q}_{rk}\frac{\partial}{\partial
x^j}\Big(\ov{Q}^{rn}\ov{N}^m\Big)+\ov{Q}_{rj}\frac{\partial}{\partial
x^k}\Big(\ov{Q}^{rn}\ov{N}^m\Big)\nonumber\\
&&
\quad
-\Big(\dot{\ov{Q}}_{jk}-\big({\cal L}_{\Nb}\ov{Q}\big)_{jk}\Big)
\Big(\big(\frac{\partial}{\partial\tau}-{\cal L}_{\Nb}\big)\Big(\QN{m}{n}\Big)
+\frac{\ov{N}}{\Qb}\frac{\partial}{\partial
x^k}\big(\frac{\Qb}{\ov{N}}\ov{N}^m\ov{Q}^{rn}\big)\Big)\Big]\delta Q_{mn}
\nonumber\\
&&
+\Big[
-\Big[\frac{\dot{\ov{N}}}{\ov{N}}-\frac{(\Qb)^{\bf\dot{}}}{\Qb}+\frac{\ov{N}}{
\Qb}\Big({\cal
L}_{\Nb}\frac{\Qb}{\ov{N}}\Big)\Big]
\Big(\ov{Q}^{mn}[\ov{Q}_{jk}]_{,r}+\ov{Q}_{mk}\frac{\partial}{\partial
x^j}\Big(\ov{Q}^{mn}\Big)+\ov{Q}_{mj}\frac{\partial}{\partial
x^k}\Big(\ov{Q}^{mn}\Big)\nonumber\\
&&
\quad
+\Big(\dot{\ov{Q}}_{jk}-\big({\cal L}_{\Nb}\ov{Q}\big)_{jk}\Big)
\Big(\big(\frac{\partial}{\partial\tau}-{\cal L}_{\Nb}\big)\Big(\Nv{m}\Big)
+\frac{\ov{N}}{\Qb}\frac{\partial}{\partial
x^k}\big(\frac{\Qb}{\ov{N}}\ov{Q}^{mn}\big)\Big)
\Big]\delta N_{m}\nonumber\\
&& +\Big[
\Big[\frac{\dot{\ov{N}}}{\ov{N}}-\frac{(\Qb)^{\bf\dot{}}}{\Qb}+\frac{\ov{N}}{\Qb
}\Big({\cal
L}_{\Nb}\frac{\Qb}{\ov{N}}\Big)\Big]\Big(\frac{\partial}{\partial\tau}-{\cal
L}_{\Nb}\Big)\Big]\delta Q_{jk}.\nonumber
\ea
\vspace{2cm}
\end{appendix}

\section{Gauge Invariant Versus Gauge Fixed Formalism}
\label{sg}

In this appendix we investigate the question under which circumstances a 
manifestly gauge invariant formulation of a constrained system can be 
equivalently described by a gauge fixed version.\\
\\
We begin quite generally and 
consider a (finite dimensional) constrained Hamiltonian system with 
first class constraints
$C_I,\;I=1,..,m$ on a phase space with canonical pairs 
$(q^a,p_a),\;a=1,..,n;\;m\le n$. 
If there is a true, gauge invariant 
Hamiltonian 
$H$ (not constrained to vanish) 
enlarge the phase space by an additional canonical pair $(q^0,p_0)$ 
and additional first class constraint $C_0=p_0+H$. The reduced phase
space and dynamics of the enlarged system is equivalent to the original 
one, hence we consider without loss of generality a system with no true
Hamiltonian.

The canonical Hamiltonian of the system is a linear combination of 
constraints
\be \label{g.1}
H_{{\rm can}}=\lambda^I C_I
\ee
for some Lagrange multipliers $\lambda^I$ whose range specifies the 
amount of gauge freedom. A gauge fixing is defined by a set of 
gauge fixing functions $G_I$ with the property that the matrix
with entries $M_{IJ}:=\{C_I,G_J\}$ has everywhere (on phase space) non 
vanishing determinant\footnote{Ideally, the gauge $G_I=0$ should define 
a unique point in each gauge orbit.}. 
Notice that 
we allow for gauge fixing conditions 
that display an explicit time dependence. The conservation in time 
of the gauge fixing conditions
\be \label{g.2}
0=\frac{d}{dt}G_I=\frac{\partial}{\partial t}G_I+\{H_{{\rm can}},G_I\}  
=\frac{\partial}{\partial t}G_I+\lambda^J M_{JI}
\ee
uniquely fixes the Lagrange multipliers to be the following phase space
dependent functions 
\be \label{g.3}
\lambda^I=-\frac{\partial G_J}{\partial t}\; (M^{-1})^{JI}=:\lambda^I_0
\ee
By arbitrarily splitting the set of canonical pairs $(q^a,p_a)$ into two 
sets $(T^I,\pi_I),\;I=1,..,m$ and $(Q^A,P_A),\;A=1,..,n-m$ we can solve 
$C_I=G_I=0$ for 
\be \label{g.4}
C'_I=\pi_I+h'_I(Q,P)=0,\;\;
G'_I=T^I-\tau^I(Q,P)=0
\ee
for certian functions $h,\tau$ which generically will be explicitly
time dependent.
The variables $T,\pi$ are called the gauge degrees of freedom and 
$Q,P$ are called the true degrees of freedom (although typically neither 
of them is gauge invariant). 

The reduced Hamiltonian $H_{{\rm red}}(Q,P)$, if it exists, is supposed 
to generate the same equations of motion for $Q,P$ as the canonical 
Hamiltonian does, when the constraints and the gauge fixing conditions 
are satisfied and the Lagrange multipliers assume their fixed values
(\ref{g.3}), that is,
\be \label{g.5}    
\{H_{{\rm red}},f\}
=\{H_{{\rm can}},f\}_{C=G=\lambda-\lambda_0=0}
=[\lambda_0^I \{C_I,f\}]_{C=G=\lambda-\lambda_0=0}
\ee
for any function $f=f(Q,P)$. For general gauge fixing functions the
reduced Hamiltonian  
will not exist, the system of PDE's to which (\ref{g.5}) is 
equivalent to, will not be integrable. 

However, a so called coordinate gauge fixing condition $G^I=T^I-\tau^I$ 
with 
$\tau^I$ independent of the phase space always leads to a reduced 
Hamiltonian as follows: We can always (locally) write the constraints in 
the form (at least weakly)
\be \label{g.6}
C_I=M_{IJ}(\pi_I+h_I(T,Q,P))
\ee
Then, noticing that $M_{IJ}=\{C_I,T_J\}$, (\ref{g.5}) becomes 
\be \label{g.7}    
\{H_{{\rm red}},f\}
=[\lambda_0^I M_{IJ} \{h_I,f\}]_{C=G=\lambda-\lambda_0=0}
=[\dot{\tau}_I \{h_I,f\}]_{G=0}
=\{\dot{\tau}_I \tilde{h}_I,f\}
\ee
with $\tilde{h}_I=h_I(T=\tau,Q,P)$ and we used 
that $f$ only depends on $Q,P$. This displays the reduced Hamiltonian as
\be \label{g.8}
H_{{\rm red}}(Q,P;t)=\dot{\tau}_I(t) h_I(T=\tau(t),Q,P))
\ee
It will be explicitly time dependent unless $\dot{\tau}_I$ is time 
independent and $h_I$ is independent of $T$, that is, unless those 
constraints can be deparametrised for which $\dot{\tau}_I\not=0$.
Hence, deparametrisation is crucial for having a conserved 
Hamiltonian system.

On the other hand, let us consider the gauge invariant point of view.
The observables associated with $f(Q,P)$ are given by
\be \label{g.9}
O_f(\tau)=[\exp(\beta^I X_I)\cdot f]_{\beta=\tau-T}
\ee
where we have denoted the Abelian Hamiltonian vector fields
$X_I$ by $X_I:=\{\pi_I+h_I,.\}$. 
Consider a one parameter family of flows $t\mapsto \tau^I(t)$ then 
with $O_f(t):=O_f(\tau(t))$ we find 
\be \label{g.10}
\frac{d}{dt} O_f(t) =
\dot{\tau}^I(t)\sum_{n=0}^\infty\;\;\frac{\beta^{J_1}..\beta^{J_n}}{n!}\;\;
X_I X_{J_1} .. X_{J_n} \cdot f
\ee
On the other hand, consider $H_I(t):=O_{h_I}(\tau(t))$, then
\ba \label{g.11}
\{H_I(t),O_f(t)\} &=&
=O_{\{h_I,f\}^\ast}(\tau(t))
=O_{\{h_I,f\}}(\tau(t))
=O_{X_I\cdot f}(\tau(t))
\nonumber\\
&=&\dot{\tau}^I(t)\sum_{n=0}^\infty\;\;\frac{\beta^{J_1}..\beta^{J_n}}{n!}\;\;
X_I X_{J_1} .. X_{J_n} \cdot f
\ea
where in the second step we used that neither $h_I$ nor $f$ depend on 
$\pi_J$, in the third we used that $f$ does not depend on $T^J$ and in 
the last we used the commutativity of the $X_J$. Thus the physical 
Hamiltonian that drives the time evolution of the observables is simply
given by
\be \label{g.12}
H(t):=\dot{\tau}^I(t) h_I(\tau(t),O_Q(t),O_P(t))
\ee
This is exactly the same as (\ref{g.8}) under the identification 
$f\leftrightarrow O_f(0)$. Hence we have shown that for suitable gauge 
fixings the reduced and the gauge invariant frameworks are equivalent.
Notice that it was crucial in the derivation that $(T^I,\pi_I)$ and 
$(Q^A,P_A)$ are two sets of canonical pairs. If that would not be the 
case, then it would be unclear whether the time evolution of the 
observables has a canonical generator. 

The power of a manifestly gauge invariant framework lies therefore not
in the gauge invariance itself. Rather, it relies on whether the gauge
fixing can be achieved globally, whether it can be phrased in terms of 
separate canonical pairs, whether the observer clocks $T_I$
are such that reduced Hamiltonian system is conserved and whether they 
do display
the time evolution of observables as viewed by a realistic observer.
In particular, the reduced Hamiltonian by construction only depends on 
$(Q^A,P_A)$. In application to gravity, if one would not add an 
additional matter component such as the dust, then the true degrees 
of freedom would sit in two of the six canonical pairs 
$(q_{ab},p^{ab})$. However, notice that it is not obvious how to split
these pairs into gauge and true degrees of freedom in an at least 
spatially covariant way and moreover it is not possible to solve for 
four of the $p^{ab}$ algebraically because the spatial diffeomorphism 
constraint involves their derivatives. Hence the physical or 
reduced Hamiltonian would become non-local. Furthermore, if one uses 
gravitational degrees of freedom for reduction then it is clear that 
onbe does not get the full set of Einstein's equations as evolution 
equations which is something that one may want to keep. Finally, the 
reduced Hamiltonian will not reduce to the standard model Hamiltonian
in the flat space limit (i.e. with unit lapse) nor will it be 
necessarily positive. Of course, when adding matter like our dust,
then similar to the Higgs mechanism the four dust degrees of freedom 
get absorbed by the metric which develops four additional Goldstone 
modes. These modes should decouple and they do as we showed explicitly 
in this paper because of the existence of an infinite number of 
conserved charges, 
however, it is not granted to happen when adding 
arbitrary matter.\\ 
\\
We close this section by verifying that the reduced Hamiltonian for the 
Brown -- Kucha\v{r} dust model with the obvious choice for the gauge 
degrees of freedom indeed agrees with the physical Hamiltonian. As gauge 
fixing conditions we choose 
\be \label{g.13}
G(x)=T(x)-\tau(x;t),\;\;G^j(x)=S^j(x)-\sigma^j
\ee
whence $\tau^j(t,x)=\sigma^j(x)$ is not explicitly time dependent. 
The stability of (\ref{g.13}) with respect to the canonical Hamiltonian
\be \label{g.14}
H_{{\rm can}}=\int_{{\cal X}}\; d^3x\; \{n[c-\sqrt{P^2+q^{ab} c_a 
c_b}]+n^a[P T_{,a}+P_j 
S^j_{,a}+c_a]\}
\ee
fixes lapse and shift to be 
\be \label{g.15}
n_0=-\frac{\dot{\tau}\sqrt{P^2+q^{ab} c_a c_b}}{P},\;\; n^a_0=0
\ee
Hence for any function $f$ independent of the dust degrees of freedom
\ba \label{g.16}
&& \{H_{{\rm can}},f\}_{c^{{\rm tot}}=\vec{c}^{{\rm 
tot}}=n-n_0=\vec{n}-\vec{n}_0=0}
\nonumber\\
&=& \int_{{\cal X}} \; d^3x\; (n_0\; [\{c,f\}-\frac{1}{2\sqrt{P^2+q^{ab} 
c_a c_b}} \;\{q^{ab} c_a c_b,f\}])_{c^{{\rm tot}}=0}
\nonumber\\
&=& \int_{{\cal X}} \; d^3x\;\dot{\tau}\; \frac{c}{h}\; 
[\{c,f\}-\frac{1}{2c} \;\{q^{ab} c_a c_b,f\}]
\nonumber\\
&=& \int_{{\cal X}} \; d^3x\;\dot{\tau}\; \frac{1}{2h}\; 
[\{c^2,f\}-\;\{q^{ab} c_a c_b,f\}]
\nonumber\\
&=& \int_{{\cal X}} \; d^3x\;\dot{\tau}\; \{h,f\}]
\ea
where we used 
\be \label{g.17}
c^{{\rm tot}}=0\;\;\Leftrightarrow\;\;-P=h=\sqrt{c^2-q^{ab} c_a c_b}
\ee
Thus the reduced Hamiltonian for $\dot{\tau}=1$ equals the physical
Hamiltonian under the identification $q_{ab}\equiv Q_{jk},\;
p^{ab}\equiv P^{jk}$.

\end{document}